\documentclass[%
 jmp,
cp,  
 amsmath,amsthm,amssymb,
 reprint,%
onecolumn]{revtex4-2}

\usepackage{pgfplots}
\pgfplotsset{compat=1.18}

\usepackage{graphicx,overpic,tikz}
\usepackage{dcolumn}
\usepackage{bm}

\usepackage[utf8]{inputenc}
\usepackage[T1]{fontenc}
\usepackage{mathptmx} 

\usepackage{amsthm}
\usepackage{hyperref}

\usepackage{comment}

\newtheorem{theorem}{Theorem}[section]
\newtheorem{corollary}[theorem]{Corollary}
\newtheorem{lemma}[theorem]{Lemma}
\newtheorem{proposition}[theorem]{Proposition}
\newtheorem{remark}[theorem]{Remark}

\newtheorem{assumption}[theorem]{Assumption}
\newtheorem{example}[theorem]{Example}

\renewcommand{\Re}{\operatorname{Re}}
\renewcommand{\Im}{\operatorname{Im}}
\newcommand{\erf}{\operatorname{erf}}
\newcommand{\erfc}{\operatorname{erfc}}
\newcommand{\sgn}{\operatorname{sgn}}
\newcommand{\Tr}{\mathop{\mathrm{Tr}}}
\begin{document}

\title{Fluctuations in Various Regimes of Non-Hermiticity and a Holographic Principle}

\author{L. D. Molag}
 \email[Corresponding author: ]{lmolag@math.uc3m.es}
\affiliation{Faculty of Mathematics, Carlos III University of Madrid,\\
Avda. de la Universidad, 30. 28911 Leganés, Spain.}

\author{G. Akemann} 
 \email{akemann@physik.uni-bielefeld.de}
 \affiliation{
 Faculty of Physics, Bielefeld University,\\
 P.O. Box 100131, 
D-33501 Bielefeld, Germany 
}
\author{M. Duits}%
 \email{duits@kth.se}
\affiliation{
Department of Mathematics, Royal Institute of Technology (KTH),\\
SE10044 Stockholm, Sweden.}

%
%

\date{\today} 

\begin{abstract}
The variance of the number of particles in a set is an important quantity in understanding the statistics of non-interacting fermionic systems in low dimensions. An exact map of their ground state in a harmonic trap in one and two dimensions to the classical Gaussian unitary and complex Ginibre ensemble, respectively, allows to determine the counting statistics at finite and infinite system size. We will establish two new results in this setup. First, we uncover an interpolating central limit theorem between known results in one and two dimensions, for linear statistics of the elliptic Ginibre ensemble. We find an entire range of interpolating weak non-Hermiticity limits, given by a two-parameter family for the mesoscopic scaling regime. Second, we considerably generalize the proportionality between the number variance and the entanglement entropy between Fermions in a set $A$ and its complement in two dimensions. Previously known only for rotationally invariant sets and external potentials, we prove a holographic principle for general non-rotationally invariant sets and random normal matrices. It states that both number variance and entanglement entropy are proportional to the circumference of $A$.\\


\noindent\textbf{MSC class:}   41A60; 60B20; 30E15; 15A52.

\end{abstract}

\maketitle


\section{Introduction} \label{sec:1}
The applications of Random Matrix Theory nowadays cover many areas of physics, mathematics, and other sciences, and we refer to \cite{OUP} for an overview. It is quite common that such applications are heuristic in nature, in particular when compared with data. On the other hand, there are quite a few known examples for an exact map of a particular Hamiltonian or field theory to an ensemble of random matrices. The examples we will focus on here are systems of non-interacting Fermions in one or two dimensions in a trap, possibly subject to external fields. A fruitful attempt to study such processes was introduced by Macchi, who described them with what we now call determinantal point processes (DPPs) \cite{Macchi}. 
Due to their Fermionic nature, the wave function of such many-body systems is given by the Slater determinant of the one-particle wave functions. 
There are 
situations to be described below, when the Slater determinant of the ground state 
is proportional to a Vandermonde determinant. Thus, the squared amplitude or probability distribution 
yields the modulus squared of the Vandermonde, which is ubiquitous in the joint eigenvalue distribution of random matrix ensembles with unitary symmetry.

Let us be more precise. In \cite{Eisler,MMSV} in dimension $d=1$ an exact map was exploited between the ground state wave function of $n$ free Fermions in a harmonic trap and the eigenvalues of the Gaussian unitary ensemble (GUE) of random matrices to compute quantities in the quantum system that are not accessible with standard techniques from physics, such as the local density approximation. 
What makes this map even more interesting is that 
it was shown that the variance of the number of Fermions in an interval in the bulk of the spectrum is directly proportional to the entanglement entropy with respect to the complement of that interval \cite{CDM}. While the entropy is difficult to measure directly in experiments, the variance of the number of Fermions in a set is directly accessible.

This setup of free Fermions in a harmonic trap has been studied in higher dimensions $d>1$ too, and although there is no random matrix ensemble associated to it, techniques for determinantal point processes and asymptotic analysis developed in random matrix theory apply. We refer to \cite{DDMSrev} for a review, including possible realizations in experiments. 
This approach allows to analytically calculate various quantities for finite and large-$n$, such as linear statistics, in particular the variance and higher-order cumulants. The proportionality between entropy and variance continues to hold for $d>1$ \cite{Vicari,CMV}.
Furthermore, there is a generalization of such a map to other quantum Hamiltonians, with particular potentials including a hard wall, to classical random matrix ensembles of Wishart and Jacobi type, cf. \cite{DDMSrev} and references therein.
{A map similar of quantum Hamiltonians to non-Hermitian random matrices was uncovered in  \cite{LMS}.}
It relates the complex Ginibre ensemble, introduced by Ginibre in 1965 
\cite{Ginibre}, with independent complex Gaussian matrix elements, to the ground state of free Fermions in two dimensions in a rotating harmonic trap under certain conditions. This is equivalent to free electrons in a harmonic trap subject to a quadrupolar magnetic field, when filling only the lowest Landau levels. This map is a particular case of the Landau Hamiltonian on noninteracting electrons in a magnetic field. Filling higher-order Landau levels has been considered too
\cite{Shirai,KMS,SDMS,Demni}, as well as an extension to non-zero temperature \cite{KDMS}. 
Let us emphasize that this setup of electrons in two dimensions in a perpendicular magnetic field is different from the first example of electrons in $d=2$ in a harmonic trap.
In this map to the complex Ginibre ensemble, the proportionality between
Renyi entropy $\mathfrak S^q$, for $q>1$, and von Neumann entropy $\mathfrak S^1$ on the one side, and the number variance of points within a set $A$ on the other side has been established in certain cases \cite{LMS}. Namely, for a centered circular domain $A$, the respective number variance and entanglement entropy $\mathfrak S^q$ are proportional to the circumference $|\partial A|=2\pi a$ of the disc $A$ with radius $a$. This proportionality is universal in the sense that it continues to hold for higher-order, rotationally invariant random matrix potentials, given by normal random matrices \cite{LMS,AkByEb}, and includes the computation of higher-order cumulants.

Two questions immediately follow. First, does this relation for the entropy extend to more general sets $A$, and can we thus establish a so-called holographic principle $\mathfrak S^q\sim |\partial A|$? And second, is it possible to interpolate between the results in $d=1$ for the GUE and in $d=2$ for the complex Ginibre ensemble, i.e. by modifying the quadrupolar magnetic field, and how does the variance behave in such an interpolating regime?
There is a further motivation to study such an interpolating regime. In \cite{RiVi} a map was constructed between the logarithm of the characteristic polynomial in the complex Ginibre ensemble and the planar Gaussian free field, which enjoys a direct relation to statistical mechanics and conformal field theory. Here, smooth linear statistics and the existence of a central limit theorem are in the focus.

{
The purpose of this article is to answer these questions. Apart from random normal matrices we will employ the 
elliptic Ginibre ensemble interpolating between the GUE and the complex Ginibre ensemble. It was introduced by Girko and Sommers et al. \cite{Girko,SCSS} and is a one-parameter family depending on $\tau\in[0,1)$.  On the physics side  it describes a two-dimensional one-component plasma in a quadrupolar magnetic field related to $\tau$ \cite{FoJa}. Furthermore, and this is relevant to the current paper, for the elliptic Ginibre ensemble, a map to a quantum Hamiltonian exists too \cite{FoJa,PeterLog,Demni}.}

{We will first focus on fixed $0\leq\tau<1$ at strong non-Hermiticity. 
Our first main result presented in subsection \ref{subsec:holo} a holographic principle for the number variance (also called rough linear statistics). Here, we generalize previous results for random normal matrices on centred discs \cite{LMS,AkByEb} and for the Gaussian Ginibre ensemble on general sets $A$ with measurable boundary \cite{Lin,LeMaOC} to random normal matrices (Theorems \ref{thm:entropyConvex} and \ref{thm:holographic}), and to the  elliptic Ginibre ensemble (Theorem \ref{thm:strongRoughAnearEdge}) on such general sets. }

{
Turning to smooth 
linear statistics and central limit theorems (CLT) as initiated by Rider and Virág \cite{RiVi} and generalised in \cite{AmHeMa}, in subsection \ref{subsec:smooth} 
we discuss that 
our asymptotic  analysis of the correlation kernel 
can be used 
to compute the limiting variance for $C^1$ functions. By standard arguments, this allows us to reduce the regularity of the test functions in the CLT from,  for example, $C^\infty$ in \cite{AmHeMa} to $C^1$. This is 
in accordance with the regularity assumption used by  Rider and Virág \cite{RiVi} for the special case of the Ginibre ensemble. The proof also prepares for our subsection \ref{subsec:weaknH}. 
Generalizations of \cite{RiVi} based on predictions \cite{Peter} include random normal matrices \cite{AmHeMa2,AmHeMa}, non-rotationally invariant potentials \cite{AC}, non-Hermitian Wigner matrices \cite{CES,CES2,BCH} and the two-dimensional Coulomb gas at general inverse temperature $\beta$ \cite{LS,PeterVar,BBNY}. For a recent overview in 
one and two dimensions we refer to \cite{PeterVar} and references therein.
}

{In subsection \ref{subsec:weaknH} we present our second main result,  
generalizing the concept of weak non-Hermiticity regime introduced by Fyodorov, Khoruzhenko, and Sommers on a microscopic scale \cite{FKS,FKS2}.
Here, the non-Hermiticity parameter $\tau=1-\kappa n^{-\alpha}$ is rescaled, with fixed $\kappa>0$ as in \cite{FKS,FKS2}, in the vicinity of $\tau\to1$, the Hermitian limit. In addition, the scaling of the test function $f(n^\gamma z)$ is parametrized by a second parameter $\gamma>0$, where we will focus on the origin region. 
This offers the possibility to attain various limiting regimes, depending on a  
relation between $\alpha$ and $\gamma$ in a continuous range of parameters $0<\alpha,\gamma<1$.
We obtain an interpolation between the variance of the GUE and the Ginibre ensemble on a mesoscopic scale (Theorem \ref{thm:VarSgamma>delta}) and obtain a CLT for $\alpha=\gamma$ (Theorem \ref{thm:GFF}).
In previous works, weak non-Hermiticity was only considered in the microscopic limit, leading to two specific values of $\alpha=1$ and $\alpha=1/3$ for the bulk \cite{FKS,FKS2}, respectively edge scaling \cite{Bender} of the spectrum, with $\gamma=1$, respectively $\gamma=2/3$. 
These microscopic limits were shown to be universal \cite{ACV,ABender}. Recent works on this aspect of weak non-Hermiticity include \cite{AmBy}, and a related,  almost circular regime at weak non-unitarity \cite{BS}, choosing the same specific values for $\alpha=\gamma=1$ in the microscopic limit.
A weak non-Hermiticity regime where the horizontal and vertical part of the variables is scaled differently, was recently considered in \cite{CrFyWu}.
}
%


\subsection{Setup: {Determinantal point processes, linear statistics and elliptic Ginibre ensemble}}

Given a function $V:\mathbb C\to \mathbb R\cup\{\infty\}$ the associated \textit{random normal matrix} model consists of all complex $n\times n$ normal matrices $M$, equipped with the probability measure
\begin{align*}
    d\mathcal P_n(M) = \frac{1}{\mathcal Z_n} e^{-n\operatorname{Tr} V(M)} \prod_{1\leq i,j\leq n} d^2 M_{ij},
\end{align*}
where $\mathcal Z_n>0$ is the normalization constant. Here
\begin{align*}
    d^2M_{ij} = d\Re(M_{ij}) d\Im(M_{ij})
\end{align*}
denotes the standard area measure on $\mathbb C$. 
Usually, $V$ is referred to as a potential or an external field.  For the model to be well-defined one has to impose certain growth and regularity conditions on $V$. At the very least one assumes that
\begin{align} \label{eq:Vgrowth}
    \liminf_{|z|\to\infty} \frac{V(z)}{\log |z|}>2,
\end{align}
which assures us that $\mathcal Z_n$ exists. We shall also assume that $V$ is $C^2$ on $\mathbb C$ (or $C^1$ in a few exceptional cases). 
The corresponding eigenvalues $z_1, \ldots, z_n$ of $M$ are distributed as
\begin{align*}
    d\mathbb P_n(z_1, \ldots, z_n) = \frac{1}{Z_n} \prod_{1\leq i<j\leq n} |z_i-z_j|^2 \prod_{j=1}^n e^{-n V(z_j)} d^2 z_j,
\end{align*}
where $Z_n>0$ is the normalization constant. Random normal matrix models are "integrable" in the sense that the distribution of their eigenvalues can be expressed by a single function of two variables called the \textit{correlation kernel}. To be more accurate, the $k$-point correlation functions can be expressed as
\begin{align} \label{eq:corFuncCorKer}
    \rho_{n,k}(z_1, \ldots, z_k) :=& \frac{n!}{(n-k)!} \int_{\mathbb C^{n-k}} \frac{1}{Z_n} \prod_{1\leq i<j\leq n} |z_i-z_j|^2 \prod_{j=1}^n e^{-n V(z_j)} \, d^2 z_{k+1} \cdots d^2z_{n}=
    \det \big(\mathcal K_n(z_i, z_j) \big)_{1\leq i,j\leq k},
\end{align}
for $k=1,\ldots,n$, where $\mathcal K_n : \mathbb C^2\to \mathbb C$ is the correlation kernel. A point process with such a determinantal structure is called a \textit{determinantal point process} (DPP). 
The correlation kernel is not unique given a potential $V$. It can be transformed $\mathcal K_n(z_i, z_j)\to \mathcal K_n(z_i, z_j) f(z_i)/f(z_j)$ with $f(z)\neq0$ that may depend on $n$. The determinant in \eqref{eq:corFuncCorKer} and thus all $k$-point correlation functions are invariant under such a transformation. 
We will make the following symmetric choice:
\begin{align} \label{eq:defSymBergmanK}
    \mathcal K_n(z,w) = e^{-\frac{1}{2} n(V(z)+V(w))} \boldsymbol k_n(z,w)\ , \quad 
    \boldsymbol k_n(z,w)=\sum_{j=0}^{n-1} p_j(z) \overline{p_j(w)}\ ,
\end{align}
where we also defined the correlation kernel without the weight factors $\boldsymbol k_n(z,w)$, known as the (polynomial) Bergman kernel.
The $p_j:\mathbb C\to\mathbb C$ are planar orthogonal polynomials satisfying the orthogonality conditions
\begin{align*}
    \int_{\mathbb C^2} p_j(z) \overline{p_k(z)} e^{-n V(z)} d^2z = \delta_{j,k}, \qquad j, k=0,1,\ldots, n-1.
\end{align*}
The Bergman kernel is defined on the Hilbert space of complex polynomials $p$ of degree less than $n$ with respect to the norm 
\begin{align*}
    \int_{\mathbb C} |p(z)|^2 e^{-n V(z)} d^2z.
\end{align*}
It satisfies the reproducing property, i.e., for any complex polynomial $p$ of degree $<n$ we have that
\begin{align} \label{eq:repKernelProperty}
    p(z) e^{-\frac12 n V(z)} = \int_{\mathbb C}  \mathcal K_n(z,w) p(w) e^{- \frac12 n V(w)}\, d^2w.
\end{align}
The $1$-point correlation function gives the density of eigenvalues, and is 
expressed as
\begin{align*}
    \rho_{n,1}(z) = \frac{1}{n} \mathcal K_n(z,z)
    = \frac{1}{n} e^{-n V(z)} \sum_{j=0}^{n-1} |p_j(z)|^2.
\end{align*}
Under the conditions mentioned above, it is known that the average density of particles converges in the weak star sense to some measure $d\sigma_V$ with compact support $S_V$\cite{HeMa}. We call $S_V$ the \textit{droplet} associated to $V$. Furthermore, under the assumption that $V$ is $C^2$ it holds that $\Delta V(z)\geq 0$ on $S_V$, and we have explicitly that 
\begin{align*}
    d\sigma_V(z) = \frac{1}{4\pi} \Delta V(z) \, \mathfrak{1}_{S_V}(z) d^2z,
\end{align*}
where $\mathfrak{1}_{S_V}$ denotes the indicator function of the set $S_V$ and $\Delta V$ is the standard Laplacian 
of $V$
in terms of the variables $\Re z$ and $\Im z$ \cite{HeMa}. By $\mathring{S}_V$ we denote the interior of $S_V$, which is often referred to as the \textit{bulk}. The boundary $\partial S_V$ is commonly called the \textit{edge}. The droplet has a characterization in terms of the obstacle function $\check{V}$, which is defined as the maximal subharmonic function that satisfies both $\check{V}\leq V$ and $\check V(z)=2\log|z|+\mathcal O(1)$ as $|z|\to \infty$. Under the assumption that $V$ is $C^2$ and $\Delta V>0$ we have that the coincidence set $\check{V}=V$ equals $S_V$ up to a set of measure $0$. Moreover, the obstacle function has a potential theoretic interpretation. Defining the logarithmic potential
\begin{align*}
    U_V(z) = \int_{\mathbb C} \log \frac1{|z-w|} d\sigma_V(w),
\end{align*}
we have that
\begin{align*}
    \check V+U_V = c_V,
\end{align*}
where $c_V$ is a Robin-type constant. Furthermore, $V=\check V$ on $S_V$. 

In the present paper we focus on linear statistics of the eigenvalues of random normal matrices, and in particular the corresponding variance. Namely, we consider linear statistics
\begin{align} \label{eq:defSmoothSn}
\mathfrak X_n(f) = \sum_{j=1}^n f(z_j),
\end{align}
where $z_1, \ldots, z_n$ are the eigenvalues, and $f$ is a test function satisfying certain growth and regularity conditions. There is a direct relation between the correlation kernel and the variance, namely
\begin{align} \label{eq:varIntKer0}
\operatorname{Var} \mathfrak X_n(f) &= \int_{\mathbb C} f(z)^2 \mathcal K_n(z, z) d^2z - \iint_{\mathbb C^2} f(z) f(w) |\mathcal K_n(z, w)|^2 d^2z d^2w\\ \label{eq:varIntKer}
&= \frac{1}{2} \iint_{\mathbb C^2} (f(z)-f(w))^2 |\mathcal K_n(z, w)|^2 d^2z d^2w. 
\end{align}
Furthermore, as is easy to check with \eqref{eq:corFuncCorKer} and \eqref{eq:repKernelProperty}, the expectation value can be expressed as
\begin{align*}
    \mathbb E_n \mathfrak X_n(f) = \int_{\mathbb C} f(z) \mathcal K_n(z,z) \, d^2z.
\end{align*}
We shall consider two types of situations. The case where $f$ is smooth (or at least differentiable), this is known as \textit{smooth} linear statistics, and the case where $f=\mathfrak{1}_A$ is an indicator function of some set $A\subset\mathbb C$, which is known as \textit{rough} linear statistics. In the physics literature, the terminology \textit{counting statistics} appears to be preferred over rough linear statistics, and the corresponding variance of the counting statistic is called the \textit{number variance}.  

A special focus in our paper will be on the complex \textit{elliptic Ginibre ensemble}. It is a one-parameter deformation of the Ginibre ensemble \cite{Ginibre}, introduced independently by Girko \cite{Girko} in 1984 and Sommers et al. \cite{SCSS} in 1988, where we consider the potential
\begin{align} \label{eq:defVelliptic}
    V(z) = \frac{|z|^2-\frac{1}{2}\tau (z^2+\overline{z}^2)}{1-\tau^2}
    = \frac{(\Re z)^2}{1+\tau}
    +\frac{(\Im z)^2}{1-\tau}, \qquad 0\leq \tau<1.
\end{align}
It interpolates between the Ginibre ensemble ($\tau=0$) and the GUE ($\tau\uparrow 1$). Both ensembles were introduced without a normality condition on the matrices. However, we will be focused on the eigenvalue distribution, which turns out not to depend on this normality constraint. Thus once we have gone to an eigenvalue basis the Ginibre and elliptic Ginibre ensemble fit into the framework of random normal matrix models. In the case of the Ginibre ensemble, i.e., $\tau=0$, the planar orthogonal polynomials are scalar multiples of monomials. In this case the correlation kernel is given by
\begin{align*}
    \mathcal K_n(z, w) = \frac{n}{\pi} e^{-\frac12 n (V(z)+V(w))}
    \sum_{j=0}^{n-1} \frac{(n z \overline w)^j}{j!}.
\end{align*}
When $0<\tau<1$ the planar orthogonal polynomials are rescaled Hermite polynomials and the correlation kernel is explicitly given by \cite{FyKhSo} 
\begin{align} \label{eq:kernelHermitePols}
\mathcal K_n(z,w) &= \frac{n}{\pi\sqrt{1-\tau^2}} e^{-\frac{1}{2}n(V(z)+V(w))} \sum_{j=0}^{n-1} \frac{1}{j!} \left(\frac{\tau}{2}\right)^j H_j\left(\sqrt{\frac{n}{2 \tau}} z\right) H_j\left(\sqrt{\frac{n}{2 \tau}} \overline{w}\right), & z, w\in\mathbb C.
\end{align}
(While the authors in \cite{FyKhSo} use the probabilist's Hermite polynomials $He_n$, we use Hermite polynomials $H_n(x)$ orthogonal with respect to $\exp[-x^2]$ on $\mathbb{R}$.)
It is a known fact that the average density of points of the elliptic Ginibre ensemble (for fixed $\tau$) converges to a uniform measure on the elliptic disk in the limit $n\to\infty$, explicitly given by \cite{Girko,SCSS}
\begin{align}
S_V = \mathcal E_\tau = \left\{z\in\mathbb C : \left(\frac{\Re z}{1+\tau}\right)^2+\left(\frac{\Im z}{1-\tau}\right)^2\leq 1\right\}.
\end{align}
For the Ginibre ensemble one finds the (infinite) Ginibre kernel in the bulk in the limit $n\to\infty$. That is, for $z\in \mathring S_V$
\begin{align}\label{KGin}
    \lim_{n\to\infty} \mathcal K_n\left(z+\frac{u}{\sqrt n}, z+\frac{v}{\sqrt n}\right) = \frac{1}{\pi} 
    \exp\left(u \overline v-\frac{|u|^2+|v|^2}{2}\right).
\end{align}
On the other hand, on the edge one finds the Faddeeva plasma (or complementary error function) kernel \cite{FH99}. For $z\in \partial S_V$
\begin{align}\label{Kerfc}
    \lim_{n\to\infty} \mathcal K_n\left(z+\frac{z u}{\sqrt n}, z+\frac{z v}{\sqrt n}\right) = \frac{1}{2\pi} 
    \exp\left(u \overline v-\frac{|u|^2+|v|^2}{2}\right) \erfc\left(\frac{u+\overline v}{\sqrt 2}\right). 
\end{align}
The scaling limits are valid uniformly for $u, v\in\mathbb C$ in compact sets.  
These bulk and edge scaling limits are also valid for the elliptic Ginibre ensemble and are in fact universal for random normal matrix models with generic potentials \cite{AmHeMa3, AmHeMa2, AmKaMa, HeWe}.
These results hold when $\tau$ is fixed, this is called the \textit{strong} non-Hermiticity regime.  Originally Fyodorov, Khoruzhenko and Sommers introduced the \textit{weak} non-Hermiticity regime as having $\tau=1-\kappa n^{-1}$ \cite{FKS,FKS2}, being motivated by quantum systems with few open channels.
They found that the correlation kernel satisfies a non-trivial scaling limit in the vicinity of the real line, interpolating between the Ginibre-kernel \eqref{KGin} at strong non-Hermiticity and the sine-kernel of the GUE in the Hermitian limit.
Later, another weak non-Hermiticity regime was considered, where $\tau=1-\kappa n^{-1/3}$, which provided an edge scaling limit for the kernel\cite{Bender}. It interpolates between the Faddeeva-plasma or complementary error function kernel \eqref{Kerfc} at strong non-Hermiticity and the Airy-kernel of the GUE, cf. \cite{ABender}.

\subsection{The number variance and a holographic principle}\label{subsec:holo}

In the case of rough linear statistics the variance is called the number variance, quite literally because $\mathfrak X_n(\mathfrak{1}_A)$ counts the number of eigenvalues in $A$. For this reason, linear statistics are also known under the terminology counting statistics. The study of such counting statistics has seen a great interest lately \cite{LMS, AkByEb,FeLa, Charlier1, ChLe, AmChCrLe, AmChCrLe2, Lin, LeMaOC}.
For the Ginibre ensemble, corresponding to $V(z)=|z|^2$, it was recently proved in \cite{Lin} (and shortly thereafter in \cite{LeMaOC}) that the number variance is proportional to the perimeter of $A$ as $n\to\infty$ for the general family of all Caccioppoli sets $A\subset S_V$ (sets that have a measurable boundary). This agrees with \cite{LMS} and \cite{AkByEb}, although less apparent due to the radial symmetry that was considered there for $A$. Our first result is for the number variance of general random normal matrix models for such general sets $A$. In what follows $A \asymp B$ as $n\to\infty$ will mean that $A$ and $B$ are of the same order, i.e., there exist constants $0<c\leq C$ (independent of $n$) such that $c |A|\leq |B|\leq C|A|$ for $n$ large enough. 

\begin{theorem} \label{thm:entropyConvex}
    Consider a random normal matrix model with a $C^2$ potential $V$ that satisfies \eqref{eq:Vgrowth} and which is assumed to be real analytic in a neighborhood of $S_V$. Fix a compact set $K\subset \mathring S_V$ and assume that $\Delta V>0$ on $K$. Then we have
    \begin{align*}
        \operatorname{Var} \mathfrak X_n(\mathfrak{1}_A) \asymp \sqrt n |\partial A|
    \end{align*}
    as $n\to\infty$ for all convex sets $A\subset K$ with a $C^2$ boundary, where the implied constants depend only on $V$ and $K$.\\
    When $\Delta V$ is constant on $K$ we have for such sets $A$ that as $n\to\infty$ 
    \begin{align*}
        \operatorname{Var} \mathfrak X_n(\mathfrak{1}_A) = 
        \frac{\sqrt n}{4\pi^{3/2}}  |\partial A| \sqrt{\Delta V|_K} +\mathcal O(1).
    \end{align*} 
\end{theorem}

Similar results were proved in \cite{ChEs} under the more restrictive setting that $A$ has a $C^\infty$ boundary, without the convexity condition, although in Theorem \ref{thm:entropyConvex} we are more precise about the error term when $\Delta V$ is constant. We mention that a similar limiting behavior for the number variance was also proved for stationary point processes (to which the Ginibre ensemble asymptotically belongs)\cite{SoWeYa, KrYo}.  That the number variance grows slower than order $n$ implies that random normal matrix models are hyperuniform (since the order is $\sqrt n$ in fact class I hyperuniform) \cite{ToSt}. The second part of Theorem \ref{thm:entropyConvex} does not only apply to rescalings and translations of the Ginibre and elliptic Ginibre ensembles, but also to models with Hele-Shaw potentials (e.g., see \cite{AmTr} with $\beta=2$). This is because the proof also works if the region where $V=\infty$ does not intersect with $S_V$. In a follow-up paper involving one of the authors \cite{MaMoOC}, the result in Theorem \ref{thm:entropyConvex} has been considerably improved, and to the dominant order the number variance is seen te be given by $\sqrt n$ times an explicit integral over the boundary of $\partial A$. However, the results in \cite{MaMoOC} do not give accurate information on the subleading terms. We can be explicit about the constants in the first part of the theorem. From the proof in Section \ref{sec:NumberVar} it will be clear that 
\begin{align*}    \frac{(\displaystyle\min_K \Delta V)^2}{(\displaystyle \max_K \Delta V)^\frac{3}{2}}
    \leq 4\pi^\frac32 \lim_{n\to\infty} \frac{\operatorname{Var} \mathfrak X_n(\mathfrak{1}_A)}{\sqrt{n} \, |\partial A|} \leq \frac{(\displaystyle\max_K \Delta V)^2}{(\displaystyle \min_K \Delta V)^\frac{3}{2}}.
\end{align*} 
Since the variance tends to $\infty$ as $n\to\infty$, a general result on DPPs tells us that $\mathfrak X_n(\mathfrak{1}_A)$ converges in distribution to a normal distribution (after rescaling), i.e., it satisfies a central limit theorem \cite{Soshnikov}. For the Ginibre ensemble, the second part of Theorem \ref{thm:entropyConvex} was proved for Caccioppoli sets $A\subset S_V$, i.e., sets that have a measurable boundary \cite{Lin, LeMaOC}. The compact set $K\subset \mathring S_V$ in the formulation is to avoid the more challenging situation where $A$ is close to the boundary. However, there is some hope Theorem \ref{thm:entropyConvex} is valid uniformly for any set $A\subset S_V$ (i.e., no compact set $K$ is needed in the formulation). 
Namely, for the Ginibre ensemble it was shown in \cite{LMS} that an interesting scaling limit arises when one picks sets $A$ whose boundary is at a microscopic distance from the droplet boundary. In fact, this limiting behavior was shown to be universal for a large class of radial symmetric models and radial symmetric sets $A$ in \cite{LMS,AkByEb}. To the best of our knowledge, we provide the first result analogous to \cite{LMS, AkByEb} for a random matrix model without radial symmetry, i.e., the elliptic Ginibre ensemble. Namely, we provide the limiting number variance when $A$ is microscopic dilation of the droplet. In the follow-up paper \cite{MaMoOC} analogous results are proved for general (non-radial) potentials. 

\begin{theorem} \label{thm:strongRoughAnearEdge}
Consider the elliptic Ginibre ensemble with fixed $0\leq \tau<1$ and let $\vec{n}(z)$ denote the outward unit normal vector at $z$ on $\partial \mathcal E_\tau$. Define $A\subset \mathbb C$ as the 
following set, where we rescale with the local mean density $\sim\Delta V(z)$: 
\begin{align*}
A = A_n(S) = 
\begin{cases}
 \mathcal E_\tau\setminus \Big\{z+\frac{s}{\sqrt{2 n\Delta V(z)}} \vec{n}(z): z\in \partial \mathcal E_\tau, \, S\leq s\leq 0\Big\}, & S \leq 0,\\
 \mathcal E_\tau \cup \Big\{z+ \frac{s}{\sqrt{2 n\Delta V(z)}} \vec{n}(z) : z\in \partial \mathcal E_\tau, 0\leq s\leq S\Big\}, & S>0.
\end{cases}
\end{align*}
Then we have
\begin{align*}
\lim_{n\to\infty} \sqrt{\frac{1-\tau^2}{n}} \frac{\operatorname{Var} \mathfrak X_n(\mathfrak{1}_A)}{|\partial A|}  = \frac{1}{\pi\sqrt \pi} f(S), \qquad \text{ where } \quad f(S) = \sqrt{2\pi} \int_S^{\infty} \frac{\erfc(t) \erfc(-t)}{4} dt. 
\end{align*}
\end{theorem}

For the local statistics at the edge including a systematic discussion of the same rescaling we refer to \cite{AmKaMa}. 
At a macroscopic distance outside the bulk, the correlations are exponentially small (see \cite{AmCr}), and Theorem \ref{thm:entropyConvex} and Theorem \ref{thm:strongRoughAnearEdge} break down. For example, in the case of the Ginibre ensemble we have for $A=\{z\in\mathbb C : |z|>a\}$ with fixed $a>1$ that
\begin{align*}
    \operatorname{Var} \mathfrak{X}_n(\mathfrak{1}_A)
    = \sqrt\frac{a}{2(1+a^2)} \frac{a^3}{a^2-1} e^{-n(a^2-\log a-1)}  (1+\mathcal O(1/n))
\end{align*}
as $n\to\infty$, which can be shown with a straightforward calculation. 

Theorem \ref{thm:entropyConvex} has a consequence for the Rényi entropy. For given $q>1$, we define the following expression
\begin{align} \label{eq:defRenyiEntropy}
    \mathfrak S_n^q(A) = \frac{1}{1-q} \operatorname{Tr} \log(\mathbb A^q+(\mathbb I-\mathbb A)^q), 
\end{align}
where $\mathbb A$ is the \textit{overlap matrix} which is given by
\begin{align} \label{eq:overlapMatrix}
\mathbb A_{jk} = \int_A p_j(z) \overline{p_k(z)} e^{- n V(z)} d^2z, \qquad j,k = 0,\ldots,n-1.
\end{align}
For the choice $V(z)=|z|^2$ it is known that \eqref{eq:defRenyiEntropy} gives the Rényi entropy with parameter $q>1$ in the physical model of rotating trapped Fermions in two dimensions \cite{LMS}, cf. \cite{Vicari,CMV}. 
The Rényi entropy can alternatively be expressed in terms of the reduced density matrix, which is 
commonly taken as its definition. In the limit $q\downarrow 1$ we obtain the von Neumann entropy. The entropy is a measure of the information contained in $A$. The following result can be seen as a holographic principle for the eigenvalues, namely the information is contained in the boundary of $A$. 

\begin{theorem}[Holography] \label{thm:holographic}
    Let $q>1$. Consider a random normal matrix model with a $C^2$ potential $V$ that satisfies \eqref{eq:Vgrowth} and which is assumed to be real analytic in a neighborhood of $S_V$. Fix a compact set $K\subset \mathring S_V$ and assume that $\Delta V>0$ on $K$. Then
    \begin{align} \label{eq:holographicIneq}
        \mathfrak S_n^q(A) \asymp \sqrt n |\partial A|,
    \end{align}
    as $n\to\infty$ for all convex $A\subset K$ with $C^2$ boundary, where the implied constants depend only on $V, q$ and $K$.\\ 
When $\Delta V>0$ on $S_V$, these constants can be chosen independently of $K$.
\end{theorem}

The assumptions on $A$ are unlikely to be essential, and it is probably enough that $A$ has a measurable boundary. We mention that our result is related to the theory of Toeplitz operators and a relevant reference considering the edge case is \cite{Marceca}. We can be more explicit about the implied constants; see Proposition \ref{prop:VarIneqq} and also Remark \ref{rem:constqto1} below. We suspect that a more direct relation should exist between the entropy and the boundary of $A$ and intend to investigate this in future work. In \cite{LMS} it was shown that the entropy satisfies the following limiting formula for radial symmetric sets $A=\{z\in\mathbb C : |z|\leq a\}$ strictly contained in the bulk, i.e., $a$ is a fixed number $<1$. 
\begin{align}\label{eq:LMSlimit}
    \lim_{n\to\infty} \frac{\mathfrak S_n^q(A)}{\sqrt n |\partial A|}
    = \frac{1}{\pi\sqrt 2} \frac{1}{1-q} \int_{-\infty}^\infty \log\left(\frac{1}{2^q} \erfc(x)^q+\frac{1}{2^q} \erfc(-x)^q\right) \, dx.
\end{align}
On the other hand, when $a$ is microscopically close to $1$, i.e., $a=1+\frac{S}{\sqrt n}$, it was shown in \cite{LMS} that
\begin{align*}
    \lim_{n\to\infty} \frac{\mathfrak S_n^q(A)}{\sqrt n |\partial A|}
    = \frac{1}{\pi\sqrt 2} \frac{1}{1-q} \int_{S}^\infty \log\left(\frac{1}{2^q} \erfc(x)^q+\frac{1}{2^q} \erfc(-x)^q\right) \, dx.
\end{align*}
Thus, a version of Theorem \ref{thm:holographic} is still valid at microscopic distance outside the droplet, with $S$ in a compact set.
According to Theorem \ref{thm:holographic}, since the entropy is of the same order as the number variance, the holographic principle breaks down at macroscopic distances outside the droplet, where the number variance is exponentially small.

We suspect that, rather than an inequality as in \eqref{eq:holographicIneq}, a limit should hold for $\mathfrak S_n^q(A)$ more generally for subsets $A$ of the bulk (i.e., not only radial symmetric sets), although contrary to \eqref{eq:LMSlimit} the limit may depend on $\partial A$ in a more sophisticated way. As a side note, we should point out that it is unclear at this point whether \eqref{eq:defRenyiEntropy} has the physical (or information-theoretic) interpretation of entropy for \textit{every} choice of $V$. 

\subsection{Smooth linear statistics of the elliptic Ginibre ensemble}\label{subsec:smooth}

In the case of smooth linear statistics, it is known under certain general conditions on $V$ (see (A1)-(A4) in \cite{AmHeMa2, AmHeMa}) that the linear statistics converge to some normal distribution. Namely, for random normal matrix models we have
\begin{align} \label{eq:CLT_fixed_tau}
\mathfrak X_n(f) - \mathbb E_n \mathfrak X_n(f) \to N\left(0, \Sigma^2\right)
\end{align}
in distribution as $n\to\infty$, where $\Sigma^2$ is the limiting variance determined implicitly by the von Neumann jump operator. We thus have a central limit theorem (CLT) in this case. This result was shown by Ameur, Hedenmalm and Makarov in 2015 \cite{AmHeMa2}. A first result in this direction concerns the paper by Rider and Virág in 2007 \cite{RiVi}, who proved the corresponding result for the complex Ginibre ensemble. It corresponds to the case $V(z)=|z|^2$ for $z\in\mathbb C$, that is $\tau=0$ in \eqref{eq:defVelliptic}. In this case, the droplet equals the unit disc, $S_V=\mathbb D$. They showed that $\Sigma^2=\sigma^2+\tilde \sigma^2$, where with Sobolev norm notation
\begin{align*}
    \sigma^2 &= \frac{1}{4\pi} \lVert f\rVert_{H^1(\mathbb D)}^2 = \frac{1}{4\pi} \int_{\mathbb D} |\nabla f(z)|^2 d^2z,\\
    \tilde\sigma^2 &= \frac{1}{2} \lVert f\rVert_{H^{1/2}(\partial \mathbb D)}^2 = \frac{1}{8\pi^2} \int_{\partial \mathbb D} \int_{\partial \mathbb D} \left|\frac{f(z)-f(w)}{z-w}\right|^2 dz dw.
\end{align*}
Here $\nabla f$ denotes the gradient, that is
\begin{align*}
  \nabla f(z) = \left(\frac{\partial f(z)}{\partial \Re z}, \frac{\partial f(z)}{\partial \Im z}\right).  
\end{align*}
Their result holds for complex-valued functions $f$, in our paper we will consider $f$ to be real-valued throughout. Actually, in \cite{RiVi} the second part of the variance was expressed in terms of the Fourier coefficients of $f$, namely 
\begin{align*}
    \lVert f\rVert_{H^{1/2}(\partial \mathbb D)}^2 = \sum_{k=-\infty}^\infty |k| |\hat{f}(k)|^2, \qquad 
    \hat{f}(k) = \frac{1}{2\pi} \int_0^{2\pi} f(e^{i t}) e^{-i k t} dt,
\end{align*}
but the reader may verify that the two formulae are equivalent by parametrizing $\partial \mathbb D$ as $e^{i t}$ with $t\in[0,2\pi)$, and carrying out the integrations and using a power series expansion. Intuitively, $\sigma^2$ can be interpreted as a "bulk" ($\mathring S_V$) contribution, while $\tilde\sigma^2$ can be interpreted as an "edge" ($\partial S_V$) contribution. 
Our result in this section concerns an explicit formula for the limiting variance of smooth linear statistics of the elliptic Ginibre ensemble with fixed $\tau$. The growth condition on $f$ below is merely to ensure that the integral in \eqref{eq:varIntKer} converges, as will become clear in Section \ref{sec:sec1}. {Moreover, by a standard argument, this allows us to reduce the regularity assumption  to $C^1$, compared to $C^\infty$ in \cite{AmHeMa} or $C^{3,1}_c$ in \cite{LS}.} 

\begin{theorem}[Limiting variance for fixed $\tau$] \label{thm:CLTfixedTau}
Let $0<\tau<1$ be fixed, and let $f:\mathbb C\to\mathbb R$ be a differentiable function satisfying the growth condition
$|f(z)| \leq e^{C \log^2 |z|}$ as $|z|\to \infty$ for some constant $C>0$.
Then $\operatorname{Var} \mathfrak X_n(f) = \sigma^2+\tilde\sigma^2$ as $n\to\infty$,
where,
\begin{align*}
\sigma^2 &= \frac{1}{4\pi} \int_{\mathcal E_\tau} |\nabla f(z)|^2 d^2z,\\
\tilde\sigma^2 &= \frac{1}{8\pi^2} \int_{\partial \mathcal E_\tau} \int_{\partial \mathcal E_\tau} \left|\frac{f(z) - f(w)}{\psi(z) - \psi(w)}\right|^2 |\psi'(z) dz| |\psi'(w) dw|,
\end{align*}
and $\psi(z) = \frac{1}{2} (z + \sqrt{z^2-4\tau})$. { Consequently,  the CLT \eqref{eq:CLT_fixed_tau} holds for $f \in C^1$ with compact support. } 
\end{theorem}

The theorem is valid for any choice of square root in the definition of $\psi$. 
Note that $\psi$ can be extended to an analytic function $\psi:\mathbb C\setminus[-2\sqrt\tau, 2\sqrt\tau]\to \mathbb C$. We shall henceforth agree that $z\mapsto \sqrt{z^2-4\tau}$ is positive for large positive numbers. In that case $\psi$ restricted to the exterior of the droplet is a conformal map to the exterior of the unit disc $\mathbb D$. On the other hand, $\psi$ restricted to the interior of the droplet excluding $[-2\sqrt\tau, 2\sqrt\tau]$, is a conformal map to the annulus $\sqrt\tau<|z|<1$.

\begin{remark}
Note that we can write the second integral as
\begin{align*}
\int_{\partial \mathcal E_\tau} \int_{\partial \mathcal E_\tau} \left|\frac{f(z) - f(w)}{\psi(z) - \psi(w)}\right|^2 |\psi'(z) dz| |\psi'(w) dw|
=  \int_{|z|=1} \int_{|w|=1} \left|\frac{f(\phi(z)) - f(\phi(w))}{z - w}\right|^2 |dz| |dw|
= 4\pi^2\lVert f \circ \phi\rVert_{H^{1/2}(\partial \mathbb D)},
\end{align*}
where $\mathbb D$ is the unit disc and $\phi(z)=z+\tau z^{-1}$. We may thus write the limiting variance in terms of Sobolev norms as
\begin{align*}
\lim_{n\to\infty} \operatorname{Var} \mathfrak X_n(f) = \frac{1}{4\pi} \lVert f \circ \phi\rVert_{H^1(\mathbb D\setminus \sqrt\tau \mathbb D)} 
+ \frac{1}{2} \lVert f \circ \phi\rVert_{H^{1/2}(\partial \mathbb D)},
\end{align*}
which should be compared to Rider and Virág's result on the Ginibre ensemble, Theorem 1 in \cite{RiVi}. 
\end{remark}


\subsection{Mesoscopic fluctuations in {various regimes of weak non-Hermiticity}}\label{subsec:weaknH}

We now turn to the weak non-Hermiticity regime $\tau=1-\kappa n^{-\alpha}$ with $\alpha\in(0,1)$, and study the interpolation between  $\alpha=0$ and  $\alpha=1$. There are various interesting aspects that one would like to understand. We start by recalling that for $\alpha=0$ the microscopic scale is at $\sim n^{-1/2}$, whereas for $\alpha=1$ it is at $\sim n^{-1}$. What are the microscopic scales for the intermediate regimes $\alpha \in (0,1)$? Do the fluctuations on the mesoscopic scales (the scales between macroscopic and microscopic) for $\alpha \in (0,1)$ resemble those for the Ginibre ensemble or the GUE? Is there an interpolation and, if so, what does it look like? 

We approach these questions by introducing a second parameter $\gamma>0$  and concentrate our linear statistic around the origin at scale $n^{-\gamma}$. Specifically, denoting by $T_n^\gamma$ the operator $(T_n^\gamma f)(z) = f(n^\gamma z)$, we consider
\begin{align} \label{eq:defSmoothSnWeak}
\mathfrak X_n(T_n^\gamma f) = \sum_{j=1}^n f(n^\gamma z_j).
\end{align} 
We shall restrict our attention to smooth functions $f:\mathbb C\to \mathbb R$ that have compact support $S=\operatorname{supp} f$. Our results below are probably true under a proper condition  on the decay of $f$ at infinity, but our derivation is easier and cleaner with the requirement that $f$ has compact support. The variance of the linear statistics $\mathfrak X_n(T_n^\gamma f)$ is given by the formula
\begin{align} \nonumber
\text{Var } \mathfrak X_n(T_n^\gamma f)
&= \frac{1}{2} \int_{\mathbb C} \int_{\mathbb C} (f(u)-f(u'))^2 \left|\frac{1}{n^{2\gamma}}\mathcal K_n\left(\frac{u}{n^\gamma}, \frac{u'}{n^\gamma}\right)\right|^2 d^2u \, d^2u'\\ \label{eq:intVarGammaComp1}
&= \int_{S} f(u)^2 \frac{1}{n^{2\gamma}}\mathcal K_n\left(\frac{u}{n^\gamma}, \frac{u}{n^\gamma}\right) d^2u
- \iint_{S^2} f(u) f(u') \left|\frac{1}{n^{2\gamma}}\mathcal K_n\left(\frac{u}{n^\gamma}, \frac{u'}{n^\gamma}\right)\right|^2 d^2u \, d^2u',
\end{align} 
which is obtained after a substitution $(z, w)\to n^{-\gamma} (u, u')$  in \eqref{eq:varIntKer}. Hence, the behavior of the kernel near the origin is important. We have good control over the asymptotic behavior of the kernel, and this brings us to the following theorem on the limiting behavior of the variance of the linear statistics.
\begin{theorem}[Limiting variance at weak non-Hermiticity] \label{thm:VarSgamma>delta}
Assume that $0<\alpha<1$ and $\gamma, \kappa>0$. Let $f:\mathbb C\to \mathbb R$ be a differentiable function with compact support, and form the rescaled linear statistic $\mathfrak X_n(T_n^\gamma f) = \sum_{j=1}^n f(n^{\gamma} z_j)$, where the summation is over points from the elliptic Ginibre ensemble with parameter $\tau = 1 - \kappa n^{-\alpha}$. 
\begin{itemize}
\item[(i)] When $\alpha<\gamma$, we have
\begin{align}
\lim_{n\to\infty} \operatorname{Var } \mathfrak X_n(T_n^\gamma f) &= 
\begin{cases}
\displaystyle\frac{1}{4\pi} \int_{\mathbb C} \left|\nabla f(z)\right|^2 d^2z, & \gamma < \frac{1+\alpha}{2},\\
\displaystyle\frac{1}{8 \pi^2 \kappa^2} \iint_{\mathbb C^2} (f(z)-f(w))^2 e^{-\frac{|z-w|^2}{2\kappa}} d^2z \, d^2w, & \gamma = \frac{1+\alpha}{2},\\
0, & \gamma > \frac{1+\alpha}{2}. 
\end{cases}
\label{eq:ThmI.6i}
\end{align}
\item[(ii)] When $\alpha=\gamma$, we have
\begin{align} \label{eq:ThmI.6ii}
\lim_{n\to\infty} \operatorname{Var } \mathfrak X_n(T_n^\gamma f) &= \frac{1}{4\pi} \int_{|\Im(z)|\leq \kappa} |\nabla f(z)|^2 d^2z
+ \frac{1}{4\pi^2} \iint_{\Im(z)=\Im(w)\in \{-\kappa, \kappa\}} \left|\frac{f(z)-f(w)}{z-w}\right|^2  dz \, dw.
\end{align}
\item[(iii)] When $\alpha>\gamma$, we have
\begin{align}
\lim_{n\to\infty} \operatorname{Var } \mathfrak X_n(T_n^\gamma f) &= 
 \frac{1}{2 \pi^2} \int_{-\infty}^\infty \int_{-\infty}^\infty \left|\frac{f(x)-f(y)}{x-y}\right|^2 dx \, dy. 
 \label{eq:ThmI.6iii}
\end{align}
\end{itemize}
\end{theorem}
From \eqref{eq:ThmI.6i} we deduce that the microscopic scale is at $\gamma=(1+\alpha)/2$. For fixed $\alpha$, the limit of the variance for the smaller mesoscopic scales with $\alpha <\gamma < (1+\alpha)/2$ is as in the Ginibre ensemble, whereas the limit for the larger mesoscopic scales $\alpha>\gamma$ is as in the GUE. The transition takes place at $\alpha=\gamma$. Here, the variance  depends continuously on the scaling parameter $\kappa$ and interpolates between the two GUE and Ginibre regimes. Indeed,  by the fact that $f$ has compact support we see that, for $\kappa\to \infty$,  the right-hand side of \eqref{eq:ThmI.6ii} agrees with that of \eqref{eq:ThmI.6i} (with $\gamma<(1+\alpha)/2$). Furthermore, when $\kappa \downarrow 0$, it converges to the right-hand side of \eqref{eq:ThmI.6iii}. For a graphical representation of these different areas, see Fig. \ref{fig1}.

 A heuristic reasoning behind Theorem \ref{thm:VarSgamma>delta}, based on \eqref{eq:intVarGammaComp1}, is as follows. Let us denote $\xi_\tau=\log \frac{1}{\sqrt\tau}$. Under the substitution $(z, w)\mapsto n^{-\gamma} (u, u')$ the boundary of the elliptic droplet $\mathcal E_\tau$ corresponds to
\begin{align*}
2 n^\gamma \sqrt \tau \cosh(\xi_\tau+i \eta)
&= n^\gamma (1+\tau) \cos \eta + i n^\gamma (1-\tau) \sin \eta\\
&= (2 n^\gamma - \kappa n^{\gamma-\alpha}) \cos \eta + i \kappa n^{\gamma-\alpha} \sin \eta, \qquad \eta\in(-\pi,\pi]. 
\end{align*} 
Since we consider functions $f$ with compact support, in particular the real part of the support is bounded, and $\eta$ is effectively close to $\pm \frac{\pi}{2}$. The boundary of the droplet, in the coordinates $u, u'$, is thus effectively given by the horizontal lines with imaginary part $\pm i \kappa n^{\gamma-\alpha}$. Because $f$ has compact support, the edge behavior of the kernel thus does not play a role when $\alpha<\gamma$ and we observe Ginibre ensemble type fluctuations without boundary term, at least within the microscopic distance. When $\alpha>\gamma$ the droplet is essentially one-dimensional and no bulk term is expected, and we observe  GUE type fluctuations. In the case $\alpha=\gamma$ the bulk and edge behavior are on equal footing and this is reflected in the result.

\begin{figure}[t]
\begin{tikzpicture}
    \draw[->] (0,-1)--(0,6);
     \draw[->] (-1,0)--(6,0);
     \draw[-, dashed] (0,0)--(6,6);
      \draw[-, dashed] (0,2)--(6,5);
      \draw[-,dashed] (4,0)--(4,4);
       \draw[very thick] (0,0)--(4,4);
       \draw[very thick] (0,2)--(4,4);
       \filldraw[black] (0,5.7) node[anchor=east]{$\gamma$};
       \filldraw[black] (0,4) circle (1pt) node[anchor=east]{1};
        \filldraw[black] (0,0.2) node[anchor=east]{0};
        \filldraw[black] (0,2) circle (1pt) node[anchor=east]{$\frac 12$};
         \draw (4.6,1.0) node[anchor=south]{$\alpha=1$};
         \draw (6,0.0) node[anchor=south]{$\alpha$};
          \draw (5,5.7)  node{$\gamma=\alpha$};
           \draw (5.5,4.2)  node{$\gamma=\frac{1+\alpha}{2}$};
            \draw (3,1)  node{$\text{GUE}$};
             \draw (1,2)  node{$\text{GinUE}$};
\end{tikzpicture}
\caption{The upper full line $\gamma=\frac{1+\alpha}{2}$ represents the microcopic scale near the origin, meaning that when $(1-\tau)\sim n^{-\alpha}$ we find that the microscopic scale is at $\sim n^{-(1+\alpha)/2}$. The mesoscopic scales are then at $\sim n^{-\gamma}$ with  $0<\gamma<(1+\alpha)/2$. We obtain the variance for the mesosopic linear statistics of GUE type if $\alpha>\gamma$ in \eqref{eq:ThmI.6iii}
and of Ginibre (GinUE) type if $\alpha<\gamma$ in \eqref{eq:ThmI.6i}. On the lower full line at  $\alpha=\gamma$ we have a transition given in \eqref{eq:ThmI.6ii} that depends on the weak non-Hermiticity parameter $\kappa=(1-\tau)n^\alpha$, and we prove a CLT.}
\label{fig1}
\end{figure}
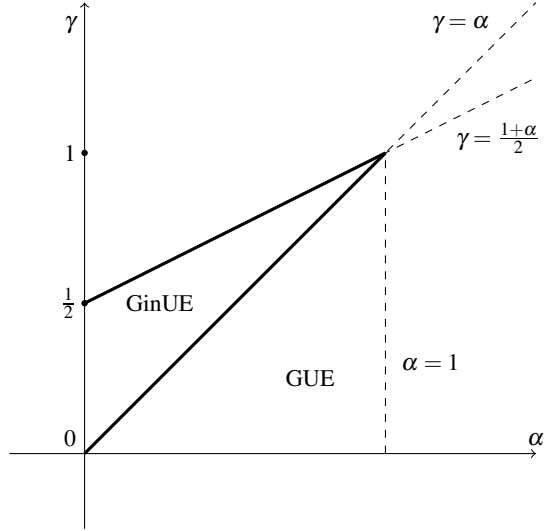

Now that we have computed the limiting behavior for the variance of the mesoscopic linear statistics in various regimes, the next question is about the limiting distribution of the fluctuations. Specifically, if there is a central limit theorem that interpolates between the one for the elliptic Ginibre Ensemble and the one for the GUE. We believe that the limiting mesoscopic fluctuations are Gaussian, but we only prove this here for the most interesting situation $\alpha=\gamma$. 

\begin{theorem} \label{thm:GFF}
    Consider the elliptic Ginibre ensemble with $\tau=1-\kappa n^{-\alpha}$ where $\kappa>0$ and $\alpha\in(0,1)$ and suppose that $f:\mathbb C\to\mathbb R$ is a smooth function with compact support. Then, as $n \to \infty$, the rescaled centered linear statistic \eqref{eq:defSmoothSnWeak} with $\alpha=\gamma$ satisfies
    $$
    \mathfrak X_n(T_n^\alpha f)-\mathbb E \left[\mathfrak X_n(T_n^\alpha f)\right] \to N(0,\sigma^2+\tilde\sigma^2),
    $$ 
    in distribution, where 
    \begin{align*}
        \sigma^2 &= \frac{1}{4\pi} \int_{|\Im(z)|\leq \kappa} |\nabla f(z)|^2 d^2z,\\
        \tilde\sigma^2 &=
        \frac{1}{4\pi^2} \iint_{\Im(z)=\Im(w)\in \{-\kappa, \kappa\}} \left|\frac{f(z)-f(w)}{z-w}\right|^2  dz \, dw.
    \end{align*}
\end{theorem}

As explained in Section 7.3 in \cite{AmHeMa2} (and \cite{AmHeMa}), Theorem \ref{thm:GFF} has an interpretation in terms of Gaussian free fields (in our case one that interpolates between dimension $1$ and dimension $2$), although we do not prove this connection rigorously. 
Theorem \ref{thm:GFF} is proved with an adaptation of the method of Ward identities, a method well-known in the physics literature; see, e.g. the papers by Wiegmann and Zabrodin \cite{WiZa, Za}. The method of Ward identities was first used in a mathematically rigorous setting by Johansson for the GUE class with general potentials \cite{Johansson}, and later for random normal matrices by Ameur, Hedenmalm and Makarov \cite{AmHeMa}. We mention that recently the method of Ward identities has been expanded to include disconnected droplets\cite{AmChCr, AmCr2}. 

{Our method differs in several ways}. For example, there is no limiting empirical measure with a two-dimensional support and there is no obstacle function in the usual sense. Nevertheless, there are associated objects that exist for any finite $n$ and, armed with explicit formulae, we are able to adapt the method of Ward identities.  

\subsection{Plan of the paper}

In Section \ref{sec:NumberVar} we prove the theorems associated to counting statistics, Theorem \ref{thm:entropyConvex}, Theorem \ref{thm:strongRoughAnearEdge} and Theorem \ref{thm:holographic}. First, we prove Theorem \ref{thm:entropyConvex} for potentials that have $\Delta V=1$ on some compact subset of the bulk. Our starting point will be the first line of \eqref{eq:varIntKer}, which for counting statistics reads
\begin{align*}
    \operatorname{Var} \mathfrak X_n(\mathfrak{1}_A)
    = \int_A \mathcal K_n(z,z) d^2z
    - \int_A \int_A |\mathcal K_n(z, w)|^2 d^2z \, d^2w.
\end{align*}
When $A$ is inside a compact subset $K\subset \mathring S_V$, the kernel is to a good approximation equal to the Ginibre kernel, and we may use this to find the limiting behavior of the number variance. A more refined approximation of the kernel with the so-called first-order approximating Bergman kernel allows us to treat the case where $\Delta V$ is not constant. 
The eigenvalues of the overlap matrix $\mathbb A$ as defined in \eqref{eq:overlapMatrix} are the same as those of the correlation kernel when multiplied by the indicator functions of $A$. This allows us to compare the entropy with the variance, as these can both be expressed as a sum over these eigenvalues. 

After Section \ref{sec:NumberVar}, we turn to smooth linear statistics. In Section \ref{sec:smoothVarElliptic} we find the asymptotic behavior of the correlation kernel in the weak non-Hermiticity regime 
for general values of $\alpha$ and $\gamma$. The approach is based on a single integral representation of the kernel, and a corresponding steepest descent analysis, which was introduced in \cite{ADM} and \cite{Mo}. In the setting of the weak non-Hermiticity regime, this approach is problematic at first glance, since the integrand has two saddle points that coalesce with each other, a pole and a square root type singularity. In the case $\alpha=1$ the limiting behavior is found with some identities for special functions, 
but for $0<\alpha<1$ this does not work. For this reason, we adapt the approach through a transformation of variables
and analyze a related single integral representation. There are generically 8 saddle points, but these are all explicit. 

Having obtained the asymptotic behavior of the kernel in various regions of $\mathbb C$ we may use the second line of \eqref{eq:varIntKer} to find the limiting variance. In Section \ref{sec:smoothVarElliptic} we find the limiting variance of smooth linear statistics of the elliptic Ginibre ensemble for fixed $\tau$ (strong non-Hermiticity). We prove Theorem \ref{thm:CLTfixedTau} in Section \ref{sec:sec1}. The approach will serve as a didactic appetiser for Theorem \ref{thm:VarSgamma>delta}, which we prove in Section \ref{sec:smoothVarWeak}. Finally, we use an adaptation of the method of Ward identities in Section \ref{sec:CLT} to prove Theorem \ref{thm:GFF}, making frequent use of the results of Section \ref{sec:smoothVarElliptic}.

\subsection*{Acknowledgments}
MD was supported by the European Research Council (ERC), Grant Agreement No. 101002013. 
This work
was partly funded by the German
Research Foundation 
DFG,  SFB 1283/2 2021 317210226 “Taming uncertainty
and profiting from randomness and low regularity in analysis, stochastics and
their applications” (GA, LM) and the Royal Society grant RF$\backslash$ERE$\backslash$210237 (LM). LM was partially supported by the UC3M grant 2024/00002/007/001/023 ``Local and global limits of complex-dimensional DPPs" and the grant ID2024-155133NB-I00,
``Orthogonality, Approximation, and Integrability:
Applications in Classical and Quantum Stochastic Processes (ORTH-CQ)" by the Agencia Estatal de Investigación.
LM thanks Andrew Ahn and Philippe Sosoe for their invitation to Cornell University to, among other things, discuss topics related to the Gaussian free field. LM also thanks Yacin Ameur for the invitation to the mini event on Coulomb gases at Lund University and the discussions about the Ward identities approach. LM thanks Matteo Levi, Jordi Marzo and Joaquim Ortega-Cerdà for discussions during a visit at the University of Barcelona. LM thanks Peter Forrester for spotting a typo in Theorem \ref{thm:CLTfixedTau} in an earlier version of this paper. GA thanks Ivan Parra for discussions and references to quantum Hamiltonians.
Last but not least, we would like to thank the
Mittag-Leffler-Institute for its hospitality, where joint work on this paper was made possible, with partial support
through the Swedish Research Council grant 2016-06596 during the program Random Matrices and Scaling Limits 2024.

\section{The number variance of random normal matrices} \label{sec:NumberVar}

Our aim is to prove Theorem \ref{thm:entropyConvex} in this section. The sets $A$ under consideration are assumed to be convex with a $C^2$ boundary, and bounded. Without loss of generality, we will assume that $A$ contains its boundary. Such sets can be written as
\begin{align} \label{eq:Avarphi}
    A &= \{z=x+i y\in\mathbb C : \varphi_-(x)\leq y\leq \varphi_+(x), \, a_1\leq x\leq a_2\}\\    \label{eq:Atildevarphi}
    &= \{z=x+i y\in\mathbb C : \tilde\varphi_-(y)\leq x\leq \tilde\varphi_+(y), \, b_1\leq y\leq b_2\}.
\end{align}
for some real numbers $a_1<a_2$, and $\varphi_-, \varphi_+ : [a_1, a_2]\to \mathbb R$ is $C^2$ and similarly for the second line (note that every line $x+i\mathbb R$ that has a nonempty intersection with $A$ is a nonempty interval by the convexity). Furthermore, since $A$ is convex, we have $\varphi_-''\geq 0 \geq \varphi_+''$. There can be at most one (nonempty) interval, possibly a point, where $\varphi_\pm'(x)=0$. In fact, the requirement that the boundary is $C^2$ means that these intervals (or points) exist, otherwise the boundary would not be smooth near one of the end points. 
We shall first consider the case where $\partial A$ does not have vertical or horizontal intervals of length larger than $0$. In other words, we consider sets that satisfy Assumption \ref{asump:AphiPoints} below (see Figure \ref{Fig2} for an example).

\begin{figure}
    \begin{tikzpicture}[scale=1.5]

    \fill[opacity=0.3, gray, smooth]
        (-1,1) .. controls (-0.9,1.5) and (0.5,2) .. (2,2)
        .. controls (2.5,2) and (3.2,1) .. (3,0)
        .. controls (2.5,-1) and (1,-2.5) .. (0,-2)
        .. controls (-0.8,-1.5) and (-1,0.5) .. (-1,1);

    \draw[->] (-2,0) -- (4,0) node[right] {$x$};
    \draw[->] (0,-3) -- (0,3) node[above] {$y$};

    \draw[thick, blue, smooth]
        (-1,1) .. controls (-0.9,1.5) and (0.5,2) .. (2,2)
        .. controls (2.5,2) and (3.2,1) .. (3,0)
        .. controls (2.5,-1) and (1,-2.5) .. (0,-2)
        .. controls (-0.8,-1.5) and (-1,0.5) .. (-1,1);

    \fill[black] (-0.99,1.07) circle (1pt) node[above left] {$(a_1, y_-)$};
    \fill[black] (2.05,2) circle (1pt) node[above right] {$(x_+, b_2)$};
    \fill[black] (3.03,0.3) circle (1pt) node[above right] {$(a_2, y_+)$};
    \fill[black] (0.4,-2.1) circle (1pt) node[below right] {$(x_-, b_1)$};
\end{tikzpicture}
\caption{Example of a set $A$ satisfying Assumption \ref{asump:AphiPoints}. The upper boundary curve is given by the graph of $\varphi_+$ while the lower boundary curve is given by the graph of $\varphi_-$, where $a_1\leq x\leq a_2$. Similarly, the right boundary curve is given by the graph of $\tilde\varphi_+$ while the left boundary curve is given by the graph of $\tilde\varphi_-$, where $b_1\leq y\leq b_2$. \label{Fig2}}
\end{figure}
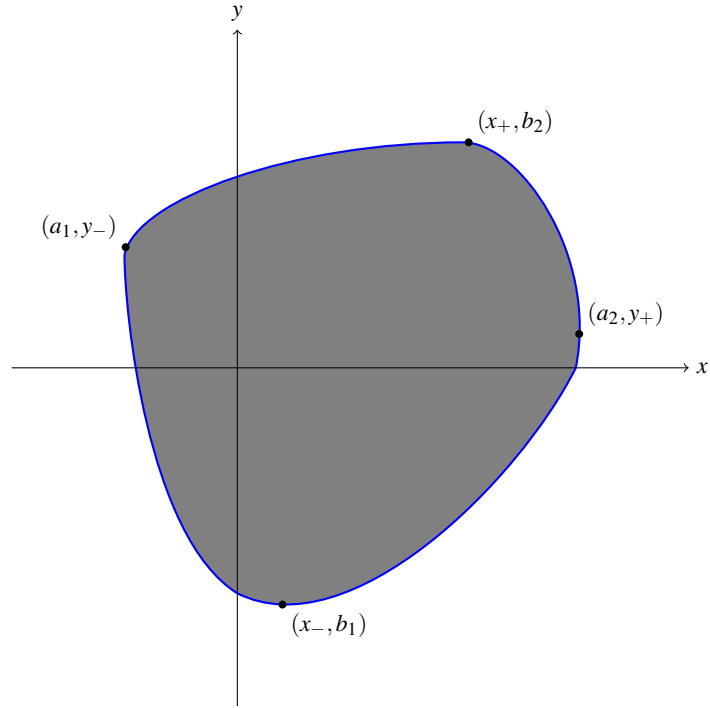

\begin{assumption} \label{asump:AphiPoints}
    We consider sets $A\subset \mathbb C$ which can be written both as \eqref{eq:Avarphi} and \eqref{eq:Atildevarphi} where $\varphi_\pm:(a_1,a_2)\to \mathbb R$ and where $\tilde\varphi_\pm:(b_1,b_2)\to \mathbb R$ are $C^2$ and satisfy $\varphi_+'', \tilde\varphi_+''\leq 0\leq \varphi_-'', \tilde\varphi_-''$. Here $a_1<a_2$ and $b_1<b_2$ are real numbers.\\
    Furthermore, we assume that there exist unique points $x_\pm\in (a_1,a_2)$ and $y_\pm\in(b_1, b_2)$ such that $\varphi_{\pm}'(x_\pm)=\tilde\varphi_{\pm}'(y_\pm)=0$. 
\end{assumption}

\subsection{The number variance for constant $\Delta V$}

We first prove the second part of Theorem \ref{thm:entropyConvex}. We may assume without loss of generality that $\Delta V=1$, after rescaling the integration variables. At this point we are not assuming that $A\subset S_V$. 

\begin{lemma}
Suppose that $A$ satisfies Assumption \ref{asump:AphiPoints}. Then we have
\begin{multline*}
\frac{n^2}{\pi^2}\int_A \int_A  e^{-n |z-w|^2} d^2z \, d^2w\\
= \frac{4 n}{\pi} \int_{a_1}^{a_2} \int_{b_1}^{b_2} \left(\erf(\sqrt n(x-\tilde\varphi_+(y))) - \erf(\sqrt n(x-\tilde\varphi_-(y)))\right)
\left(\erf(\sqrt n(y-\varphi_+(x)))-\erf(\sqrt n(y-\varphi_-(x)))\right) dy \, dx.
\end{multline*}
\end{lemma}

\begin{proof}
    This follows be carrying out the integrations over $y=\Im z$ and $x'=\Re w$, whose integration bounds are expressed in terms of $\varphi_\pm$ and $\tilde\varphi_\pm$.
\end{proof}

\begin{lemma} \label{lem:x+y+}
    Suppose $A$ satisfies Assumption \ref{asump:AphiPoints}. We have as $n\to\infty$ that
    \begin{multline} \label{eq:interferf}
        \int_{x_+}^{a_2} \int_{y_+}^{b_2} \left(\erf(\sqrt n(x-\tilde\varphi_+(y))) - \erf(\sqrt n(x-\tilde\varphi_-(y)))\right)
\left(\erf(\sqrt n(y-\varphi_+(x)))-\erf(\sqrt n(y-\varphi_-(x)))\right) dy dx\\
= 4\int_{x_+}^{a_2} \int_{y_+}^{b_2} \mathfrak{1}_A(z) dy dx - \frac{2}{\sqrt{\pi n}} \int_{x_+}^{a_2} \sqrt{1+\varphi_+'(x)^2} dx + \mathcal O\Big(\frac{1}{n}\Big). 
    \end{multline}
\end{lemma}

Before we prove Lemma \ref{lem:x+y+} it will be convenient to prove an auxiliary lemma. It should be thought of as Laplace's method, except with complementary error functions rather than exponentials.

\begin{lemma} \label{lem:ferfc}
    Let $a>0$ and suppose that $f:[0,a]\to \mathbb R$ is $C^k$ where $k$ is a non-negative integer. Then we have as $n\to\infty$ that
    \begin{align*}
        \int_0^a f(x) \erfc(\sqrt n x) dx
        = \sum_{j=0}^{k-1} \frac{f^{(j)}(0)}{\Gamma(\frac{j+1}{2}+1)} \frac1{(4n)^{\frac{j+1}{2}}}
        + \mathcal O(n^{-\frac{k+1}{2}} \max |f^{(k)}|)
        + \mathcal O\big(\frac{e^{-n a^2}}{\sqrt n} \max |f|\big),
    \end{align*}
    where the constants implied by the $\mathcal O$-terms are independent of $f$ and $n$.
    In particular, we have as $n\to\infty$ that
    \begin{align*}
       \int_0^a f(x) \erfc(\sqrt n x) dx
        = \frac{f(0)}{\sqrt{\pi n}}+\mathcal O\Big(\frac{1}{n}\Big)
    \end{align*}
    in the following two cases where $f$ may depend on $n$:
    \begin{itemize}
    \item[(i)] $f$ is $C^1$ and $\max |f'|=\mathcal O(1)$ as $n\to\infty$.
    \item[(ii)] $f$ is $C^2$, $f'(0)=\mathcal O(1)$ and $\max |f''|=\mathcal O(\sqrt n)$ as $n\to\infty$.
    \end{itemize}
\end{lemma}

\begin{proof}
    Let $F$ be the primitive for $f$ such that $F(0)=0$. An integration by parts shows that
    \begin{align*}
        \int_0^a f(x) \erfc(\sqrt n x) dx
        = F(a) \erfc(\sqrt n a)
        + 2 \sqrt\frac{n}{\pi}\int_0^a F(x) e^{-n x^2} dx. 
    \end{align*}
    By Taylor's theorem we have
    \begin{align*}
        F(x) = \sum_{j=1}^{k} \frac{f^{(j-1)}(0)}{j!} x^j + \mathcal O(x^{k+1} \max |f^{(k)}|)
    \end{align*}
    for $x\in[0,a]$. Then applying Laplace's method (or simply using the behavior of the incomplete gamma function) we get
    \begin{align*}
        2 \sqrt\frac{n}{\pi}\int_0^a F(x) e^{-n x^2} dx
        &= \frac1{\sqrt\pi} \sum_{j=1}^k \frac{\Gamma(\frac{j+1}{2})}{j!} f^{(j-1)}(0) n^{-\frac{j}{2}} + \mathcal O(n^{-\frac{k+1}{2}} \max |f^{(k)}|)
        + \mathcal O\Big(\frac{e^{-n a^2}}{\sqrt n}\Big)\\
        &= \sum_{j=1}^{k} \frac{f^{(j-1)}(0)}{2^j\Gamma(\frac{j}{2}+1)} n^{-\frac{j}{2}}
        + \mathcal O(n^{-\frac{k+1}{2}} \max |f^{(k)}|)
        + \mathcal O\Big(\frac{e^{-n a^2}}{\sqrt n}\Big),
    \end{align*}
    as $n\to\infty$, where we used the duplication formula for the gamma function in the last step.
\end{proof}

Note that in fact an asymptotic series with half-integer powers of $1/n$ is valid when $f$ is assumed to be smooth. 

\begin{proof}[Proof of Lemma \ref{lem:x+y+}]
    Since $x\geq x_+>\tilde\varphi_-(y)$ for all $x$ in the integration domain and $y\geq y_+>\varphi_-(x)$ for all $y$ in the integration the integral in the first line of \eqref{eq:interferf} is up to an exponential error equal to
    \begin{align*}
        \int_{x_+}^{a_2} \int_{y_+}^{b_2} \left(\erf(\sqrt n(x-\tilde\varphi_+(y))) - 1\right)
\left(\erf(\sqrt n(y-\varphi_+(x)))-1\right) dy dx.
    \end{align*}
    Performing a substitution $y\to \varphi_+(y)$, we get
    \begin{align*}
        -\int_{x_+}^{a_2} \int_{x_+}^{a_2} \left(\erf(\sqrt n(x-y)) - 1\right)
\left(\erf(\sqrt n(\varphi_+(y)-\varphi_+(x)))-1\right) \varphi'(y) dy dx.
    \end{align*}
    Inserting the identity $\erf x = \sgn(x) (1-\erfc |x|)$ this becomes
    \begin{multline*}
        - 4\int_{x_+}^{a_2} \int_{x_+}^y \varphi_+'(y) dx dy
        +2\int_{x_+}^{a_2} \int_{x_+}^y \erfc(\sqrt n(y-x)) \varphi_+'(y) dx dy
        + 2\int_{x_+}^{a_2} \int_{x_+}^y \erfc(\sqrt n(\varphi_+(x)-\varphi_+(y))) \phi_+'(y) dx dy\\
        - \int_{x_+}^{a_2} \int_{x_+}^{a_2} \erfc(\sqrt n|y-x|) 
        \erfc(\sqrt n |\varphi_+(x)-\varphi_+(y)|) \varphi_+'(y) dx dy.
    \end{multline*}
    We see that
    \begin{align*}
        - \int_{x_+}^{a_2} \int_{x_+}^y \varphi_+'(y) dx dy = - \int_{x_+}^{a_2} \int_{a_2}^{\varphi_+(y)} dx dy
        = \int_{x_+}^{a_2} \int_{y_+}^{b_2} \mathfrak{1}_A(z) dy dx.
    \end{align*}
    Carrying out the integration over $x$ and then applying Lemma \ref{lem:ferfc} for $k=1$ we get
    \begin{align*}
        \int_{x_+}^{a_2} \int_{x_+}^y \erfc(\sqrt n(y-x)) \varphi_+'(y) dx dy 
        &= \frac{y_+-b_2}{\sqrt{\pi n}} + \mathcal O(1/n).
    \end{align*}
    With similar arguments, after a substitution $(x,y)\to (\tilde\varphi_+(x), \tilde\varphi_+(y))$, we find that
    \begin{align*}
        \int_{x_+}^{a_2} \int_{x_+}^y \erfc(\sqrt n(\varphi_+(x)-\varphi_+(y))) \phi_+'(y) dx dy
        = \int_{y_+}^{b_2} \int_{y_+}^x \erfc(\sqrt n (x-y)) \tilde\varphi_+(x) dy dx
        = \frac{x_+ - a_2}{\sqrt{\pi n}}+\mathcal O(1/n).
    \end{align*}
    There is one more term left to analyse. We first consider the case where $x\in[x_+, y]$. Let $\delta_n=\sqrt\frac{\log n}{n}$ in what follows. Notice that
    \begin{multline*}
        \int_{x_+}^{a_2} \int_{x_+}^{y} \erfc(\sqrt n(y-x)) 
        \erfc(\sqrt n(\varphi_+(x)-\varphi_+(y))) \varphi_+'(y) dx dy\\
        = \int_{x_+}^{a_2} \int_{x_++\sqrt 2 \delta_n}^{y-\delta_n} \erfc(\sqrt n(y-x)) 
        \erfc(\sqrt n(\varphi_+(x)-\varphi_+(y))) \varphi_+'(y) dx dy\\
        + \int_{x_+}^{a_2} \int_{y-\delta_n}^y \erfc(\sqrt n(y-x)) 
        \erfc(\sqrt n(\varphi_+(x)-\varphi_+(y))) \varphi_+'(y) dx dy
        + \mathcal O(\delta_n^3).
    \end{multline*}
    Here the $\mathcal O$ term comes from the overcounting of a region (a triangle) of size of order $\delta_n^2$ where $\varphi_+'(y)$ is linearly close to $0$. The first term is $\mathcal O(1/n)$ by the asymptotic behavior of the complementary error function for large values. We can rewrite the middle term as
    \begin{align*}
        &\int_{x_+}^{a_2} \int_{y-\delta_n}^y \erfc(\sqrt n(y-x)) 
        \erfc(\sqrt n(\varphi_+(x)-\varphi_+(y))) \varphi_+'(y) dx dy\\
        &= \frac{1}{\sqrt n} \int_{x_+}^{a_2} \int_0^{\sqrt{\log n}} \erfc(x) \erfc\big(\sqrt n (\varphi_+(y-x/\sqrt n)-\varphi_+(y))\big) \varphi_+'(y) dx dy\\
        &= \frac{1}{\sqrt n} \int_{x_+}^{a_2} \int_0^{\sqrt{\log n}} \erfc(x) \erfc(|\varphi_+'(y)| x) +\mathcal O\Big(\frac{x^2\erfc x}{\sqrt n}\Big) dx dy\\
        &= \frac{1}{\sqrt n} \int_{x_+}^{a_2} \int_0^\infty \erfc(x) \erfc(|\varphi_+'(y)| x) \varphi_+'(y) dx dy + \mathcal O\Big(\frac1n\Big).
    \end{align*}
    There are various ways to show the identity
    \begin{align*}
\int_0^\infty \erfc(x) \erfc(t x) dx = 
\begin{cases}
\displaystyle\frac1{\sqrt\pi} + \frac{1-\sqrt{1+t^2}}{t\sqrt \pi}, &  t>0,\\
    \displaystyle \frac1{\sqrt\pi}, & t=0.
    \end{cases}
\end{align*}
(For example by differentiating with respect to $t$ and carrying out an adequate integration by parts with respect to $x$.)
Applying this to our present situation yields
\begin{align*}
    \frac{1}{\sqrt{n}}\int_{x_+}^{a_2} \int_0^\infty \erfc(x) \erfc(|\varphi_+'(y)| x) \varphi_+'(y) dx dy
    &= \frac1{\sqrt{\pi n}}\int_{x_+}^{a_2} \left(1-\frac1{\varphi_+'(y)}+\frac{\sqrt{1+\varphi_+'(y)^2}}{\varphi_+(y)}\right) \varphi_+'(y) dy\\
    &= \frac{y_+-b_2}{\sqrt{\pi n}} - \frac{a_2-x_+}{\sqrt{\pi n}} + \frac1{\sqrt{\pi n}} \int_{x_+}^{a_2} \sqrt{1+\varphi_+'(y)^2} dy.
\end{align*}
We conclude that
\begin{align*}
    \int_{x_+}^{a_2} \int_{x_+}^{y} \erfc(\sqrt n(y-x)) 
        \erfc(\sqrt n(\varphi_+(x)-\varphi_+(y))) \varphi_+'(y) dx dy
        = \frac{y_+-b_2}{\sqrt{\pi n}} - \frac{a_2-x_+}{\sqrt{\pi n}} + \frac1{\sqrt{\pi n}} \int_{x_+}^{a_2} \sqrt{1+\varphi_+'(y)^2} dy+\mathcal O(1).
\end{align*}
Analogously we find that
\begin{align*}
    \int_{x_+}^{a_2} \int_{y}^{a_2} \erfc(\sqrt n(y-x)) 
        \erfc(\sqrt n(\varphi_+(x)-\varphi_+(y))) \varphi_+'(y) dx dy
        &= \frac{y_+-b_2}{\sqrt{\pi n}} - \frac{a_2-x_+}{\sqrt{\pi n}} + \frac1{\sqrt{\pi n}} \int_{x_+}^{a_2} \sqrt{1+\varphi_+'(y)^2} dy+\mathcal O(1).
\end{align*}
Combining  this with the previous contributions yields the result.
\end{proof}

\begin{lemma} \label{lem:x-x+y-y+}
    Suppose $A$ satisfies Assumption \ref{asump:AphiPoints}. We have that
    \begin{align*}
        \int_{x_-}^{x_+} \int_{y_-}^{y_+} \left(\erf(\sqrt n(x-\tilde\varphi_+(y))) - \erf(\sqrt n(x-\tilde\varphi_-(y)))\right)
\left(\erf(\sqrt n(y-\varphi_+(x)))-\erf(\sqrt n(y-\varphi_-(x)))\right) dy dx\\
= 4 (x_+-x_-) (y_+-y_-) + \mathcal O(e^{-c n})
    \end{align*}
    as $n\to\infty$, for some constant $c>0$.
\end{lemma}

\begin{proof}
    Without loss of generality we assume that $x_+>x_-$ and $y_+>y_-$. Notice then that $\tilde\varphi_-(y)<x_-\leq x\leq x_+<\tilde\phi_+(y)$ for all $x$ in the integration. Similarly, we have
    $\varphi_-(x)<y_-\leq y\leq y_+<\varphi_+(x)$ for all $y$ in the integration. Thus up to an exponential error, the integral equals
    \begin{align*}
        \int_{x_-}^{x_+} \int_{y_-}^{y_+} (\sgn(x-\tilde\varphi_+(y))-\sgn(x-\tilde\varphi_-(y))) (\sgn(y-\varphi_+(x))-\sgn(y-\varphi_-(x))) dy dx
        = \int_{x_-}^{x_+} \int_{y_-}^{y_+} 4 \, dy \, dx.
    \end{align*}
\end{proof}

\begin{lemma}
    Suppose $A$ satisfies Assumption \ref{asump:AphiPoints}. We have as $n\to\infty$ that
    \begin{align*}
        \frac{n}{\pi} \int_A d^2z - \frac{n^2}{\pi^2} \int_A \int_A e^{-n |z-w|^2} d^2z \, d^2w = \frac{1}{4\pi} \sqrt\frac{n}{\pi} |\partial A| + \mathcal O(1).
    \end{align*}    
\end{lemma}

\begin{proof}
    For the regions $[x_+, a_2]\times [b_1, y_-]$, $[a_1, x_-]\times [y_+, b_2]$ and $[a_1, x_-]\times [b_1, y_-]$ a statement analogous to Lemma \ref{lem:x+y+} holds. Adding the contributions of these three regions together with the contribution of $[x_+, a_2]\times [y_+, b_2]$ yields that
    \begin{multline*}
        \frac{n^2}{\pi^2} \int_A \int_A e^{-n |z-w|^2} d^2z \, d^2w\\
        +\frac{4n}{\pi}\int_{x_-}^{x_+} \int_{y_-}^{y_+} \left(\erf(\sqrt n(x-\tilde\varphi_+(y))) - \erf(\sqrt n(x-\tilde\varphi_-(y)))\right)
\left(\erf(\sqrt n(y-\varphi_+(x)))-\erf(\sqrt n(y-\varphi_-(x)))\right) dy dx\\
        = \frac{n}{\pi} \int_A d^2z+ \frac{4n}{\pi}|x_+-x_-||y_+-y_-| - \frac{1}{4\pi} \sqrt\frac{n}{\pi} |\partial A|.
    \end{multline*}
    The result now follows after applying Lemma \ref{lem:x-x+y-y+}.
\end{proof}

\begin{proposition} \label{prop:AconvexInt}
    Assume that $A$ is a convex set with a $C^2$ boundary. Then we have as $n\to\infty$ that
    \begin{align*}
        \frac{n}{\pi} \int_A d^2z - \frac{n^2}{\pi^2} \int_A \int_A e^{-n |z-w|^2} d^2z \, d^2w = \frac{1}{4\pi} \sqrt\frac{n}{\pi} |\partial A| + \mathcal O(1).
    \end{align*} 
\end{proposition}

\begin{proof}
    In this case Assumption \ref{asump:AphiPoints} does not necessarily hold, i.e., there might be horizontal or vertical parts in $\partial A$. Define
    \begin{align*}
        x_{\pm,1} &= \inf\{x\in[a_1,a_2] : \pm\varphi_\pm'(x)\geq 0\},\\
        x_{\pm,2} &= \sup\{x\in[a_1,a_2] : \pm\varphi_\pm'(x)\geq 0\},\\
        y_{\pm,1} &= \inf\{y\in[a_1,a_2] : \pm\tilde\varphi_\pm'(x)\geq 0\},\\
        y_{\pm,2} &= \sup\{y\in[a_1,a_2] : \pm\varphi_\pm'(y)\geq 0\}.
    \end{align*}
    The integration over a region such as $x>x_{+,2}$ and $y>y_{+,2}$, where $\varphi_+$ is injective, is entirely analogous as in the proof of Lemma \ref{lem:x+y+}. Regions as in Lemma \ref{lem:x-x+y-y+} are analogously shown to only contribute an area times $\mathcal O(n)$. The only regions we have not treated yet are those where $\varphi_\pm$ or $\tilde\varphi_\pm$ are vertical on some neighborhood. Let us assume that $x_{+,1}<x_{+,2}$, but $x_{-,1}=x_{-,2}=x_-$ and $y_{\pm,1}=y_{\pm,2}=y_{\pm}$.  Then we have to integrate over the region $[x_{+,1}, x_{+,2}]\times [y_{+,2},b_2]$. Inserting the relation $\erf x = \sgn(x)(1-\erfc |x|)$ gives
    \begin{multline*}
        \int_{x_{+,1}}^{x_{+,2}} \int_{y_{+,2}}^{b_2} \left(\erf(\sqrt n(x-\tilde\varphi_+(y))) - \erf(\sqrt n(x-\tilde\varphi_-(y)))\right)
\left(\erf(\sqrt n(y-\varphi_+(x)))-\erf(\sqrt n(y-\varphi_-(x)))\right) dy dx\\
= \int_{x_{+,1}}^{x_{+,2}} \int_{y_{+,2}}^{b_2} \big(2-\erfc(\sqrt n (x-\tilde\varphi_+(y)))-\erfc(\sqrt n (x-\tilde\varphi_-(y)))\big)
\big(2-\erfc(\sqrt n (y-\varphi_+(x)))-\erfc(\sqrt n (y-\tilde\varphi_-(x)))\big) dx dy
    \end{multline*}
    Using Lemma \ref{lem:ferfc} with $k=2$ multiple times we infer that this equals
    \begin{align*}
        4 (x_{+,2}-x_{+,1}) (b_2-y_{+,2}) - 2\frac{x_{+,2}-x_{+,1}}{\sqrt{\pi n}} + \mathcal O\Big(\frac1n\Big)
    \end{align*}
    as $n\to\infty$. It thus produces the required area and boundary contribution. We omit the details of other configurations where $x_{\pm,1}\neq x_{\pm,2}$ or $y_{\pm,1}\neq y_{\pm,2}$ which are treated with similar arguments and may contain more rectangles to be integrated over.
\end{proof}

We mention that Proposition \ref{prop:AconvexInt} also follows from Theorem 2.8 in \cite{LeMaOC}, if we only consider the dominant $\sqrt n$ order. 

\begin{example}
When we consider the square $A = [-a,a]^2$ and the elliptic Ginibre ensemble, the Gaussian integrations can be done directly, and we get
\begin{align*}
\int_A \int_A \left(\frac{n}{\pi (1-\tau^2)}\right)^2 e^{-n \frac{|z-w|^2}{1-\tau^2}} d^2z d^2w
&= \frac{1}{\pi} \left(2 a \sqrt{\frac{n}{1-\tau^2}} \erf\left(2 a \sqrt{\frac{n}{1-\tau^2}}\right) - \frac{1-e^{-n \frac{4 a^2}{1-\tau^2}}}{\sqrt\pi}\right)^2\\
&= \frac{n (2a)^2}{\pi(1-\tau^2)} - \frac{4 a}{\sqrt{\pi^3(1-\tau^2)}} \sqrt n + \mathcal O(1)
\end{align*}
as $n\to\infty$. This implies that
\begin{align*}
\operatorname{Var} \mathfrak X_n(\mathfrak{1}_A) = \frac{\sqrt n}{2\sqrt{\pi^3(1-\tau^2)}} 8a + \mathcal O(1)
\end{align*}
as $n\to\infty$, for any fixed $a\in [0,1-\tau)$. This agrees with Proposition \ref{prop:AconvexInt}.
\end{example}


Proposition \ref{prop:AconvexInt} is enough to prove the second part of Theorem \ref{thm:entropyConvex}, after a rescaling $(z,w)\to \Delta V(z)^{-1} (z,w)$. 

\subsection{The number variance for general potentials}

Our next step is to prove the first part of Theorem \ref{thm:entropyConvex}, i.e., where we have general potentials $V$. 
We assume that $V$ is real analytic in a neighborhood $\mathcal N$ of $S_V$. This means that for every $z_0 \in \mathcal N$ we have an absolutely convergent expansion
\begin{align*}
    V(z) = \sum_{k, \ell=0}^\infty c_{k,\ell} (z-z_0)^k (\overline{z-z_0})^\ell,
\end{align*}
where $c_{k,\ell}$ are complex coefficients. These are explicitly given by
\begin{align*}
    c_{k,\ell} = \frac1{k!} \frac1{\ell!} \partial^k \overline \partial^\ell V(z_0).
\end{align*}
Here, when the variable is $z=x+iy$, the notation means
\begin{align*}
\partial = \frac12\left(\frac{\partial}{\partial x}-i \frac{\partial}{\partial y}\right)\text{ and }\overline\partial = \frac12\left(\frac{\partial}{\partial x}+i \frac{\partial}{\partial y}\right).
\end{align*}
Under this assumption $V$ admits a polarization, i.e., a function $V(z,w)$ analytic  in both its variables such that
\begin{align*}
    V(z, \overline{z}) = V(z),
\end{align*}
defined on some (small enough) neighborhood of the set $\{(z,\overline z) : z\in S_V\}$. Explicitly, we have
\begin{align*}
    V(z, w) = \sum_{k, \ell=0}^\infty \frac1{k!} \frac1{\ell!} \partial^k \overline \partial^\ell V(z_0) (z-z_0)^k (w-\overline{z_0})^\ell
\end{align*}
around any $z_0\in\mathcal N$. 
If we pick this neighborhood small enough, then we have the following expansion as well. 
\begin{align} \label{eq:VexpansionInw}
    V(z,\overline{w}) = \sum_{k=0}^\infty \frac{1}{k!} (z-w)^k \partial^k V(w).
\end{align}
We use the polarization to define an approximation of the kernel. The (weighted) first order approximating Bergman kernel is
    \begin{align} \label{eq:BergmanApprox1stO}
        \mathcal K_n^1(z,w) = \frac1{\pi}\big(n B_0(z,\overline{w})+B_1(z, \overline{w})\big) e^{n (V(z, \overline{w})-\frac{1}{2}V(z)-\frac12 V(w))}, 
    \end{align}
    see \cite{Berman,AmHeMa3} and references therein, where
    \begin{align*}
        B_0(z,w) = \partial_1\partial_2 V(z,w),
        \qquad B_1(z,w)= \frac12 \partial_1\partial_2 \log \partial_1\partial_2 V(z,w).
    \end{align*}
    Here $\partial_1$ is the same differential operator as before for the variable $z$, while $\partial_2$ is the one for the variable $w$.
    Let us now assume that the conditions of the first part of Theorem \ref{thm:entropyConvex} hold. Note that $B_1$ is well-defined due to the condition $\Delta V(z)>0$ which implies that $\Re\partial_1\partial_2 V(z,w)>0$ for $|z-\overline w|<\varepsilon$, when $\varepsilon$ is chosen small enough. Let $\mathcal K_n$ be as defined in \eqref{eq:defSymBergmanK}. By Theorem 2.8 in \cite{AmHeMa3} (and in greater generality in \cite{Berman}) there exists a possibly smaller $\varepsilon>0$ (independent of $A$) and for all $z, w\in K$ with $|z-w|<\varepsilon$ we have
    \begin{align} \label{eq:firstAppBergmanK}
        |\mathcal K_n(z,w)-\mathcal K_n^1(z,w)| \leq \frac{C_V^K}{n}
    \end{align}
    for $n$ big enough, where $C_V^K>0$ does not depend on $A$. 
    In what follows, let us denote
\begin{align*}
    \delta_n = \frac{\log n}{\sqrt n}
\end{align*}
    
    \begin{lemma} \label{lem:2V-V-V}
    Let $K\subset S_V$ be compact. We have uniformly for $w\in K$ and $|z-w|$ small enough that
    \begin{align} \nonumber
        2 V(z, \overline{w}) - V(z) - V(w)
        &= - |z-w|^2 \Delta V(w) - |z-w|^2 \Re\big((z-w) \partial \Delta V(w)\big)\\ \label{eq:Vexpansion}
        &\qquad + 2 i \Im\big((z-w) \partial V(w)\big)
        + i \Im(\partial^2 V(w) (z-w)^2)
        + \frac{1}{3} i\Im\big((z-w)^3 \partial^3 V(w)\big)
        + \mathcal O(|z-w|^4).
    \end{align} 
    In particular, we have
    \begin{align} \label{eq:Vexpansion2}
        \left|e^{n(V(z, \,\overline{w}) - \frac12 V(z) - \frac12 V(w))}\right|^2
        = e^{- n \Delta V(w) |z-w|^2} \Big(1- n  \, |z-w|^2 \Re((z-w) \partial \Delta V(w)) + \mathcal O\big(\frac{\log^6 n}{n}\big)\Big)
    \end{align}
    as $n\to\infty$, uniformly for $w\in K$ and $|z-w|<\delta_n$.  
\end{lemma}

\begin{proof}
    For $|z-w|$ small enough we have by Taylor expansion that
    \begin{align*} 
        V(z) = \sum_{k,\ell=0}^\infty \frac1{k! \ell!} \partial^k \overline \partial^\ell V(w) (z-w)^k \overline{(z-w)}^\ell.
    \end{align*}
    Combining this with \eqref{eq:VexpansionInw} and subtracting $V(w)$ we obtain \eqref{eq:Vexpansion}. For this one also uses the identity $\overline{V(z,\overline{w})}=V(\overline z, w)$ several times. The second part follows from the first part. For this one uses that for $|z-w|<\delta_n$
    \begin{align*}
        n^2 |z-w|^6 = \mathcal O\big(\frac{\log^6 n}{n}\big), \qquad 
        n|z-w|^4 = \mathcal O\big(\frac{\log^4 n}{n}\big).
    \end{align*} 
\end{proof}

\begin{lemma}
    Let $K$ be a compact subset of $\mathring S_V$. We have uniformly for $w\in K$ and $|z-w|<\delta_n$ that as $n\to\infty$
    \begin{multline} \label{eq:behav|Kn|2}
        |\mathcal K_n(z,w)|^2 = \frac1{16\pi^2}
        \left(n^2 (\Delta V(w))^2- (\frac14 n^3 |z-w|^2 \Delta V(w)-2n^2) \Delta V(w) \Re\big((z-w) \partial \Delta V(w)\big)
        +\mathcal O(n \log^6 n)\right) e^{-n \Delta V(w) |z-w|^2}\\
         + \mathcal O(e^{-\frac12 n \Delta V(w) |z-w|^2}) + \mathcal O(1/n^2).
    \end{multline}
\end{lemma}

\begin{proof}
      Using \eqref{eq:BergmanApprox1stO}, \eqref{eq:firstAppBergmanK} and the triangle inequality,
      we infer that
    \begin{align*}
        |\mathcal K_n(z,w)|^2 = \frac1{\pi^2} |n B_0(z,\overline{w})+B_1(z,\overline{w})|^2 |e^{-n ( V(z, \overline{w})-\frac12 V(z)-\frac12 V(w))}|^2
        + \mathcal O(e^{-n (V(z,\, \overline{w})-\frac12 V(z)- \frac12 V(w))})+ \mathcal O(1/n^2),
    \end{align*}
    as $n\to\infty$, uniformly for $w\in K$ and $|z-w|<\delta_n$. A Taylor expansion shows that
    \begin{align*}
        B_0(z, \overline{w}) = \frac14 \Delta V(w) + \frac14\partial \Delta V(w) (z-w) + \mathcal O(|z-w|^2)
    \end{align*}
    as $n\to\infty$ for $|z-w|<\delta_n$.
    It follows that
    \begin{align*}
        16 |n B_0(z,\overline{w})+B_1(z,\overline{w})|^2 = 
        16 n^2 |B_0(z, \overline w)|^2+\mathcal O(n)
        =n^2 (\Delta V(w))^2+2 n^2 \Delta V(w)  \Re\big((z-w)\partial \Delta V(w) \big)+\mathcal O(n^2|z-w|^2)+\mathcal O(n)
    \end{align*}
    as $n\to\infty$ for $|z-w|<\delta_n$. 
    Combining this with \eqref{eq:Vexpansion2} yields the result.

\end{proof}

To treat each of the individual terms in \eqref{eq:behav|Kn|2}, the following lemma will suffice.

\begin{lemma} \label{lem:hxy}
    Suppose that $\Delta V>0$ on $\mathring S_V$. Let $K$ be a compact subset of $\mathring S_V$. Let $h:K\to \mathbb C$ be any bounded integrable function. We have for all fixed $m,\ell=0,1,\ldots$ that
    \begin{align*}
        n^2\int_{A} \int_{\mathbb C\setminus A} h(w) |z-w|^{2m} (z-w)^\ell e^{-n \Delta V(w) |z-w|^2} d^2z d^2w &= \mathcal O(n^{\frac{1-2m-\ell}{2}})
    \end{align*}
    as $n\to\infty$, uniformly for all convex sets $A\subset K$ with a $C^2$ boundary, and the implied constant depends only on $K$. 
\end{lemma}

\begin{proof}
    Without loss of generality, we may assume that $\Delta V=1$. After a substitution we have for any $m, \ell=0,1,\ldots$ 
    \begin{align*}
        \int_{A} \int_{\mathbb C\setminus A} h(w) |z-w|^{2m} (z-w)^\ell e^{-n \Delta V(w) |z-w|^2} d^2z \, d^2w
        = \int_A \int_{-w+\mathbb C\setminus A} h(w) (x^2+y^2)^{m} (x+i y)^\ell e^{-n(x^2+y^2)} dx \, dy \, d^2w.
    \end{align*}
    We will prove a slightly stronger statement by induction. Namely, we claim that as $n\to\infty$
    \begin{align} \label{eq:doubleInd}
        n^2\int_A \int_{-w+\mathbb C\setminus A} h(w) x^{k} y^{q-k} e^{-n(x^2+y^2)} dx \, dy \, d^2w = \mathcal O(n^{\frac{1-q}2}),
    \end{align}
    for any non-negative integers $q$ and $k=0,1,\ldots,q$. The main result then follows by the binomial formula.  
    For $q=0$ and $k=0$ \eqref{eq:doubleInd} holds due to Proposition \ref{prop:AconvexInt}. Now let us assume that $q=1$ and $k=1$ (the case $q=1$ and $k=0$ is similar and we omit it).  
    We write $w=x'+i y'$ and notice that
    \begin{align*}
        \int_{-w+\mathbb C\setminus A} x e^{-n(x^2+y^2)} dx \, dy
        &= \int_{b_1-y'}^{b_2-y'} \int_{\tilde\varphi_+(y+y')-x'}^{\infty} x e^{-n(x^2+y^2)} dx \, dy
        + \int_{b_1-y'}^{b_2-y'} \int_{-\infty}^{\tilde\varphi_-(y+y')-x'} x e^{-n(x^2+y^2)} dx \, dy
    \end{align*}
    as $n\to\infty$. Note that the regions $\mathbb R\times (-\infty,b_1-y']$ and $\mathbb R\times [b_2-y',\infty)$ do not contribute due to anti-symmetry in the $x$ variable.  After an integration in both integrals we infer that
    \begin{align*}
        \int_{-w+\mathbb C\setminus A} x e^{-n(x^2+y^2)} dx dy
        &= \frac{1}{2n} \int_{b_1}^{b_2} e^{-n (\tilde\varphi_+(y)-x')^2-n (y-y')^2} dy
        -\frac{1}{2n} \int_{b_1}^{b_2} e^{-n (\tilde\varphi_-(y)-x')^2-n (y-y')^2} dy
    \end{align*}
    as $n\to\infty$. On the other hand, we find that
    \begin{align*}
        \left|\frac1n \int_A \int_{b_1}^{b_2} h(w) e^{-n (\tilde\varphi_+(y)-x')^2-n (y-y')^2} dy d^2w\right|
        &\leq \max_K |h| \frac1n \int_{\mathbb C} \int_{b_1}^{b_2}  e^{-n (\tilde\varphi_+(y)-x')^2-n (y-y')^2} dy d^2w
        = \max_K |h| \frac{b_2-b_1}{n^2}. 
    \end{align*}
    A similar estimate holds for the other integral and we thus infer that
    \begin{align*}
        n^2\int_A \int_{-w+\mathbb C\setminus A} h(w) x e^{- n |z-w|^2} d^2z d^2w= \mathcal O(1)
    \end{align*}
    as $n\to\infty$, where the implied constant depends only on $K$. We conclude that \eqref{eq:doubleInd} holds for $q=0$ and $q=1$. Now suppose that \eqref{eq:doubleInd} holds up to some $q$, and consider the integral
    \begin{align*}
        \int_A \int_{-w+\mathbb C\setminus A} h(w) x^{k} y^{q+1-k} e^{-n(x^2+y^2)} dx dy d^2w
    \end{align*}
    for any $k=0,1,\ldots,q+1$. Without loss of generality we may assume that $k\geq 1$ (if not then $q+1-k\geq 1$ and the argument is the same). An integration by parts yields 
    \begin{multline*}
        \int_{-w+\mathbb C\setminus A} x^k y^{q+1-k} e^{-n(x^2+y^2)} dx dy
        = \frac{1}{2n} \int_{b_1}^{b_2} (\tilde\varphi_+(y)-x')^{k-1} (y-y')^{q+1-k} e^{-n (\tilde\varphi_+(y)-x')^2-n (y-y')^2} dy\\
        \qquad -\frac{1}{2n} \int_{b_1}^{b_2} (\tilde\varphi_+(y)-x')^{k-1} (y-y')^{q+1-k} e^{-n (\tilde\varphi_-(y)-x')^2-n (y-y')^2} dy
        + \frac{k-1}{2n} \int_{-w+\mathbb C\setminus A} x^{k-2} y^{q+1-k} e^{-n(x^2+y^2)} d^2z
        +\mathcal O(e^{-c n})
    \end{multline*}
    for some constant $c>0$, where the $\mathcal O$ term comes from the regions $\mathbb R\times (-\infty,b_1]$ and $\mathbb R\times [b_2,\infty)$. 
    By the induction hypothesis we have as $n\to\infty$ that
    \begin{align*}
        n^2 \int_A \frac{k-1}{2n} \int_{-w+\mathbb C\setminus A} x^{k-2} y^{q-1-(k-2)} e^{-n|z-w|^2} d^2z d^2w = \mathcal O(n^{\frac{1-(q-1)}{2}-1})
        = \mathcal O(n^{\frac{1-(q+1)}{2}}).
    \end{align*}
    Furthermore, we have as $n\to\infty$ that
    \begin{multline*}
        \left|n^2 \int_A h(w) \frac{1}{2n} \int_{b_1}^{b_2} (\tilde\varphi_+(y)-x')^{k-1} (y-y')^{q+1-k} e^{-n (\tilde\varphi_+(y)-x')^2-n (y-y')^2} dy \, d^2w\right|\\
        \leq \frac12 n \max_K |h| \int_{b_1}^{b_2} \int_{-\infty}^\infty \int_{-\infty}^{\infty} |\tilde\varphi_+(y)-x'|^{k-1} |y-y'|^{q+1-k} e^{-n (\tilde\varphi_+(y)-x')^2-n (y-y')^2} d^2w \, dy\\
        = \frac{b_2-b_1}2 \Gamma\Big(\frac{k}{2}\Big) \Gamma\Big(\frac{q-k}{2}+1\Big) n^{-\frac{q}{2}} \max_K |h| = \mathcal O(n^{\frac{1-(q+1)}2}).
    \end{multline*}
    A similar estimate holds for the remaining integral. We conclude that \eqref{eq:doubleInd} is true for $q+1$. 
\end{proof}


\begin{lemma} \label{lem:VarIneqGinibre}
    Let $\mathcal K_n:\mathbb C^2\to \mathbb C$ be the weighted Bergman kernel with respect to the weight $e^{-n V(z)}$ as defined in \eqref{eq:defSymBergmanK}, where $V$ is a $C^2$ potential $V$ that satisfies \eqref{eq:Vgrowth} and which is assumed to be real analytic in a neighborhood of $S_V$. Fix a compact subset $K$ of the interior of $S_V$, and assume that $\Delta V>0$ on $K$.
    Then we have as $n\to\infty$
    \begin{align*}
        \int_{\mathbb C^2} (\mathfrak{1}_A(z)-\mathfrak{1}_A(w))^2 |\mathcal K_n(z, w)|^2 d^2z d^2w = \frac{n^2}{2\pi^2} \int_{\mathbb C^2} (\mathfrak{1}_A(z)-\mathfrak{1}_A(w))^2 \Delta V(w)^2 e^{-n\Delta V(w) |z-w|^2} d^2z \, d^2w
        +\mathcal O(1),
    \end{align*}
    as $n\to\infty$, uniformly for  convex $A\subset K$ with $C^2$ boundary. 
\end{lemma}

\begin{proof}
    Let us first rewrite the integral as
    \begin{align*}
        \int_{\mathbb C^2} (\mathfrak{1}_A(z)-\mathfrak{1}_A(w))^2 |\mathcal K_n(z, w)|^2 d^2z d^2w 
        &= 2\int_A \int_{\mathbb C\setminus A} |\mathcal K_n(z, w)|^2 d^2z d^2w.
    \end{align*}
    Using \eqref{eq:behav|Kn|2} combined with Lemma \ref{lem:hxy} for each individual term we get
    \begin{align*}
        \int_A \int_{\mathbb C\setminus A} \mathfrak{1}_{|z-w|<\delta_n} |\mathcal K_n(z, w)|^2 d^2z d^2w
        &= \frac{n^2}{16\pi^2} \int_A \int_{\mathbb C\setminus A}
        \Delta V(w)^2 e^{- n \Delta V(w) |z-w|^2} d^2z d^2w+\mathcal O(1)
    \end{align*}
    as $n\to\infty$. Here we used that the integration of $n^2 e^{-n \Delta V(w) |z-w|^2}$ over the region with $|z-w|\geq \delta_n$ is negligible. What remains is to show that the integration over our original integrand over the region where $|z-w|\geq \delta_n$ is negligible. Let us form the Berezin measure $B_n^{(w)}(z) d^2z$ rooted at $w$, where
    \begin{align*}
        B_n^{(w)}(z) = \frac{|\mathcal K_n(z,w)|^2}{\mathcal K_n(w,w)}. 
    \end{align*}
    It is known that the Berezin measure forms a probability measure, due to the reproducible property of the kernel. Based on \eqref{eq:behav|Kn|2} we find that as $n\to\infty$
    \begin{align*}
         B_n^{(w)}(z) = \frac{1}{2\pi} (n\Delta V(w)- n^2 \Delta V(w) \partial \Delta(w) |z-w|^2 \Re(z-w) + \mathcal O(n|z-w|^2)+\mathcal O(\log^6 n)) e^{-n \Delta V(w) |z-w|^2} + \mathcal O(1/n)
    \end{align*}
    for $|z-w|<\delta_n$. Then by Lemma \ref{lem:hxy} we get \begin{align*}
        \int_A \int_{\mathbb C\setminus A} \mathfrak{1}_{|z-w|<\delta_n} B_n^{(w)}(z) d^2z d^2w
        = \frac{n}{16\pi} \int_A \int_{\mathbb C\setminus A} \mathfrak{1}_{|z-w|<\delta_n} \Delta V(w) e^{-n \Delta V(w) |z-w|^2} d^2z d^2w
        + \mathcal O(1/n).
    \end{align*}
    Since $B_n^{(w)}$ is a probability measure, this implies using $\mathfrak{1}_{|z-w|\geq \delta_n}=1-\mathfrak{1}_{|z-w|< \delta_n}$ that
    \begin{align*}
        \int_A \int_{\mathbb C\setminus A} \mathfrak{1}_{|z-w|\geq\delta_n} B_n^{(w)}(z) d^2z d^2w
        &= \int_A d^2w - \frac{n}{16\pi} \int_A \int_{\mathbb C\setminus A} \mathfrak{1}_{|z-w|<\delta_n} \Delta V(w) e^{-n \Delta V(w) |z-w|^2} d^2z d^2w
        + \mathcal O(1/n)\\
        &= \frac{n}{16\pi} \int_A \int_{\mathbb C\setminus A} \mathfrak{1}_{|z-w|\geq\delta_n} \Delta V(w) e^{-n \Delta V(w) |z-w|^2} d^2z d^2w
        + \mathcal O(1/n) = \mathcal O(1/n).
    \end{align*}
    We conclude that
    \begin{align*}
        \int_A \int_{\mathbb C\setminus A} \mathfrak{1}_{|z-w|\geq\delta_n} |\mathcal K_n(z,w)|^2 d^2z d^2w = \int_A \int_{\mathbb C\setminus A} \mathfrak{1}_{|z-w|\geq\delta_n} \mathcal K_n(w,w) B_n^{(w)}(z) d^2z d^2w = \mathcal O(n) \int_A \int_{\mathbb C\setminus A} \mathfrak{1}_{|z-w|\geq\delta_n} B_n^{(w)}(z) d^2z d^2w = \mathcal O(1)
    \end{align*}
    as $n\to\infty$, which finishes the proof.
\end{proof}

\begin{proof}[Proof of Theorem \ref{thm:entropyConvex}]
    Under the conditions of the theorem, we know that $\Delta V$ is bounded from above and below by a positive constant. The first part of Theorem \ref{thm:entropyConvex} then follows from Lemma \ref{lem:VarIneqGinibre}, e.g., one has that
    \begin{align*}
        \frac{n^2}{2\pi^2}\int_{\mathbb C^2} (\mathfrak{1}_A(z)-\mathfrak{1}_A(w))^2 \Delta V(w)^2 e^{-n\Delta V(w) |z-w|^2} d^2z d^2w
        &\leq \frac{n^2}{32\pi^2}\int_{\mathbb C^2} (\mathfrak{1}_A(z)-\mathfrak{1}_A(w))^2 \big(\max_K \Delta V\big)^2 e^{-n (\min_K \Delta V) |z-w|^2} d^2z d^2w
    \end{align*}
    and one may apply Proposition \ref{prop:AconvexInt} after rescaling both the integration variables and noting that $|\sqrt{\min_K \Delta V} A|=\sqrt{\min_K \Delta V} |A|$. 
    The second part of the theorem follows directly from Proposition \ref{prop:AconvexInt}.
\end{proof}


\subsection{The number variance at the edge of the elliptic Ginibre ensemble} \label{sec:steepest}

To prove Theorem \ref{thm:strongRoughAnearEdge} we need the edge behavior rather than the bulk behavior of the kernel. Namely, we need to understand $\mathcal K_n(z, w)$ for $z$ and $w$ close to $\partial \mathcal E_\tau$, and close to each other. This situation was first worked out by Lee and Riser \cite{LeRi}. Sharper error terms for their results were recently obtained in Proposition 5.1 in \cite{Mo}. Namely, after a rescaling these results imply that
\begin{multline} \label{eq:edgeBehavKer}
    \left|\mathcal K_n\Big(z+\sqrt{1-\tau^2}\frac{u \vec{n}(z)}{\sqrt n}, z+\sqrt{1-\tau^2}\frac{u' \vec{n}(z)}{\sqrt n}\Big)\right|\\
    =
    \frac{n}{2\pi (1-\tau^2)} e^{-\frac{|u-u'|^2}{2}} \left|\erf\Big(\frac{u+\overline{u'}}{\sqrt 2}\Big)\right|
    + \mathcal O\left(\sqrt n e^{-\Re(u)^2-\Re(u')^2} (|u|^2+|u'|^2+n^{-\frac{1}{2}+3\nu})\right)
\end{multline}
as $n\to\infty$, where $z\in\partial \mathcal E_\tau$, $\vec{n}(z)$ denotes the outward unit normal vector at $z$ on the edge, and $0<\nu<\frac{1}{6}$ is any constant. Furthermore, the behavior as $n\to\infty$ is uniform over $z\in\partial \mathcal E_\tau$ and $u, u' = \mathcal O(n^\nu)$ (indeed \eqref{eq:edgeBehavKer} also holds for $\tau=0$, see (72) in \cite{Mo}).

\begin{proof}[Proof of Theorem \ref{thm:strongRoughAnearEdge}]
It suffices to show that the limiting variance exists and is the same for all $0\leq \tau<1$, because then it must equal the limit in Theorem 1.2 in \cite{AkByEb}. First we remark that
\begin{align*}
\operatorname{Var} \mathfrak X_n(\mathfrak{1}_A) = 2 \int_{A} \int_{\mathbb C\setminus A} |\mathcal K_n(z,w)|^2 d^2z \, d^2w.
\end{align*}
Using the asymptotic behavior of the kernel, we deduce that for any fixed $S\in\mathbb R$ as $n\to\infty$
\begin{align*}
\operatorname{Var} \mathfrak X_n(\mathfrak{1}_A) &= \frac{2c^2}{n} \int_{\partial \mathcal E_\tau} \int_{\partial \mathcal E_\tau} \int_S^{n^\nu} \int_{-n^\nu}^S  \left|\mathcal K\Big(z+c \frac{u}{\sqrt n} \vec{n}(z),w+c \frac{u'}{\sqrt n}\vec{n}(z)\Big)\right|^2 du \, du' \, dz \, dw+\mathcal O\left(e^{-\ell n^{2\nu}}\right),
\end{align*}
where $c = \sqrt{1-\tau^2}/\sqrt 8$, $\ell>0$ is some constant and $0<\nu<\frac{1}{6}$ is a fixed constant. The form of the $\mathcal O$ term can be argued from Theorem I.1 in \cite{ADM}.  As usual, we parametrize the boundary of the ellipse as $z = e^{i\eta}+\tau e^{-i\eta}$. 
We now split the integral in one where $|z-w|=\mathcal O(n^{-\frac{1}{2}+\nu})$ and the remaining part.
First, we consider the region where $|\eta-\eta'|\leq n^{-\frac{1}{2}+\nu}$. After a substitution $\theta=(\eta-\eta')\sqrt n$, this part of the integral gives 
\begin{multline*}
\frac{2c^2}{n\sqrt n} \int_{-\pi}^\pi \int_{-n^\nu}^{n^\nu} \int_S^{n^\nu} \int_{-n^\nu}^S  \left|\mathcal K_n\Big(e^{i\eta}+\tau e^{-i\eta}+c \frac{u}{\sqrt n} \vec{n}(z),e^{i\eta}+\tau e^{-i\eta}-i |e^{i\eta}-\tau e^{-i\eta}| \frac{\theta}{\sqrt n} \vec{n}(z)+c \frac{u'}{\sqrt n}\vec{n}(z)+\mathcal O(n^{-1+2\nu})\Big)\right|^2\\
\times \left(|e^{i\eta}-\tau e^{-i\eta}|^2+\mathcal O(n^{-\frac{1}{2}+2\nu})\right) du \, du' \, d\theta \, d\eta.
\end{multline*}
Here we made use of the fact that explicitly
\begin{align*}
\vec{n}(z) = \frac{\sinh(\xi_\tau+i\eta)}{|\sinh(\xi_\tau+i\eta)|}
= \frac{e^{i\eta}-\tau e^{-i\eta}}{|e^{i\eta}-\tau e^{-i\eta}|}.    
\end{align*}
Now inserting the behavior of the kernel \eqref{eq:edgeBehavKer}, the behavior of this integral is given by
\begin{multline*}
\frac{2\sqrt n}{8(2\pi)^2} \int_{-\pi}^\pi \int_{-n^\nu}^{n^\nu} \int_S^{n^\nu} \int_{-n^\nu}^S  \left|\erfc\left(-\frac{u+u'}{4}+i \sqrt{1-\tau^2} |e^{i\eta}-\tau e^{-i\eta}| \frac{\theta}{\sqrt 2}\right)\right|^2
e^{-(u-u')^2/8-(1-\tau^2)|e^{i\eta}-\tau e^{-i\eta}|^2 \theta^2}
\\
\times
\left(|e^{i\eta}-\tau e^{-i\eta}|^2+\mathcal O(n^{-\frac{1}{2}+2\nu})\right) du \, du' \, d\theta \, d\eta.
\end{multline*}
After a substitution $\theta\to \sqrt{1-\tau^2}|e^{i\eta}-\tau e^{-i\eta}| \theta$ this becomes
\begin{multline*}
\frac{1}{\sqrt{1-\tau^2}}\frac{\sqrt n}{16 \pi^2} \int_{-\pi}^\pi \int_{-\sqrt{1-\tau^2}|e^{i\eta}-\tau e^{-i\eta}|n^\nu}^{\sqrt{1-\tau^2}|e^{i\eta}-\tau e^{-i\eta}|n^\nu} \int_S^{n^\nu} \int_{-n^\nu}^S  \left|\erfc\left(-\frac{u+u'}{4}+i \frac{\theta}{\sqrt 2}\right)\right|^2
e^{-(u-u')^2/8-\theta^2}\\
\times
\left(|e^{i\eta}-\tau e^{-i\eta}|+\mathcal O(n^{-\frac{1}{2}+2\nu})\right) du du' d\theta d\eta\\
= \frac{1}{\sqrt{1-\tau^2}}\frac{\sqrt n}{16\pi^2} \int_{\partial \mathcal E_\tau} \int_{-\sqrt{1-\tau^2}|z^2-4\tau|^\frac{1}{2}n^\nu}^{\sqrt{1-\tau^2}|z^2-4\tau|^\frac{1}{2}n^\nu} \int_S^{n^\nu} \int_{-n^\nu}^S  \left|\erfc\left(-\frac{u+u'}{4}+i \frac{\theta}{\sqrt 2}\right)\right|^2
e^{-(u-u')^2/8-\theta^2}
\left(1+\mathcal O(n^{-\frac{1}{2}+2\nu})\right) 
du du' d\theta dz.
\end{multline*}
After some estimations (involving the asymptotics of $\erfc$ for large arguments) we may replace $n^\nu$ by $\infty$ in the integration bounds, and this yields
\begin{align} \label{eq:edgePartThm1.4}
\frac{1}{\sqrt{1-\tau^2}}\frac{\sqrt n}{16\pi^2} |\partial \mathcal E_\tau| \int_{-\infty}^{\infty} \int_0^{\infty} \int_{0}^\infty  \left|\erfc\left(-\frac{u-u'}{4}-\frac{S}{2}+i \frac{\theta}{\sqrt 2}\right)\right|^2
e^{-\frac{(u+u')^2}{2}} e^{-\theta^2}
du du' d\theta +\mathcal O(n^{2\nu})
\end{align}
as $n\to\infty$. For this, one uses estimates like $\int_{n^\nu}^\infty (T^2+\theta^2)^{-1} d\theta =
T^{-1} \arctan(T n^{-\nu}) = \mathcal O(n^{-\nu})$ uniformly for $T\geq 0$. The integral in \eqref{eq:edgePartThm1.4} can probably be simplified with light-cone coordinates, but for now it will be enough to know that the limit exists and is independent of $\tau$ after division by $|\partial E_\tau| \sqrt{n}/\sqrt{1-\tau^2}$. 

Next, we show that the remaining part of the integral, where $\theta>n^{-\frac{1}{2}+\nu}$, vanishes in the limit $n\to\infty$. Here, the asymptotics of the kernel are different. Let us momentarily assume that $0<\tau<1$. For big enough $T$, the integral is dominated by
\begin{multline*}
\frac{1}{n} \int_{-\pi}^\pi \int_{|\theta|>n^{-\frac{1}{2}+\nu}} \int_{-Tn^\nu}^{Tn^\nu} \int_{-Tn^\nu}^{T^\nu} \left|\mathcal K_n(2\sqrt\tau \cosh(\xi_\tau+\xi/\sqrt n+i\eta), 2\sqrt\tau \cosh(\xi_\tau+\xi'/\sqrt n+i\eta'))\right|^2\\
\times 4\tau |\sinh(\xi_\tau+\xi/\sqrt n+i\eta)| |\sinh(\xi_\tau+\xi'/\sqrt n+i\eta')| d\xi d\xi' d\eta d\eta'.
\end{multline*}
We know from \cite{ADM} that
\begin{align*}
\frac{1}{n}\left|\mathcal K_n(2\sqrt\tau \cosh(\xi_\tau+\xi/\sqrt n+i\eta), 2\sqrt\tau \cosh(\xi_\tau+\xi'/\sqrt n+i\eta')\right|^2
= \mathcal O\left(\frac{e^{-2 \xi^2 g(\xi_\tau+i\eta)} e^{-2 \xi'^2 g(\xi_\tau+i\eta')}}{\cosh(2\xi/\sqrt n+2\xi'/\sqrt n)-\cos(\eta-\eta')}\right)
\end{align*}
as $n\to\infty$, uniformly for $\xi, \xi'=\mathcal O(n^\nu)$ and $\sin(\frac{\eta-\eta'}{2})^{-2}=\mathcal O(n^{-\frac{1}{2}+\nu})$. Now setting $\theta_\pm=\eta'\pm \eta$, and using the fact that
\begin{align*}
\frac{1}{\cosh(2\xi/\sqrt n+2\xi'/\sqrt n)-\cos(\eta-\eta')}\leq \frac{1}{1-\cos(\eta-\eta')},
\end{align*}
we are tasked with estimating the integral
\begin{align*}
\int_{n^{-\frac{1}{2}+\nu}}^\pi \frac{d\theta_-}{1-\cos \theta_-} = \cot\Big(\frac{1}{2} n^{-\frac{1}{2}+\nu}\Big) = \mathcal O(n^{\frac{1}{2}-\nu}). 
\end{align*}
Hence the integral vanishes as $n\to\infty$, after division by $\sqrt n$. A similar procedure works for $\tau=0$, and the estimate may in fact also be obtained by taking the limit $\tau\downarrow 0$.
Thus the dominant part of $\operatorname{Var} \mathfrak X_n(\mathfrak{1}_A)$ is given by \eqref{eq:edgePartThm1.4} and we obtain Theorem \ref{thm:strongRoughAnearEdge} after division by $|\partial E_\tau| \sqrt{n}/\sqrt{1-\tau^2}$, since the result is the same for all $0\leq \tau<1$ and for $\tau=0$ the result is known \cite{AkByEb}.
\end{proof}


\subsection{Estimating the entropy}

\begin{proposition} \label{prop:VarIneqq}
    Consider any random normal matrix model with $C^{1}$ potential $V$ that satisfies \eqref{eq:Vgrowth}.  For any fixed $q>1$ and any measurable set $A\subset \mathbb C$ there exists a constant $c_{q,V}>0$ such that
    \begin{align*}
    c_{q,V} \operatorname{Var} \mathfrak X_n(\mathfrak{1}_A) \leq
        \mathfrak S_n^q(A) \leq \frac{4q \log 2}{q-1} \operatorname{Var} \mathfrak X_n(\mathfrak{1}_A).
    \end{align*}
\end{proposition}

\begin{proof} We start by recalling that from the determinantal structure one simply computes 
\begin{align*}
\operatorname{Var} \mathfrak X_n(\mathfrak{1}_A) =  \Tr \mathfrak{1}_A \mathcal K_n - \Tr (\mathfrak{1}_A \mathcal K_n)^2  =\sum_{j=1}^n \lambda_j (1-\lambda_j),   
\end{align*}
where $\lambda_j$ are the eigenvalues of $\mathfrak{1}_A \mathcal K_n$.  These are also eigenvalues of the overlap matrix as defined in \eqref{eq:overlapMatrix} (the overlap matrix represents the linear transformation associated to the operator). Therefore the Rényi entropy as defined in \eqref{eq:defRenyiEntropy} is given by
    \begin{align*}
        \mathfrak S_n^q(A) = \frac{1}{1-q} \sum_{j=1}^n \log(\lambda_j^q+(1-\lambda_j)^q).
    \end{align*}
Now consider the function
\begin{align*}
    h_q(\lambda) &= \frac1{1-q}\frac{\log(\lambda^q+(1-\lambda)^q)}{\lambda(1-\lambda)}, 
    & \lambda\in[0,1].
\end{align*}
We note that
\begin{align*}
    \lim_{\lambda\downarrow 0} h_q(\lambda)
    = \lim_{\lambda\uparrow 1} h_q(\lambda) = \frac{q}{q-1}. 
\end{align*}
Furthermore 
$2^{1-q}\leq\lambda^q+(1-\lambda)^q\leq 1$ for all $\lambda\in[0,1]$,
which follows, e.g., by differentiation, which shows that there is a unique minimum at $\lambda=\frac12$. 
Thus $h_q$ is a continuous strictly positive function on $[0,1]$. Consequently, we obtain constants $0<c_{q,V}\leq C_{q,V}$ such that $c_{q,V} \leq h_q(\lambda) \leq C_{q,V}
$.
We infer that
\begin{align*}
    c_{q,V} \operatorname{Var} \mathfrak X_n(\mathfrak{1}_A) \leq
        \mathfrak S_n^q(A) \leq C_{q,V} \operatorname{Var} \mathfrak X_n(\mathfrak{1}_A).
\end{align*}
It remains to find an explicit $C_{q,V}$ for which the upper bound is valid.
Since the expression is invariant under the substitution $\lambda\mapsto 1-\lambda$ it suffices to restrict our attention to $\lambda\in[0,\frac12]$. 
It follows from the trivial inequality $\lambda^q+(1-\lambda)^q\geq (1-\lambda)^q$ that for all $\lambda\in (0,\frac12]$
\begin{align*}
    -\frac{\log(\lambda^q+(1-\lambda)^q)}{\lambda(1-\lambda)} \leq
    -q \frac{\log(1-\lambda)}{\lambda(1-\lambda)}.
\end{align*}
Using the Taylor series of the expressions involved we infer that
\begin{align*}
    -\frac{\log(1-\lambda)}{\lambda(1-\lambda)}
    = \sum_{j=0}^\infty \frac{\lambda^j}{j+1} \sum_{k=0}^\infty \lambda^k    = \sum_{k=0}^\infty \left(\sum_{j=0}^k \frac1{j+1}\right) \lambda^k
\end{align*}
for any $\lambda\in(0,\frac12]$, and since all coefficients in the power series are positive, the maximum on $(0,\frac12]$ is attained at $\lambda=\frac12$. We conclude that
\begin{align*}
\frac1{1-q} \frac{\log(\lambda^q+(1-\lambda)^q)}{\lambda(1-\lambda)} 
\leq \frac{q}{q-1} 4 \log 2
\end{align*}
for all $\lambda\in (0,\frac12]$, and thus for $\lambda\in(0,1)$ by symmetry. 
We conclude that
\begin{align*}
   \mathfrak S_n^q(A) = \frac{1}{1-q} \sum_{j=1}^n \log(\lambda_j^q+(1-\lambda_j)^q) \leq \frac{q}{q-1} 4\log 2 \sum_{j=0}^n \lambda_j (1-\lambda_j)
   = \frac{4 q\log 2}{q-1} \operatorname{Var} \mathfrak X_n(\mathfrak{1}_A).
\end{align*}
\end{proof}

\begin{proof}[Proof of Theorem \ref{thm:holographic}]
This is an immediate consequence of Theorem \ref{thm:entropyConvex} and Proposition \ref{prop:VarIneqq}. As is clear from the proof of Theorem \ref{thm:entropyConvex}, the constants can be picked as
\begin{align*}    \frac{(\displaystyle\min_K \Delta V)^2}{(\displaystyle \max_K \Delta V)^\frac{3}{2}}
    \leq 4\pi^\frac32 \lim_{n\to\infty} \frac{\operatorname{Var} \mathfrak X_n(\mathfrak{1}_A)}{\sqrt{n} \, |\partial A|} \leq \frac{(\displaystyle\max_K \Delta V)^2}{(\displaystyle \min_K \Delta V)^\frac{3}{2}}.
\end{align*} 
When $\Delta V>0$ on $S_V$ we may in fact replace the extrema of $\Delta V$ on $K$ by the extrema on $S_V$, and the resulting constants do not depend on $K$ anymore. 
\end{proof}

\begin{remark} \label{rem:constqto1}
    In the limit $q\to\infty$ the proportionality constant for the upper bound in the inequality relating the entropy and the variance becomes $4\log 2$. Based on simulations there appears to be a critical value of $q_*$ in $(1,2)$ such that
    \begin{align*}
        -\frac{\log(\lambda^q+(1-\lambda)^q)}{\lambda(1-\lambda)} \leq 4 \log 2
    \end{align*}
    for $q\geq q_*$. The fact that
    \begin{align*}
        \lim_{\lambda\downarrow 0} -\frac{\log(\lambda^q+(1-\lambda)^q}{\lambda(1-\lambda)}         =\lim_{\lambda\uparrow 1} -\frac{\log(\lambda^q+(1-\lambda)^q}{\lambda(1-\lambda)}
        = \frac{q}{q-1}
    \end{align*}
    is the main reason that such an inequality cannot hold for values of $q$ close to $1$. However, in determining the entropy we are averaging and there is some optimism that one nevertheless has 
    \begin{align*}
   \mathfrak S_n^q(A)\leq C \operatorname{Var} \mathfrak X_n(\mathfrak{1}_A)
\end{align*}
for all $q>1$ and generic random normal matrix models, where $C>0$ is a uniform constant independent of $q$ and $A$. If this is the case, then one may take the limit $q\downarrow 1$ and conclude a similar inequality for the von Neumann entropy. For $q>1$ close to $1$ it in fact seems that $4\log 2$ is the lower bound rather than the upper bound, based on simulations. In \cite{DeleporteLambert} a lower bound was given for the von Neumann entropy in a related setting. Note that our setting allows for a lower bound for the von Neumann entropy as well, since $-\lambda\log \lambda-(1-\lambda)\log(1-\lambda)\geq 4\log(2) \lambda(1-\lambda)$ for $\lambda\in(0,1)$.  
\end{remark}

\section{A steepest descent analysis for the elliptic Ginibre ensemble} \label{sec:smoothVarElliptic}

Many of our results for the variance of smooth linear statistics of the elliptic Ginibre ensemble will follow from a steepest descent approach first introduced in \cite{ADM}, that allows to find the asymptotic behavior of the correlation kernel (or equivalently a truncation of the Mehler kernel) in various regions and regimes. The variance can then be determined using \eqref{eq:varIntKer}.
The kernel $\mathcal K_n(z,w)$ has a single integral representation that is appropriate for a steepest descent analysis. We remark that most arguments are considerably different from \cite{ADM} in the regimes of weak non-Hermiticity considered. 
To calculate the variance, we need to patch together various regions of integration in \eqref{eq:varIntKer}, where the kernel has a different asymptotic behavior to be determined in this section.\\

In this section we recap some results from \cite{ADM} that will be used throughout the rest of the paper. Furthermore, we introduce a new (though related) single integral representation that is especially suited for the weak non-Hermiticity regime. 
Throughout this section we assume $0<\tau<1$. In what follows we use elliptic coordinates $\xi+i\eta$ and $\xi'+i\eta'$, and we denote
\begin{align} \label{eq:defzzprimEGE}
z=2\sqrt \tau\cosh(\xi+i\eta) \qquad\text{ and }\qquad w=2\sqrt \tau\cosh(\xi'+i\eta'),
\end{align}
where $\xi\geq 0$ and $\eta\in (-\pi,\pi]$ when $\xi>0$, $\eta\in [0,\pi]$ when $\xi=0$, and similar for $\xi'$ and $\eta'$. 
Note that any fixed $\xi>0$ corresponds to a particular ellipse. We shall also consider the light-cone versions
\begin{align} \label{eq:defXiEtaLightcone}
    \xi_\pm = \frac{\xi\pm \xi'}{2}, \qquad \text{ and }\qquad \eta_\pm = \frac{\eta\pm \eta'}{2}.
\end{align}
It was shown in \cite{ADM} that the kernel has a single integral representation of the form
\begin{align} \label{eq:EGEInwithWeights}
\mathcal K_n(z,w) = \frac{n}{\pi\sqrt{1-\tau^2}} \sqrt{\omega\left(\sqrt n z\right) \omega\left(\sqrt n w\right)} I_n(\tau;z,w),
\end{align}
with $I_n(\tau;z,w)$ defined as
\begin{align} \label{eq:defIn}
I_n(\tau;z,w) := -\frac{1}{2\pi i} \oint_{\gamma_0} \frac{e^{n F(s)}}{s-\tau} \frac{ds}{\sqrt{1-s^2}},
\end{align}
where $\gamma_0$ is a small contour with positive orientation, that encloses $0$ but not $\tau$, $\sqrt{1-s^2}$ is defined with cut $(-\infty,-1]\cup [1,\infty)$ and positive on $(-1,1)$, and $F(s) = F(\tau;z,w; s)$ is defined by
\begin{align} \label{eq:defF}
F\left(\tau;z,w; s\right) :=& \frac{s(z+\overline{w})^2}{4\tau(1+s)}-\frac{s (z-\overline{w})^2}{4\tau(1-s)} - \log s + \log \tau.
\end{align}
The branch for the logarithm is not relevant for \eqref{eq:defIn}. 
In \cite{ADM} the definition of elliptic coordinates  was with a complex conjugation on $w$ in \eqref{eq:defzzprimEGE} (except in Theorems I.1 and I.2). This created a certain symmetry (see \eqref{eq:defF} above) in the corresponding steepest descent analysis. However, since we would like to treat $z$ and $w$ symmetrically in the necessary integrations to determine the variance, we now choose to define elliptic coordinates without the complex conjugation on $w$. This means that our expressions will have an extra minus sign in front of $\eta'$ compared to \cite{ADM}. The choice $\xi=\xi_\tau$, where
$$\xi_\tau=-\frac{1}{2} \log\tau,$$ 
corresponds to $\partial\mathcal E_\tau$, the edge of the droplet. 

\subsection{Summary of previous results: saddle points and other relevant expressions}

What follows is essentially a summary of results relating to the saddle point analysis in \cite{ADM}. The saddle points of $F$ have a remarkably simple form if we view them in elliptic coordinates. We define the expressions 
\begin{align} \label{eq:defab}
a = e^{\xi+\xi'} e^{i(\eta-\eta')} = e^{2(\xi_++i\eta_-)}\qquad\text{ and }\qquad 
b = e^{\xi-\xi'} e^{i(\eta+\eta')} = e^{2(\xi_-+i\eta_+)}.
\end{align}
We have the following proposition, which was proved in \cite{ADM}.  

\begin{proposition} \label{prop:saddlePointsDef}
Let $z,w\in \mathbb C\setminus \{-2\sqrt \tau, 2\sqrt \tau\}$, then the saddle points of $s\mapsto F(\tau;z,w;s)$ are simple, and:
\begin{itemize}
\item[(i)] When $z \neq  \pm \overline{w}$, there are exactly four saddle points  given by $a, a^{-1}, b$ and $b^{-1}$.
\item[(ii)] When $z=\pm \overline{w}$ and $z\neq 0$, there are exactly two  saddle points, which are given by $a$ and $a^{-1}$.
\end{itemize}
If $z\in \{-2\sqrt\tau,2\sqrt\tau\}$ or $w\in \{-2\sqrt\tau,2\sqrt\tau\}$, then all saddle points have order two and we have the following:
\begin{itemize}
    \item[(iii)] If $z\in \{-2\sqrt\tau, 2\sqrt\tau\}$ and $w\not\in \{-2\sqrt\tau, 2\sqrt\tau\}$, then we have two saddle points $a=b^{-1}$ and $a^{-1}=b$.
        \item[(iv)] If $z\not\in \{-2\sqrt\tau, 2\sqrt\tau\}$ and $w\in \{-2\sqrt\tau, 2\sqrt\tau\}$, then we have two saddle points $a=b$ and $a^{-1}=b^{-1}$;
    \item[(v)] If $z=\pm w\in\{-2\sqrt\tau, 2\sqrt \tau\}$, then we have one saddle point $a^{-1}=a=b=b^{-1}=\pm 1$.
\end{itemize}
Finally, when $z=w=0$ there are no saddle points.
\end{proposition}

Furthermore, expressions for $F$ and $F''$ in the saddle points also take a simple form \cite{ADM}.

\begin{proposition} \label{prop:Finsaddleelliptic}
Let $z, w\in \mathbb C$. We have
\begin{align} \label{prop:Finsaddleelliptica}
F(a) &= 1 + \log \tau - \xi - \xi' - i(\eta-\eta') +  \frac{1}{2} e^{2 (\xi + i \eta)} + \frac{1}{2} e^{2 (\xi'- i \eta')},\\ \label{prop:Finsaddleelliptica-}
F(a^{-1}) &= 1 + \log \tau + \xi + \xi' + i(\eta-\eta') +  \frac{1}{2} e^{-2 (\xi + i \eta)} + \frac{1}{2} e^{-2 (\xi'- i \eta')},\\
F(b) &= 1 + \log \tau - \xi + \xi' - i(\eta+\eta') +  \frac{1}{2} e^{2 (\xi + i \eta)} + \frac{1}{2} e^{-2 (\xi'- i \eta')},\\
F(b^{-1}) &= 1 + \log \tau + \xi - \xi' + i(\eta+\eta') +  \frac{1}{2} e^{-2 (\xi + i \eta)} + \frac{1}{2} e^{2 (\xi'- i \eta')}.
\end{align}
These identities hold up to multiples of $2\pi i$ depending on the choice of the branch of the logarithm in \eqref{eq:defF} but this is irrelevant for the integral $I_n$ in \eqref{eq:defIn}. Furthermore, we have
\begin{align}
    F(\tau) &= -\frac{\tau (z^2+\overline{w}^2)-2 z \overline{w}}{2(1-\tau^2)}.
\end{align}
\end{proposition}

Note that the value of $F(\tau)$ is important when the residue at $s=\tau$ of the integrand of \eqref{eq:defIn} contributes.

\begin{proposition} \label{prop:F''a-1}
We have
\begin{align}
F''(a^{\pm 1}) = \mp 2 a^{\mp 2} \frac{\sinh(\xi+i\eta) \sinh(\xi'-i\eta')}{\sinh(\xi+i\eta+\xi'-i\eta')}
\qquad\text{ and }\qquad
F''(b^{\pm 1}) = \pm 2 b^{\mp 2} \frac{\sinh(\xi+i\eta) \sinh(\xi'-i\eta')}{\sinh(\xi+i\eta-\xi'+i\eta')},
\end{align}
unless these are not well-defined, which happens when $z, w\in\{-2\sqrt\tau, 2\sqrt\tau\}$ for both identities,  $z=-w\in (-2\sqrt\tau, 2\sqrt\tau)$ for the first identity, and $z=w\in (-2\sqrt\tau, 2\sqrt\tau)$ for the second identity. 
\end{proposition}

In the case of strong non-Hermicity these results have been used to find the asymptotic behavior of the kernel in virtually any region \cite{ADM, Mo}. These results can readily be used to calculate the variance in Section \ref{sec:sec1}. The main result that was derived with a steepest descent analysis, Theorem I.1 together with Remark I.2 in \cite{ADM}, can be recapped as the following proposition. 

\begin{proposition} \label{prop:mainThmADM}
    Let $0<\tau<1$ and $0\leq \nu<\frac{1}{6}$ be fixed. Write $z, w\in\mathbb C$ in elliptic coordinates as defined in \eqref{eq:defzzprimEGE}. As previously defined, denote $\xi_+=\frac{1}{2}(\xi+\xi'), \, \eta_-=\frac12(\eta-\eta')$ and $\xi_\tau=-\frac{1}{2}\log\tau$. We have
    \begin{multline} \label{eq:mainThmADM}
        \mathcal K_n(z, w) = \frac{n \mathfrak{1}_{\xi_+<\xi_\tau}}{\pi(1-\tau^2)} \exp\left(-\frac{n}{1-\tau^2} \frac{|z-w|^2}{2}\right) C_\tau^n(z, w)\\
        + \sqrt{\frac{n}{32\pi^3\tau(1-\tau^2)}} 
        \frac{e^{-n(\xi-\xi_\tau)^2 g(\xi+i\eta))}e^{-n(\xi-\xi_\tau)^2 g(\xi'+i\eta'))}}{\sinh(\xi_+-\xi_\tau+i\eta_-) \sqrt{\sinh(\xi+i\eta) \sinh(\xi'+i\eta')}} D_\tau^n(z, w)\\
        + \mathcal O(n^{3\nu} e^{-n(\xi-\xi_\tau)^2 g(\xi+i\eta))}e^{-n(\xi-\xi_\tau)^2 g(\xi'+i\eta'))})
    \end{multline}
    as $n\to\infty$, uniformly for $z, w\in\mathbb C$ satisfying $((\xi_+-\xi_\tau)^2+\sin^2\eta_-)^{-1}=\mathcal O(n^{1-2\nu})$. Here $C_\tau^n(z, w)$ and $D_\tau^n(z, w)$ are unimodular factors and $g:\mathbb C\to [c,\infty)$, for some $c>0$, is a continuous function explicitly defined as
    \begin{align} \label{eq:defG}
        g(\xi+i\eta) = \frac{\frac{1}{2}+\xi-\xi_\tau+\frac{1}{2}e^{-2\xi} \cos(2\eta)-\frac{2\tau}{1+\tau} \cosh^2\xi\cos^2\eta-\frac{2\tau}{1-\tau} \sinh^2\xi\sin^2\eta}{(\xi-\xi_\tau)^2}.
    \end{align}
\end{proposition}

The first term on the RHS of \eqref{eq:mainThmADM} essentially describes the bulk behavior of the kernel, while the second term essentially describes the edge behavior (unless $z$ and $w$ are too close to each other). Note that $C_\tau^n(z, w)$ and (in some cases) $D_\tau^n(z, w)$ differ from their definition in \cite{ADM}. In the current paper, all we really need to know about these factors is that they have modulus $1$. The following related result appeared (in a slightly different form) as Theorem III.5 in \cite{ADM} and concerns an inequality that is uniform in $\tau$ and $n$. It will turn out to be quite a convenient inequality for the bulk part contributing to the integral in \eqref{eq:varIntKer}.

\begin{proposition} \label{prop:kernelIneq}
    For all $z, w\in\mathbb C$, $\tau\in(0,1)$ and $n=1,2,\ldots$ we have the inequality
    \begin{align*}
        \left||\mathcal K_n(z, w)| - \frac{n \mathfrak{1}_{\xi_+<\xi_\tau}}{\pi(1-\tau^2)} \exp\left(-\frac{n}{1-\tau^2} \frac{|z-w|^2}{2}\right) \right| 
        \leq \frac{K}{2\pi\sqrt{1-\tau^2}} \frac{n}{|1-e^{2(\xi_+-\xi_\tau)}|} e^{-n(\xi-\xi_\tau)^2 g(\xi+i\eta))}e^{-n(\xi-\xi_\tau)^2 g(\xi'+i\eta'))},
    \end{align*}
    where $K=\frac{1}{\pi} \int_0^\pi \frac{dt}{\sqrt{2\sin t}}$.
\end{proposition}

The situation described by Proposition \ref{prop:mainThmADM} concerns the case where $z$ and $w$ cannot be both close to the edge \textit{and} close to each other. In this particular scenario the so-called Faddeeva plasma kernel (also called error function kernel) emerges. The result with the sharpest error terms appeared as Proposition II.5 in \cite{Mo}. It turns out that the contribution of this particular case to \eqref{eq:varIntKer} is negligible  in the case that $\tau$ is fixed.

\subsection{A single integral representation for the weak non-Hermiticity regime} \label{sec:steepestWeakNonH}

In the weak non-Hermicity regime the steepest descent analysis has only been done in the case $\alpha=\gamma=1$ in \cite{ADM,ACV} (although the results that were proved have been known since 1998 \cite{FyKhSo}). In what follows we assume that $\tau=1-\kappa n^{-\alpha}$ and we rescale the eigenvalues with a factor $n^{-\gamma}$. Here and everywhere else $\kappa$ is a positive constant. To eventually calculate the variance we shall use
\begin{align} \nonumber
\operatorname{Var} \mathfrak X_n(f) &=
\frac{1}{2} \iint_{\mathbb C^2} (f(n^\gamma z)-f(n^\gamma w))^2 \left|\mathcal K_n\left(z, w\right)\right|^2 d^2z d^2w\\ \label{eq:intVarGammaComp}
&= n^{-2\gamma} \int_{\mathbb C} f(u)^2 \mathcal K_n(n^{-\gamma}u, n^{-\gamma}u) d^2u - n^{-4\gamma}\iint_{\mathbb C^2} f(u) f(u') |\mathcal K_n(n^{-\gamma}u, n^{-\gamma}u')|^2 d^2u d^2u'.
\end{align}
Based on the second line we may use the assumption that $f$ has compact support, which effectively means that $u$ and $u'$ are bounded. It is tempting to use the second line as the basis for our analysis, it will turn out however that the first line is more convenient to determine the smooth linear statistics. This has to do with some intricate cancellations that are hard to see when using only the second line. There are also some advantages with this choice, namely that we only have to consider the parameter $\alpha$ in our results below, and may ignore $\gamma$ meanwhile. In fact, it will turn out that most of our results are easiest to acquire in elliptic coordinates. A change of variables to elliptic coordinates as in \eqref{eq:defzzprimEGE} shows that
\begin{multline*}
    \operatorname{Var} \mathfrak X_n(f)
    = \frac{1}{2} \int_{-\pi}^\pi \int_{-\pi}^\pi \int_0^\infty \int_0^\infty (f(2\sqrt\tau n^\gamma \cosh(\xi+i\eta))-f(2\sqrt\tau n^\gamma \cosh(\xi'+i\eta')))^2\\ 
   \times  |\mathcal K_n(2\sqrt\tau\cosh(\xi+i\eta), 2\sqrt\tau\cosh(\xi+i\eta)')|^2 16\tau^2
    |\sinh(\xi+i\eta)|^2 |\sinh(\xi'+i\eta')|^2 d\xi \, d\xi' \, d\eta \, d\eta'.
\end{multline*}
For $\xi\to 0$, notice that
\begin{align*}
    \Re z &= 2\sqrt\tau \cosh\xi\cos\eta= 2\sqrt\tau \cos\eta+\mathcal O(\xi^2),\\
    \Im z &= 2\sqrt\tau \sinh\xi\sin\eta
    = 2\sqrt \tau \xi\sin\eta+\mathcal O(\xi^3).
\end{align*}
For the contribution to the integral to be non-negligible we remark that then $\xi=\mathcal O(n^{-\gamma})$ and at least one of $\eta, \eta'$ must be close to $\pm \frac{\pi}{2}$. Note that when $\xi=\mathcal O(n^{-\gamma})$ the two points $z$ and $w$ are close to each other when either $\eta_-=\frac{\eta-\eta'}2\approx 0 \mod 2\pi$, or $\eta_+=\frac{\eta+\eta'}2\approx 0\mod 2\pi$. We may exclude the posibility of having these cases simultaneously since this situation corresponds to both $z$ and $w$ being close to $\pm 2$, for which $(f(n^\gamma z)-f(n^\gamma w))^2$ vanishes in the limit $n\to\infty$. The following proposition avoids these cases.

\begin{proposition} \label{prop:saddlePointsWeakNonH}
    Let $\delta>0$ be fixed and consider a set of $z, w$ such that $\eta, \eta'\in (-\pi-\delta,-\delta)\cup (\delta, \pi-\delta)$.\\
    In the limit $\xi_+=\frac{\xi+\xi'}2\to 0$ and (one of) $\eta_\pm=\frac{\eta\pm\eta'}2\to 0$ we have uniformly in $z, w$ and $\tau\in[0,1)$ that
    \begin{align*}
        a &= \begin{cases}
            1+\frac{z-\overline{w}}{2i \sqrt\tau \sin\eta_+}+\mathcal O(\xi_+^2+\eta_-^2), & \eta_-\to 0,\\
            e^{2i\eta_-} + \mathcal O(\xi_+), & \eta_+\to 0.
        \end{cases}\\
        b &= \begin{cases}
            1+\frac{z-\overline{w}}{2i \sqrt\tau \sin\eta_+}+\mathcal O(\xi_+^2+\eta_-^2), & \eta_+\to 0,\\
            e^{2i\eta_-} + \mathcal O(\xi_+), & \eta_-\to 0.
        \end{cases}
    \end{align*}
    Hence, when $\xi_+\to 0$ and $\eta_\pm\to 0$, two saddle points of $F(\tau; z, w; s)$ are close to $1$, while the remaining saddle points (if any) are bounded away from $1$. Alternatively, this behavior holds under the condition $\Im z, \Im w, \Re z-\Re w\to 0$.
\end{proposition}

\begin{proof}
So let us consider elliptic coordinates $z=2\sqrt\tau\cosh(\xi+i\eta)$ and $w=2\sqrt\tau\cosh(\xi'+i\eta')$. When $\xi\to 0$ and $\xi'\to 0$ we have
\begin{align*}
    z &= 2\sqrt\tau (\cos\eta+i\xi\sin\eta+\mathcal O(\xi^2)),\\
    \overline{w} &= 2\sqrt\tau (\cos\eta'-i\xi\sin\eta'+\mathcal O(\xi'^2)).
\end{align*}
In particular
\begin{align*}
    \frac{z-\overline{w}}{2\sqrt\tau} = \cos\eta-\cos\eta'+i (\xi\sin\eta+\xi'\sin\eta')+\mathcal O(\xi_+^2).
\end{align*}
If we now have $\eta_\pm\to 0$ as well this turns into 
\begin{align*}
    \frac{z-\overline{w}}{2\sqrt\tau} &= -\sin \eta_\pm 2\eta_\mp+i 2\xi_\pm \sin \eta_\pm+\mathcal O(\xi_+^2+\xi_+\eta_\mp+\eta_\mp^2)\\
    &= 2i\sin \eta_\mp (\xi_\pm+i\eta_\mp)+O(\xi_+^2+\eta_\mp^2).
\end{align*}
We thus get
\begin{align*}
    \xi\pm \xi'+i(\eta\mp\eta')
    = \frac{z-\overline{w}}{2i \sqrt\tau \sin\eta_\pm}+\mathcal O(\xi_+^2+\eta_\mp^2)
\end{align*}
as $\xi_+\to 0$ and $\eta_\pm\to 0$. We conclude that
\begin{align*}
    a = e^{\xi+ \xi'+i(\eta-\eta')}
    = e^{\frac{z-\overline{w}}{2i \sqrt\tau \sin\eta_+}} (1+\mathcal O(\xi_+^2+\eta_-^2))
    = 1+\frac{z-\overline{w}}{2i \sqrt\tau \sin\eta_+}+\mathcal O(\xi_+^2+\eta_-^2)
\end{align*}
as $\xi_+\to 0$ and $\eta_-\to 0$, while
\begin{align*}
    b = e^{\xi-\xi'+i(\eta+\eta')}
    = e^{\frac{z-\overline{w}}{2i \sqrt\tau \sin\eta_-}} (1+\mathcal O(\xi_+^2+\eta_+^2))
    = 1+\frac{z-\overline{w}}{2i \sqrt\tau \sin\eta_+}+\mathcal O(\xi_+^2+\eta_+^2)
\end{align*}
as $\xi_+\to 0$ and $\eta_+\to 0$. On the other hand, we have
$a = e^{2i\eta_-} + \mathcal O(\xi_+)$
as $\xi_+\to 0$ and $\eta_+\to 0$, while
$b = e^{2i\eta_+} + \mathcal O(\xi_-)
    = e^{2i\eta_+} + \mathcal O(\xi_+)$ as $\xi_+\to 0$ and $\eta_-\to 0$. The fact that these are bounded away from $1$ is due to the condition that $\eta, \eta'$ have to stay a distance $\delta$ away from $0$ and $\pm\pi$. The conditions imply that we can only be in case (i) or (ii) of Proposition \ref{prop:saddlePointsDef}. In case (i) all points $a, a^{-1}, b, b^{-1}$ are saddle points and thus two of them are close to $1$. In the second case $a$ and $a^{-1}$ are the only two saddle points. Indeed, when $z=\overline{w}$ we have $\eta=\eta'$. Hence this situation corresponds trivially to the limit $\eta_-\to 0$. The case $z=-\overline{w}$ implies $\eta'=-\eta\pm \pi$, and this does not contradict the conditions on $\eta, \eta'$, only when $\eta_-\to 0$. 
\end{proof}

As one can see two of the saddle points converge to $1$ when $\xi_+\to 0$ and $\eta_\pm\to 0$, and hence coalesce with both $\tau=1-\kappa n^{-\alpha}$ and the singularity $s=1$ of $s\mapsto \frac{1}{\sqrt{1-s^2}}$ as $n\to\infty$. This is a delicate situation where a different but trivially related single integral representation significantly improves the analysis. The purpose of the following single integral representation is that it removes the singular behavior of $s\mapsto \frac{1}{\sqrt{1-s^2}}$ in the integrand near $s=1$. This creates a situation where we still have a saddle point coalescing with a pole, but it is known how to handle such a situation \cite{Mo}. 

\begin{figure}
    \begin{minipage}{0.48\textwidth}
\centering
\begin{tikzpicture}[scale=1.2]

    \draw[->] (-3,0) -- (3,0) node[right] {$\Re(t)$}; 
    \draw[-] (0,-{sqrt (2)}) -- (0,{sqrt(2)});  

    \draw[->, dashed, red] (0,{sqrt(2)}) -- (0,4); 
    \draw[dashed, red] (0,{-sqrt(2)}) -- (0,-3) node[midway,right] {};

    \draw[thick, blue, ->] (0.2,1) arc[start angle=0, end angle=360, radius=0.2];

    \node[above left] at (0,4) {$\Im(t)$};
    \fill[red] (0,{sqrt(2)}) circle (2pt);
    \node[above left] at (0,{sqrt(2)}) {$i\sqrt{2}$};
    \fill[red] (0,{-sqrt(2)}) circle (2pt);
    \node[below left] at (0,{-sqrt(2)}) {$-i\sqrt{2}$};
    \fill[black] (0,1) circle (2pt);
    \node[above left] at (-0.1,1) {$i$};
    \node[below right] at (0.1,1) {$\gamma_i$};

    \fill[black] (0,0.4) circle (2pt);
    \node[above left] at (0,0.2) {$i\sqrt{1-\tau}$};
    \fill[black] (0,-0.4) circle (2pt);
    \node[below left] at (0,-0.2) {$-i\sqrt{1-\tau}$};

    \fill[black] (0,0) circle (2pt);
    \node[below right] at (0,0) {0};

\end{tikzpicture}
\end{minipage}
\hfill
\begin{minipage}{0.48\textwidth}
\centering
\begin{tikzpicture}[scale=1.2]

    \draw[->] (-3,0) -- (3,0) node[right] {$\Re(t)$}; 
    \draw[-] (0,-{sqrt (2)}) -- (0,{sqrt(2)});  

    \draw[->, dashed, red] (0,{sqrt(2)}) -- (0,4); 
    \draw[dashed, red] (0,{-sqrt(2)}) -- (0,-3);

    \draw[thick, blue] 
        (-3, -0.1) 
        .. controls (-2.5, -0.15) and (-1.5, -0.4) .. (-1, -0.5) 
        .. controls (-0.5, -0.6) and (-0.3, -0.2) .. (0, 0) 
        .. controls (0.3, 0.2) and (0.5, 0.6) .. (1, 0.5) 
        .. controls (1.5, 0.4) and (2.5, 0.15) .. (3, 0.1);

    \draw[thick, blue] 
        (1, 0.5) 
        .. controls (1.5, 0.4) and (2.5, 0.15) .. (3, 0.1)
        node[midway] {>};

    \draw[thick, blue] 
        (-3, 1.5)
        .. controls (-2.8, 1.3) and (-1, 1.1) .. (-0.38, 0.92)
        node[midway] {<}
        .. controls (-0.2, 0.88) and (0.2, 0.88) .. (0.38, 0.92) 
        .. controls (1, 1.1) and (2.8, 1.3) .. (3, 1.5);

    \fill[black] (-0.38, 0.92) circle (2pt) node[above left] {$\tilde b_1$};
    \fill[black] (0.38, 0.92) circle (2pt) node[above right] {$\tilde b_2$};

    \node[above left] at (0,4) {$\Im(t)$};
    \fill[red] (0,{sqrt(2)}) circle (2pt);
    \node[above left] at (-0.05,{sqrt(2)}) {$i\sqrt{2}$};
    \fill[red] (0,{-sqrt(2)}) circle (2pt);
    \node[below right] at (0,{-sqrt(2)}) {$-i\sqrt{2}$};
    
    \fill[black] (0.4,0.35) circle (2pt);
    \node[below right] at (0.4,0.35) {$\tilde a$};
    \fill[black] (-0.4,-0.35) circle (2pt);
    \node[above left] at (-0.4,-0.35) {$-\tilde a$};

    \fill[black] (0,0.4) circle (2pt);
    \node[above left] at (0,0.2) {$i\sqrt\kappa n^{-\frac{\alpha}2}$};
    \fill[black] (0,-0.4) circle (2pt);
    \node[below right] at (0,-0.2) {$-i\sqrt\kappa n^{-\frac{\alpha}2}$};

    \fill[black] (0,0) circle (2pt);
    \node[below right] at (0,0) {0};

\end{tikzpicture}
\end{minipage}
\caption{On the left: The contour $\gamma_i$ for the single integral representation in Proposition \ref{prop:alternativeSingleInt}. The integrand has a jump on $(-i\infty,-i\sqrt 2]\cup [i\sqrt 2, i\infty)$. On the right: The contour (the deformation of $\gamma_i$) for the single integral representation \eqref{eq:IntRepTildeF} suitable for the weak non-Hermiticity regime. The direction of the path close to $t=0$ is chosen such that it passes through $\tilde a$ and $-\tilde a$ in the allowed sectors. As $n\to\infty$ we have $\pm\tilde a\to 0$. Furthermore, the image of the contour from $-\infty$ to $\infty$ is invariant under a rotation by $\pi$ radians. The remaining contour is the steepest descent path going through the two saddle  points $\tilde b_1, \tilde b_2$ in the upper half-plane that are not close to $0$.\label{Fig4}}
\end{figure}

\begin{proposition} \label{prop:alternativeSingleInt}
    We have
    \begin{align} \label{eq:IntRepTildeF}
        I_n(\tau, z; w) = -\frac{1}{2\pi i} \oint_{\gamma_i} \frac{e^{n \tilde F(t)}}{t^2+1-\tau} \frac{dt}{\sqrt{2+t^2}},
    \end{align}
    where $\gamma_i$ is a small loop around $i$ with positive orientation (see Figure \ref{Fig4}), $\sqrt{2+t^2}$ has cut $(-i\infty, -i\sqrt 2]\cup [i\sqrt 2, i\infty)$ and $\tilde F(t) = \tilde F(\tau; z, w; t)$ is defined as
    \begin{align} \label{eq:defTildeF}
        \tilde F(t) = \frac{z \overline{w}}{\tau} - \frac{(z+\overline{w})^2}{4\tau(2+t^2)}-\frac{(z-\overline{w})^2}{4\tau t^2}-\log(1+t^2)+\log\tau.
    \end{align} 
    The saddle points of $\tilde F$ are given by $\pm\sqrt{a-1}, \pm \sqrt{a^{-1}-1}, \pm \sqrt{b-1}$ and $\pm \sqrt{b^{-1}-1}$ with $a$ and $b$ as defined in \eqref{eq:defab}, where the degeneracies are determined by Proposition \ref{prop:saddlePointsDef}. When $z=w=0$ there is one saddle point $t=0$ which is simple.
\end{proposition}

\begin{proof}
    This follows by a substitution $s\mapsto 1+t^2$. We simply have $\tilde F(t)=F(1+t^2)$. Due to the degeneracy of the square root there are several options for $\gamma_i$, but it is clear that the choice as stated works. Since $\tilde F'(t)=2t  F'(1+t^2)$ the saddle points and the degeneracies readily follow. The only extra case is when $z=w=0$. In that case we have $\tilde F(t)=-\log(1+t^2)+\log \tau$, and thus there is a simple saddle point $t=0 = \sqrt{a-1}$. That $\gamma_i$ is a \textit{small} loop around $i$ assures that it does not enclose the singularities $\pm i\sqrt{1-\tau}$ and $\pm i\sqrt 2$.  
\end{proof}

\begin{corollary} \label{cor:corSaddles}
 Let $\delta>0$ be fixed and consider a set of $z, w$ such that $\eta, \eta'\in (-\pi-\delta,-\delta)\cup (\delta, \pi-\delta)$.\\
    The function $\tilde F(\tau; z, w; t)$ has four saddle points $\tilde a_k$ that behave as
    \begin{align*}
        \tilde a_k = i^k (1+i) \sqrt{\frac{z-\overline{w}}{2\sqrt\tau |\sin\eta_\mp|}}+\mathcal O(\xi_+^2+\eta_\pm^2), \qquad k=0,1,2,3,
    \end{align*}
    when $\xi_+\to 0$ and $\eta_\pm\to 0$, and the remaining saddle points (if any) are bounded away from $0$. 
\end{corollary}

Note that the integral representation in \eqref{eq:IntRepTildeF} is valid for any $\tau$ and any $z, w\in\mathbb C$. However, in the weak non-Hermiticity regime the following deformation makes sense (see Figure \ref{Fig4} on the right). 
We deform $\gamma_i$ to a collection of two contours, one contour $\tilde\gamma$ from $-\infty$ to $+\infty$, and one contour that is the steepest descent contour through the two saddle points, let's call them $\tilde b_1$ and $\tilde b_2$, in the upper half-plane that are not close to $t=0$ (i.e., not $\tilde a_k$). This steepest descent path is the same as the one from Figure 7 in \cite{ADM}, after a substitution $s=1+t^2 $. The deformation is allowed since $e^{n \tilde F(t)}$ is uniformly bounded for, say, $|t|\geq 2$. Four of the saddle points are close to $t=0$ as in Corollary \ref{cor:corSaddles}. We will focus on these four first, they are important for the integral from $-\infty$ to $\infty$. We may deform $\tilde\gamma$ and pass it through the origin. Note that, around $t=0$, $e^{n \tilde F(t)}$ blows up in the two sectors where $\arg(t)-\arg(z-\overline{w})\in (\frac{\pi}{4},\frac{3\pi}{4})\cup (-\frac{3\pi}{4}, -\frac{\pi}{4})$, but in the other two sectors it behaves nicely. Each of these four sectors contains exactly one of the saddle points $\tilde a_k$. In fact, the saddle point in the allowed sectors are $\tilde a_0$ and $\tilde a_2$ when $\Im(z-\overline{w})>0$ and $\tilde a_1$ and $\tilde a_3$ when $\Im(z-\overline{w})<0$. 

Let us now denote 
\begin{align*}
    \tilde a = \begin{cases}
        \tilde a_0, & \arg(z-\overline{w})\in [0,\frac{\pi}{2}],\\
        \tilde a_1, & \arg(z-\overline{w})\in [-\pi,-\frac{\pi}{2}),\\
        \tilde a_2, & \arg(z-\overline{w})\in (\frac{\pi}{2}, \pi],\\
        \tilde a_3, & \arg(z-\overline{w})\in [-\frac{\pi}{2},\pi).
    \end{cases}
\end{align*}
Then $\tilde a$ is a saddle point in an allowed sector that furthermore is in the right half-plane. 
Now we deform $\tilde\gamma$ in such a way that it passes through $-\tilde a, 0$ and $\tilde a$, in an allowed direction near the origin.  
We may perform the deformation in such a way that the contour is invariant under the substitution $t\mapsto -t$. See Figure \ref{Fig4} for a graphic description of the resulting deformation. We are thus tasked with understanding the integral
\begin{align} \label{eq:inteG}
    \int_0^\infty \frac{e^{n \tilde F(t)}}{t^2+\kappa n^{-\alpha}} \frac{dt}{\sqrt{2+t^2}},
\end{align}
where the contour passes through the correct $\tilde a$ (in the direction of steepest descent without loss of generality). In what follows, we will assume that $\arg(z-\overline{w})\in[0,\frac{\pi}{2}]$. This is allowed due to the symmetries
\begin{align*}
    \mathcal K_n(-z,-w) = \mathcal K_n(z, w) = \overline{\mathcal K_n(\overline z, \overline{w})}.
\end{align*}
Thus we effectively assume that
\begin{align} \label{eq:tildeabehavAssump}
        \tilde a = (1+i) \sqrt{\frac{z-\overline{w}}{2\sqrt\tau |\sin\eta_\mp|}}+\mathcal O(\xi_+^2+\eta_\pm^2)
    \end{align}
as $\xi_+\to 0$ and $\eta_\pm\to 0$.

Note that it is not necessarily implied in the following proposition that $\tau\uparrow 1$. Note that in the case that $\tau$ is fixed the condition on $\eta, \eta'$ can be omitted. The result indicates that there is a microscopic edge scaling limit near the collapsing ellipse boundary in the weak non-Hermiticity regimes that we consider. In particular, assuming that $(1-\tau)n\to\infty$ and $\tau\to 1$ as $n\to\infty$, one extracts from Proposition \ref{prop:behavEdgeClose} below that
\begin{multline*}
    \lim_{n\to\infty} \overline{C_\tau^n(z, w)} \frac{1-\tau^2}{n} \mathcal K_n\left(2\sqrt\tau\cosh\left(\xi_\tau+i\eta+\sqrt\frac{1-\tau^2}{\tau n}\frac{u}{4|\sin \eta|}\right), 2\sqrt\tau\cosh\left(\xi_\tau+i\eta+\sqrt\frac{1-\tau^2}{\tau n} \frac{v}{4|\sin \eta|}\right)\right)\\
    = \erfc\left(\frac{u+\overline{v}}{\sqrt 2}\right) \exp\left(u \overline v-\frac{|u|^2+|v|^2}2\right), \quad u, v\in\mathbb C. 
\end{multline*}
when $\eta\in (-\pi,0)\cup (0, \pi)$.

\begin{proposition} \label{prop:behavEdgeClose}
     Let $\tau\in(0,1)$ and assume that $(1-\tau)n\to\infty$ as $n\to\infty$.\\
     Suppose that $\eta, \eta'\in(-\pi+\delta,-\delta)\cup (\delta, \pi-\delta)$ where $\delta>0$ is fixed. Then there exists a constant $c>0$ such that
     \begin{multline} \label{eq:weakNonHerfcResult}
        \mathcal K_n(z, w) =   \erfc\left(\sqrt{8n\tau\sin(\eta'-i\xi_\tau) \sin(\eta+i\xi_\tau)} \frac{\xi_+-\xi_\tau+i\eta_-}{\sqrt{1-\tau^2}}\right)
        C_\tau^n(z, w)\frac{n}{1-\tau^2} e^{-n \frac{|z-w|^2}{2(1-\tau^2)}}  \\
        + e^{-n (\xi-\xi_\tau)^2 g(\xi+i\eta)} e^{-n (\xi'-\xi_\tau)^2 g(\xi'+i\eta')}\left(\mathcal O(1-\tau)+\mathcal O\Big(\sqrt\frac{n}{1-\tau}\Big)\right),
    \end{multline}
    as $n\to\infty$, uniformly for $|\xi_+-\xi_\tau+i\eta_-|\leq c(1-\tau)$, where $C_\tau^n(z, w)$ is a unimodular factor.     
\end{proposition}

\begin{proof}
We start with the single integral representation \eqref{eq:inteG} and the deformation of $\gamma_i$ as described above, and depicted in Figure \ref{Fig4}. By symmetry may ignore the part of the contour from $-\infty$ to $0$.
    In this setting we have that $\arg(z-\overline{w})\in [0,\frac{\pi}{2}]$. We have explicitly that $\tilde a=\sqrt{a^{-1}-1}$ (this is true even when $\tau$ does not tend to $1$) and thus
\begin{align} \label{eq:tildeaxieta}
    \tilde a &= i\sqrt{1-\tau} + i\tau\frac{\xi_+-\xi_\tau+i\eta_-}{\sqrt{1-\tau}}+\mathcal O((1-\tau)^{-\frac{3}{2}}|\xi_+-\xi_\tau+i\eta_-|^2)
    = i\sqrt{1-\tau}+\mathcal O(c\sqrt{1-\tau})
\end{align}
as $n\to\infty$, uniformly for $|\xi_+-\xi_\tau+i\eta_-|\leq c(1-\tau)$. If we pick $c>0$ small enough, then $\tilde a/\sqrt{1-\tau^2}$ will be bounded away from $0$. The dominant contribution comes from \eqref{eq:inteG}. We start by noticing that
\begin{align} \label{eq:IntDecompEdge}
    \int_0^\infty \frac{e^{n \tilde F(t)}}{t^2+1-\tau} \frac{dt}{\sqrt{2+t^2}}
    = \frac{1}{2i\sqrt{1-\tau}} \left(\int_0^\infty \frac{e^{n \tilde F(t)}}{t-i\sqrt{1-\tau}} \frac{dt}{\sqrt{2+t^2}}-\int_0^\infty \frac{e^{n \tilde F(t)}}{t+i\sqrt{1-\tau}} \frac{dt}{\sqrt{2+t^2}}\right).
\end{align}
In the first integral on the right-hand side a saddle point and pole coalesce and this can be solved asymptotically with the method in \cite{Mo}. Namely, one has to transform the exponent to a quadratic function using a properly chosen conformal map. To construct this conformal map we first observe 
    using Proposition \ref{prop:F''a-1} that
\begin{align*}
    \tilde F''(\tilde a) = 4 \tilde a^2 F''(1+\tilde a^2)
    = -\frac{16\tau}{1+\tau} \sin(\eta'-i\xi_\tau) \sin(\eta+i\xi_\tau) (1+o(1-\tau))
\end{align*}
as $n\to\infty$, which is close to being a negative number. In a small neighborhood of the saddle point $\tilde a$ we define the conformal map $\phi_n$ through the relation
\begin{align*}
    \phi_n(t)^2 = \frac{\tilde F(\tilde a)-\tilde F(\tilde a+t\sqrt{1-\tau})}{1-\tau},
\end{align*}
and the requirement
\begin{align*}
    \phi_n'(0) = \sqrt{\frac{2}{1+\tau}} \sqrt{4\tau\sin(\eta'-i\xi_\tau) \sin(\eta+i\xi_\tau)}+\mathcal O(n^{-\frac{1+\alpha}{2}+\nu})
\end{align*}
as $n\to\infty$. Note that $\phi_n^{(m)}(0)=\mathcal O(1)$ for all $m$ which can easily be seen with \eqref{eq:defTildeF}. Let $\delta$ be a complex number in the direction of steepest descent. Up to an exponentially  small error (and assuming $\delta$ is small enough), the first integral in \eqref{eq:IntDecompEdge} can be rewritten as
\begin{align*}
    \int_{\tilde a-\sqrt{1-\tau}\delta}^{\tilde a+\sqrt{1-\tau}\delta} \frac{e^{n \tilde F(t)}}{t-i\sqrt{1-\tau}} \frac{dt}{\sqrt{2+t^2}}
    &= e^{n \tilde F(\tilde a)}\int_{\phi_n(-\delta)}^{\phi_n(\delta)} \frac{e^{-(1-\tau)n t^2}}{\phi_n^{-1}(t)-i +\tilde a/\sqrt{1-\tau}} \frac{1}{\phi_n'(\phi_n^{-1}(t))} \frac{dt}{\sqrt{2+(1-\tau)\phi_n^{-1}(t)^2}}\\
    &= e^{n \tilde F(\tilde a)}\int_{\phi_n(-\delta)}^{\phi_n(\delta)} \frac{e^{-(1-\tau) n t^2}}{t-\phi_n(i -\tilde a/\sqrt{1-\tau})} h_n(t) dt,
\end{align*}
where
\begin{align*}
    h_n(t)=\frac{1}{\phi_n'(\phi_n^{-1}(t))} 
    \frac{t-\phi_n(i- \tilde a/\sqrt{1-\tau})}{\phi_n^{-1}(t)-i + \tilde a/\sqrt{1-\tau}} \frac{1}{\sqrt{2+(1-\tau)\phi_n^{-1}(t)^2}}.
\end{align*}
Notice that $h_n$ is an analytic function in a neighborhood of $0$ with bounded derivatives and $h_n(0)=\frac{1}{\sqrt 2}$. The integration bounds necessarily satisfy $\phi_n(\delta)>0$ and $\phi_n(-\delta)<0$. According to \cite{Mo} the  behavior of this integral will be given by
\begin{multline*}
    -\frac{\pi i}{\sqrt 2}e^{n \tilde F(i\sqrt{1-\tau})} \erfc\left(i\sqrt{(1-\tau)n} \, \phi_n(i- \tilde a/\sqrt{1-\tau})\right) 
    + \sqrt\frac{\pi}{(1-\tau) n} e^{n \tilde F(\tilde a)} h'\Big(\phi_n(i - \tilde a/\sqrt{1-\tau})\Big) \left(1+\mathcal O\Big(\frac{1}{\sqrt{(1-\tau) n}}\Big)\right)\\
    =-\frac{\pi i}{\sqrt 2}e^{n F(\tau)} \erfc\left(i\sqrt{(1-\tau)n} \, \phi_n(i- \tilde a/\sqrt{1-\tau})\right) 
    + \sqrt\frac{\pi}{(1-\tau) n} e^{n F(a^{-1})} h_n'(0)\left(1+\mathcal O\Big(\frac{1}{\sqrt{(1-\tau) n}}\Big)\right).
\end{multline*}
The conditions of Proposition 3.1 in \cite{Mo} require that $\phi_n(-\delta)<\Re(\phi_n(i-\tilde a/\sqrt{1-\tau}))<\phi_n(\delta)$, which is indeed satisfied under the condition $|\xi_+-\xi_\tau+i\eta_-|\leq c n^{-\alpha}$ (provided that $c$ is small enough), since then the term in the middle is small (and the integration bounds are of order $1$). We conclude that
\begin{align*}
    \int_0^\infty \frac{e^{n \tilde F(t)}}{t^2+\kappa n^{-\alpha}} \frac{dt}{\sqrt{2+t^2}}
    = -\frac{\pi i}{\sqrt 2}e^{n F(\tau)} \erfc\left(i\sqrt{(1-\tau)n} \, \phi_n(i- \tilde a/\sqrt{1-\tau})\right) 
   + \sqrt\pi e^{n F(a^{-1})} \mathcal O\Big(\frac{1}{\sqrt{(1-\tau) n}}\Big)
\end{align*}
as $n\to\infty$, uniformly for $\xi_+-\xi_\tau+i\eta_-=\mathcal O(n^{-\frac{1+\alpha}{2}+\nu})$. We notice that
\begin{align*}    
\sqrt{(1-\tau)n} \, \phi_n(i- \tilde a/\sqrt{1-\tau})
= \frac{1}{\sqrt{1+\tau}} \sqrt{8n\tau\sin(\eta'-i\xi_\tau) \sin(\eta+i\xi_\tau)}(i\sqrt{1-\tau}- \tilde a)+\mathcal O(1-\tau).
\end{align*}
Plugging this behavior in the error functions together with \eqref{eq:tildeaxieta}, and reinstating the weight factors of the kernel, we obtain the result, although we need to argue that the contribution to the integral of the steepest descent contour through $\tilde b_1$ and $\tilde b_2$ is negligible. It is clear that the corresponding steepest descent contributions are not dominant since the real part of $F$ in this saddle points is smaller than $\Re F(\tilde a)$ when $\eta, \eta'$ are in compact subsets of $(-\pi,0)\cup (0,\pi)$, by Proposition \ref{prop:Finsaddleelliptic}.
\end{proof}

\begin{remark}
    The condition that $\eta, \eta'$ are in compact subsets of $(-\pi,0)\cup (0,\pi)$ is to avoid the situation where $z, w$ are close to $\pm 2$ (with same sign). Since we assume that $f$ has compact support, the expression $(f(n^\gamma z)-f(n^\gamma w))^2$ will be zero for such points, and this complicated case thus does not play a role in determining the variance. Note that for $\alpha=\gamma=\frac{1}{3}$ it is known that the kernel approaches the weak non-Hermiticity Airy kernel \cite{Bender}.
\end{remark}

Note that the result can be rewritten in terms of $z$ and $w$. One may use \eqref{eq:tildeabehavAssump} for example, and one may notice that
\begin{align*}
    \sqrt{4\tau\sin(\eta'-i\xi_\tau) \sin(\eta+i\xi_\tau)}
    = (4\tau-z^2)^\frac{1}{4} (4\tau-w^2)^\frac{1}{4}.
\end{align*}
However, the result in elliptic coordinates will turn out to be most convenient once we determine the variance by integration.

Note that for $\alpha=\gamma=1$ also the second term on the right-hand side of \eqref{eq:IntDecompEdge} plays a role and in the limit we get a sum of two rather than one complementary error functions. This yields a 
different result, namely, we obtain the weak non-Hermiticity kernel.

\begin{proposition} \label{prop:behavKerWeakNotFaddeeva}
Let $\tau\in[0,1)$ and assume that $\tau\to 1$ and $(1-\tau)n\to\infty$ as $n\to\infty$. Then we have
    \begin{multline*} 
        \mathcal K_n(z, w) = \frac{n \mathfrak{1}_{\xi_+<\xi_\tau}}{\pi(1-\tau^2)} \exp\left(-\frac{n}{1-\tau^2} \frac{|z-w|^2}{2}\right) C_\tau^n(z, w)\\
        + \sqrt{\frac{n}{32\pi^3\tau(1-\tau^2)}} 
        \frac{e^{-n(\xi-\xi_\tau)^2 g(\xi+i\eta))}e^{-n(\xi-\xi_\tau)^2 g(\xi'+i\eta'))}}{\sinh(\xi_+-\xi_\tau+i\eta_-) \sqrt{\sinh(\xi+i\eta) \sinh(\xi'+i\eta')}} D_\tau^n(z, w)\\
        + e^{-n(\xi-\xi_\tau)^2 g(\xi+i\eta))}e^{-n(\xi-\xi_\tau)^2 g(\xi'+i\eta'))} \mathcal O\Big(\frac{1}{(1-\tau)\sqrt n}\Big)
    \end{multline*}
    as $n\to\infty$, uniformly for $((\xi_+-\xi_\tau)^2+\sin^2\eta_-)^{-2}=\mathcal O((1-\tau)^{-2})$. 
\end{proposition}

\begin{proof}
    We start with the integral in \eqref{eq:inteG}. As in the proof of Proposition \ref{prop:behavEdgeClose} we may use an isomorphism to make the exponent quadratic. The integral to consider then is
    \begin{align*}
        \frac{1}{\sqrt{1-\tau}} e^{n \tilde F(\tilde a)} \int_{\phi_n(-\delta)}^{\phi_n(\delta)} \frac{e^{-(1-\tau) n t^2}}{1+(\phi_n^{-1}(t)+\tilde a/\sqrt{1-\tau})^2} \frac{1}{\phi_n'(\phi_n^{-1}(t))} \frac{dt}{\sqrt{2+(1-\tau)(\phi_n^{-1}(t)+\tilde a/\sqrt{1-\tau})^2}}.
    \end{align*}
    In our present case the assumptions imply that $1+\tilde a^2/(1-\tau)$ is bounded away from $0$. We can therefore apply classical steepest descent arguments without issues. Then we put $t=0$ everywhere in the integrand. The result is
    \begin{multline*}
        e^{n \tilde F(\tilde a)} \sqrt\frac{\pi}{(1-\tau) n} \frac{1}{1+\tilde a^2/(1-\tau)} \frac{1}{\phi_n'(0)}\frac{1}{\sqrt{2+\tilde a^2}} \left(1+\mathcal O\Big(\frac{1}{(1-\tau) n}\Big)\right)\\
        = e^{n F(a^{-1})}  \sqrt\frac{\pi}{(1-\tau)^3 n}  
        \frac{1}{a^{-1}-\tau}  \frac{1}{\sqrt{1-a^{-2}}}\left(1+\mathcal O\Big(\frac{1}{(1-\tau) n}\Big)\right),
    \end{multline*}
    which can be rewritten in terms of elliptic coordinates to produce the term in the theorem.
\end{proof}

For fixed $\tau$, it was argued in \cite{ADM} that $g$ is bounded from below by a positive constant. In that case $e^{-n (\xi-\xi_\tau)^2 g(\xi+i\eta)}$ is thus sharply peaked around $\xi=\xi_\tau$. We give explicit lower bounds in the lemma below.

\begin{lemma} \label{lem:behavg}
    We have uniformly for $\xi\geq 0$, $\eta\in(-\frac{3\pi}{4},-\frac{\pi}{4})\cup (\frac{\pi}{4}, \frac{3\pi}{4})$ and $\tau\in [0,1)$ that
    \begin{align*}
        g(\xi+i\eta) \geq \frac{2\tau}{1-\tau^2}.
    \end{align*}
    We have uniformly for $\xi\geq \xi_\tau$, $\eta\in (-\pi, \pi]$ and $\tau\in [0,1)$ that
    \begin{align*}
        g(\xi+i\eta) \geq \frac{1}{\tau}-1.
    \end{align*}
    There exists an $R>0$ such that
    \begin{align*}
        g(\xi+i\eta) \geq \frac{\tau}{1+\tau} e^{2\xi}
    \end{align*}
    uniformly for $\xi\geq R$ and $\eta\in(-\pi, \pi]$, $\tau\in[0,1)$. 
    Furthermore, we have
    \begin{align*}
        e^{-n(\xi-\xi_\tau)^2g(\xi+i\eta)}=e^{-n\frac{1+\tau^2-2\tau\cos(2\eta)}{1-\tau^2}(\xi-\xi_\tau)^2} \left(1+\mathcal O\Big(\frac{n}{1-\tau}(\xi-\xi_\tau)^3\Big)\right)
    \end{align*}
    as $n\to\infty$ uniformly for $\eta\in(-\pi,\pi]$, $\tau\in[0,1)$ and $\xi-\xi_\tau=\mathcal O(n^{-\frac{1}{3}}(1-\tau)^\frac{1}{3})$.
\end{lemma}

\begin{proof}
Consider the function $G_\eta(\xi)=-(\xi-\xi_\tau)^2 g(\xi+i\eta)$. It was proved in \cite{ADM} that $G_\eta(\xi_\tau)=G_\eta'(\xi_\tau)=0$. Using \eqref{eq:defG} we find that
\begin{align*}
    G_\eta''(\xi) &= 2e^{-2\xi} \cos(2\eta)-\frac{4\tau}{1-\tau^2} \cosh(2\xi) (1-\tau\cos(2\eta))
    \leq -\frac{4\tau}{1-\tau^2},
\end{align*}
where we used that $\cosh(2\xi)\geq 1$ and that $\cos(2\eta)\leq 0$ under the given conditions. Integrating twice from $\xi_\tau$ to $\xi$ will yields the lower bound for $g(\xi+i\eta)$ as stated. The second inequality follows along similar lines. Here we have
\begin{align*}
    G_\eta''(\xi)\leq 2 e^{-2\xi} - \frac{4\tau}{1+\tau} \cosh(2\xi)
    = \frac{2}{1+\tau} (e^{-2\xi}-\tau e^{2\xi}) \leq 2(1-\frac{1}{\tau}),
\end{align*}
under the condition $\xi\geq \xi_\tau$.
Let us now move to the second part of the lemma.
The function $G_\eta(\xi)$ is obviously analytic in its variable $\xi$. Therefore we may Taylor expand it around $\xi=\xi_\tau$ to find
\begin{align*}
    G_\eta(\xi)
    = -g(\xi_\tau+i\eta)(\xi-\xi_\tau)^2+\mathcal O\left(\frac{(\xi-\xi_\tau)^3}{1-\tau}\right)
\end{align*}
as $\xi\to\xi_\tau$, which is uniform in $\eta\in(-\pi,\pi]$ and $\tau\in [0,1)$. It was explained in \cite{ADM} why there are no constant and linear term. Furthermore, it was shown that $g(\xi_\tau+i\eta) = \frac{1+\tau^2-2\tau\cos(2\eta)}{1-\tau^2}$. 
\end{proof}

\section{Smooth linear statistics in the strong non-Hermiticity regime} \label{sec:sec1}

Let $0<\tau<1$ be fixed. In this section we shall prove Theorem \ref{thm:CLTfixedTau}. As mentioned, it is known that there is a CLT here \cite{AmHeMa}, our goal is however to prove the exact formulae for the limiting variance. We consider smooth linear statistics
\begin{align}
\mathfrak X_n(f) = \sum_{j=1}^n f(z_j),
\end{align}
where $z_1, \ldots, z_n\in\mathbb C$ are eigenvalues picked from the elliptic Ginibre ensemble with fixed parameter $0<\tau<1$. 
Here $f$ is a differentiable test function that satisfies the growth condition $f(z)=\mathcal O(e^{C|z|^2})$ as $|z|\to\infty$, for an arbitrary constant $C>0$. We shall assume that $f$ is real-valued. It is a known general fact that the variance of $\mathfrak X_n(f)$ is given by \eqref{eq:varIntKer}.
We do not have to perform the method of steepest descent for \eqref{eq:defIn} in this case to understand the asymptotic behavior of the correlation kernel, because it has already been carried out in \cite{ADM} and \cite{Mo}. The large $n$ behavior of the correlation kernel is qualitatively different depending on what part of $\mathbb C^2$ we are considering (whether points $z, w$ are close to the boundary and/or close to each other). We therefore have to divide $\mathbb C^2$ into appropriate regions to integrate over. We remind the reader of the notations
\begin{align*}
    \xi_\pm = \frac{\xi\pm \xi'}2, \quad \eta_\pm = \frac{\eta\pm \eta'}2, \quad \xi_\tau=-\frac12\log \tau.
\end{align*}

\begin{lemma} \label{lem:DIn}
Let $0<\nu<\frac{1}{6}$, and let $D^I_n$ be the set of $(z, w)\in \mathbb C^2$ such that $|\xi_+-\xi_\tau| \geq n^{-\frac{1}{2}+\nu}$.
Then
\begin{align*}
\lim_{n\to\infty} \frac{1}{2} \iint_{D^I_n} (f(z)-f(w))^2 |\mathcal K_n(z, w)|^2 d^2z d^2w
= \frac{1}{4\pi^2} \int_{\mathcal E_\tau} |\nabla f(z)|^2 d^2 z. 
\end{align*} 
\end{lemma}

\begin{proof}
It follows from Proposition \ref{prop:kernelIneq} (or Theorem III.5 in \cite{ADM}) that 
\begin{align*}
\left|\mathcal K_n(z, w)\right|^2 = \frac{n^2}{\pi^2 (1-\tau^2)^2} e^{-n \frac{|z-w|^2}{1-\tau^2}} \mathfrak{1}_{\xi_+<\xi_\tau}
+ \mathcal O\left(n^{3-2\nu} e^{- n (\xi-\xi_\tau)^2g(\xi+i\eta)} e^{- n (\xi'-\xi_\tau)^2g(\xi'+i\eta')}\right),
\end{align*}
uniformly for $(z, w)\in D^I_n$ as $n\to\infty$. In fact, for some constant $c>0$, we may write that
\begin{align*}
\left|\mathcal K_n(z, w)\right|^2 = \frac{n^2}{\pi^2 (1-\tau^2)^2} e^{-n \frac{|z-w|^2}{1-\tau^2}} \mathfrak{1}_{\xi<\xi_\tau+n^{-\frac{1}{2}+\nu}} \mathfrak{1}_{\xi'<\xi_\tau+n^{-\frac{1}{2}+\nu}}
+ \mathcal O(e^{-c n^{2\nu}} \mathfrak 1_{\xi_+<\xi_\tau})
+ \mathcal O\left(n^{3-2\nu} e^{- n (\xi-\xi_\tau)^2g(\xi+i\eta)} e^{- n (\xi'-\xi_\tau)^2g(\xi'+i\eta')}\right),
\end{align*}
because, when either $\xi\geq \xi_\tau+n^{-\frac{1}{2}+\nu}$ or $\xi'\geq \xi_\tau+n^{-\frac{1}{2}+\nu}$, we can only have $\xi_+<\xi_\tau$ when $|z-w|\geq \text{const.} \times n^{-\frac{1}{2}+\nu}$. It was proved in \cite{ADM} that $g$ is bounded from below by a positive constant. We infer that
\begin{align*}
\left|\mathcal K_n(z, w)\right|^2 = \mathcal O(e^{-\frac{1}{2} c n^{2\nu}}e^{-n d( (\xi-\xi_\tau)^2+(\xi'-\xi_\tau)^2)})
+ \mathcal O(e^{-c n^{2\nu}} \mathfrak 1_{\xi_+<\xi_\tau}),
\end{align*}
uniformly for $\xi>\xi_\tau+n^{-\frac12+\nu}$ or $\xi'>\xi_\tau+n^{-\frac12+\nu}$ as $n\to\infty$, for some constants $c, d>0$ (where $c$ is possibly different from before). 
Thus we find as $n\to\infty$ that
\begin{align*}
\frac{1}{2} \iint_{D^I_n} (f(z)-f(w))^2 |\mathcal K_n(z, w)|^2 d^2z d^2w
= \frac{1}{2} \iint_{0\leq \xi, \xi'<\xi_\tau+n^{-\frac{1}{2}+\nu}} (f(z)-f(w))^2 \frac{n^2}{\pi^2 (1-\tau^2)^2} e^{-n \frac{|z-w|^2}{1-\tau^2}} d^2z d^2w 
+ \mathcal O\left(e^{-c n^{2\nu}}\right).
\end{align*}
We used the growth condition $|f(z)|\leq e^{C\log^2|z|}$ as $|z|\to\infty$ here, which implies that $f(z)=e^{\mathcal O(\xi^2)}$ as $\xi\to\infty$.
Now let $\delta>0$ be arbitrary. Then for $n> \delta^{-\frac{1}{\frac{1}{2}-\nu}}$, we have
\begin{align*}
\iint_{0\leq \xi, \xi'<\xi_\tau+\delta} (f(z)-f(w))^2 e^{-n \frac{|z-w|^2}{1-\tau^2}} d^2z d^2w
&\geq \iint_{0\leq \xi, \xi'< \xi_\tau+n^{-\frac{1}{2}+\nu}} (f(z)-f(w))^2 e^{-n \frac{|z-w|^2}{1-\tau^2}} d^2z d^2w\\
&\geq \iint_{\mathcal E_\tau \times \mathcal E_\tau} (f(z)-f(w))^2 e^{-n \frac{|z-w|^2}{1-\tau^2}} d^2z d^2w .
\end{align*}
A linear approximation of $f(z)-f(w)$ at $w$ combined with dominated convergence, and carrying out the Gaussian integrals we get, for any $\delta> 0$
\begin{align*}
n^2 \iint_{0<\xi,\xi'<\xi_\tau+\delta} (f(z)-f(w))^2 e^{-n \frac{|z-w|^2}{1-\tau^2}} d^2z d^2w
= \frac{\pi}{4} (1-\tau^2)^2 \int_{0<\xi<\xi_\tau+\delta} |\nabla f(z)|^2 d^2z +  o(1). 
\end{align*}
We infer that there exists a constant $M>0$, independent of $n$ and $\delta$, such that
\begin{align*}
\left|\frac{n^2}{\pi^2 (1-\tau^2)^2}\iint_{0<\xi,\xi'<\xi_\tau+n^{-\frac{1}{2}+\nu}} (f(z)-f(w))^2 e^{-n \frac{|z-w|^2}{1-\tau^2}} d^2z d^2w
- \frac{1}{4\pi}\int_{0<\xi<\xi_\tau} |\nabla f(z)|^2 d^2z\right|
\leq M \left(\delta+o(1)\right)
\end{align*}
for $n> \delta^{-\frac{1}{\frac{1}{2}-\nu}}$, and we are done. 
\end{proof}

\begin{lemma} \label{lem:DIIn}
Let $0<\nu<\frac{1}{6}$, and let $D^{II}_n$ be the set of $z, w\in \mathbb C$ such that $|\xi-\xi_\tau|, |\xi'-\xi_\tau| \leq n^{-\frac{1}{2}+\nu} \leq \frac{1-\tau}{4}|\sin \eta_-|$.
Then
\begin{align*}
\lim_{n\to\infty} \frac{1}{2} \iint_{D^{II}_n} (f(z)-f(w))^2 |\mathcal K_n(z, w)|^2 d^2z d^2w
= \frac{1}{32\pi^2}  
\int_{-\pi}^\pi \int_{-\pi}^\pi 
\left(\frac{f(2\sqrt \tau \cosh(\xi_\tau+i\eta)) - f(2\sqrt\tau \cosh(\xi_\tau+i\eta'))}{\sin \eta_-}\right)^2 d\eta d\eta'.
\end{align*} 
\end{lemma}

\begin{proof}
First, we note that, for $(z, w)\in D^{II}_n$, we have by Remark III.4 in \cite{ADM} that for large enough $n$ 
\begin{align*}
& \Re(F(a^{-1}) - F(\tau)) = \frac{1}{2} \frac{|z-w|^2}{1-\tau^2} - (\xi-\xi_\tau)^2 g(\xi+i\eta) - (\xi'-\xi_\tau)^2 g(\xi'+i\eta')\\
&=4\tau \frac{|\sinh(\xi_++i\eta_+)|^2|\sinh(\xi_-+i\eta_-)|^2}{1-\tau^2}
- (\xi-\xi_\tau)^2  \frac{1+\tau^2 - 2\tau\cos 2\eta}{1-\tau^2}
- (\xi'-\xi_\tau)^2  \frac{1+\tau^2 - 2\tau\cos 2\eta'}{1-\tau^2}
+ \mathcal O(n^{-\frac{3}{2}+3\nu})\\
&\geq 4\tau  \frac{\sinh^2\xi_+ \sin^2 \eta_-}{1-\tau^2}
- (\xi-\xi_\tau)^2  \frac{1+\tau^2 - 2\tau\cos 2\eta}{1-\tau^2}
- (\xi'-\xi_\tau)^2  \frac{1+\tau^2 - 2\tau\cos 2\eta'}{1-\tau^2}
+ \mathcal O(n^{-\frac{3}{2}+3\nu})\\
&\geq \frac{1}{1-\tau^2}\left(((1-\tau)^2+\mathcal O(n^{-\frac12+\nu}))\sin^2 \eta_- - 4 (\xi-\xi_\tau)^2 - 4 (\xi'-\xi_\tau)^2\right) + \mathcal O(n^{-\frac{3}{2}+3\nu}),
\end{align*}
which is positive for $n$ large enough, under the conditions of the lemma. Hence the edge term dominates over the bulk term, and we have a constant $c>0$ such that
\begin{align*}
\left|\mathcal K_n(z, w)\right|^2 
= \frac{n}{32\pi^3 \tau (1-\tau^2)} 
\frac{e^{-2n(\xi-\xi_\tau)^2 g(\xi+i\eta)} e^{-2n(\xi'-\xi_\tau)^2 g(\xi'+i\eta')}}
{\left|\sinh\left(\xi_+-\xi_\tau+i\eta_-\right)\right|^2 |\sinh(\xi+i\eta)| |\sinh(\xi'+i\eta')|}
 \left(1+ \mathcal O\left(n^{-\frac{1}{2}+3\nu}\right)\right)
 + \mathcal O\left(e^{-c n^{2\nu}}\right),
\end{align*}
uniformly for $(z, w)\in D^{II}_n$, by Theorem I.1 combined with Remark I.2 from \cite{ADM}. Next we do a substitution of the integration variables, and we find
\begin{align*}
&\frac{1}{2}\iint_{D^{II}_n} (f(z)-f(w))^2 |\mathcal K_n(z, w)|^2 d^2z d^2w
\nonumber\\
&= \frac{n (1+\mathcal O(n^{-\frac{1}{2}+\nu}))}{64\pi^3 \tau (1-\tau^2)} \iint \int_{\xi_\tau-n^{-\frac{1}{2}+\nu}}^{\xi_\tau+n^{-\frac{1}{2}+\nu}} \int_{\xi_\tau-n^{-\frac{1}{2}+\nu}}^{\xi_\tau+n^{-\frac{1}{2}+\nu}} 
\left(f(2\sqrt \tau \cosh(\xi+i\eta)) - f(2\sqrt\tau \cosh(\xi'+i\eta'))\right)^2 
\nonumber\\ 
&\quad\times
\frac{e^{-2n(\xi-\xi_\tau)^2 g(\xi+i\eta)} e^{-2n(\xi'-\xi_\tau)^2 g(\xi'+i\eta')}}
{\left|\sinh\left(\xi_+-\xi_\tau+i\eta_-\right)\right|^2  |\sinh(\xi+i\eta)| |\sinh(\xi'+i\eta')|}
16\tau^2
|\sinh(\xi+i\eta)|^2 |\sinh(\xi'+i\eta')|^2 d\xi d\xi' d\eta d\eta',
\end{align*}
where the integration over $d\eta d\eta'$ respects $|\sin \eta_-|\geq \frac{1-\tau}{4} n^{-\frac{1}{2}+\nu}$. This integral we can write as
\begin{multline*}
\frac{\tau n^{2\nu}}{4\pi^3 (1-\tau^2)} \iint \int_{-1}^1 \int_{-1}^1 
\left(f(2\sqrt \tau \cosh(\xi_\tau+i\eta)) - f(2\sqrt\tau \cosh(\xi_\tau+i\eta'))+\mathcal O(n^{-\frac{1}{2}+\nu})\right)^2 \\ \times
\frac{e^{-2n^{2\nu} \xi^2 g(\xi_\tau+i\eta)} e^{-2n^{2\nu} \xi'^2  g(\xi_\tau+i\eta')} (1+\mathcal O(n^{-\frac{1}{2}+3\nu}))} 
{n^{-1+2\nu} (\xi+\xi')^2 + \sin^2\eta_-}
|\sinh(\xi_\tau+i\eta)| |\sinh(\xi_\tau+i\eta')| d\xi d\xi' d\eta d\eta',
\end{multline*}
and by Laplace's method (and using a geometric series expansion), combined with the fact that $$g(\xi_\tau+i\eta)=\frac{1+\tau^2-2\tau \cos(2\eta)}{1-\tau^2}
= \frac{4\tau}{1-\tau^2} |\sinh(\xi_\tau+i\eta)|^2$$ (see \eqref{eq:defG}), this gives
\begin{multline*}
\frac{1}{32\pi^2} (1+\mathcal O(n^{-\frac{1}{2}+3\nu})+\mathcal O(n^{-\nu}))  
\iint_{|\sin \eta_-|\geq \frac{1-\tau}{4} n^{-\frac{1}{2}+\nu}} 
\left(\frac{f(2\sqrt \tau \cosh(\xi_\tau+i\eta)) - f(2\sqrt\tau \cosh(\xi_\tau+i\eta'))}{\sin \eta_-}\right)^2 d\eta d\eta'\\
= \frac{1}{32\pi^2} (1+\mathcal O(n^{-\frac{1}{2}+3\nu})+\mathcal O(n^{-\nu}))  
\int_{-\pi}^\pi \int_{-\pi}^\pi 
\left(\frac{f(2\sqrt \tau \cosh(\xi_\tau+i\eta)) - f(2\sqrt\tau \cosh(\xi_\tau+i\eta'))}{\sin \eta_-}\right)^2 d\eta d\eta'.
\end{multline*}
\end{proof}

\begin{lemma} \label{lem:DIIIn}
Let $0<\nu<\frac{1}{6}$, and let $D^{III}_n$ be the set of $(z, w)\in \mathbb C^2$ such that $(\xi_+-\xi_\tau)^2 + \frac{(1-\tau)^2}{16} \sin^2 \eta_- \leq n^{-1+2\nu}$. Then we have
\begin{align*}
\lim_{n\to\infty}\frac{1}{2} \iint_{D^{III}_n} (f(z)-f(w))^2 |\mathcal K_n(z, w)|^2 d^2z d^2w
= 0.
\end{align*} 
\end{lemma}

\begin{proof}
From Proposition 5.1 in \cite{Mo}, we can deduce that
uniformly for $(z,w)\in D^{III}_n$ as $n\to\infty$
\begin{align*}
\left|\mathcal K_n\left(z, w\right)\right|^2 =
\mathcal O( n^2 e^{-n \frac{|z-w|^2}{1-\tau^2}}) 
+\mathcal O(n^{1+4\nu} e^{-2 n(\xi-\xi_\tau)^2 g(\xi+i\eta)} e^{-2 n(\xi'-\xi_\tau)^2 g(\xi'+i\eta')}).
\end{align*}
This follows by writing the complementary error function as an asymptotic series in the different regimes. Namely,  by the NIST handbook \cite[7.12.1]{NIST}, we have 
 \begin{align*}
 \operatorname{erfc}(z) &\sim \frac{e^{-z^2}}{\sqrt \pi} \sum_{k=0}^\infty (-1)^k \frac{(\frac{1}{2})_k}{z^{2k+1}},\\
 \operatorname{erfc}(-z) &\sim 2 - \frac{e^{-z^2}}{\sqrt \pi} \sum_{k=0}^\infty (-1)^k \frac{(\frac{1}{2})_k}{z^{2k+1}},
 \end{align*}
as $z\to\infty$, both valid for $|\arg z|\leq \frac{3}{4}\pi-\delta$, for any fixed $\delta>0$. 
Let $D_b$ be a bounded domain that contains $n^{\frac{1}{2}-\nu} D^{III}_n$. Then as $n\to\infty$, we have
\begin{align*}
\iint_{D^{III}_n} (f(z)-f(w))^2 n^2 e^{-n \frac{|z-w|^2}{1-\tau^2}} \mathfrak{1}_{\xi_+<\xi_\tau} d^2z d^2w
&\leq \iint_{D_b} (f(n^{-\frac{1}{2}+\nu} z)-f(n^{-\frac{1}{2}+\nu} w))^2 n^{1+2\nu} e^{-n^{2\nu} \frac{|z-w|^2}{1-\tau^2}} d^2z d^2w\\
&= \int_{D_b} |\nabla f(z)|^2 d^2z (n^{-4\nu} +\mathcal O(n^{-5\nu})) = \mathcal O(n^{-4\nu}). 
\end{align*}
With similar aguments, but with a change of variables, we find that
\begin{align*}
\iint_{D^{III}_n} (f(z)-f(w))^2 n^{1+4\nu} e^{-2 n(\xi-\xi_\tau)^2 g(\xi+i\eta)} e^{-2 n(\xi'-\xi_\tau)^2 g(\xi'+i\eta')} d^2z d^2w
= \mathcal O\left(n^{-1+6\nu} \right)
\end{align*}
as $n\to\infty$. 
\end{proof}

\begin{proof}[Proof of Theorem \ref{thm:CLTfixedTau}]
{
We first prove $\operatorname{Var} \mathfrak X_n(f) = \Sigma(f)^2=\sigma^2(f)+\tilde\sigma^2(f)$ (we included the dependence on $f$ in the notation here).} This follows by combining Lemma \ref{lem:DIIn}, Lemma \ref{lem:DIn} and Lemma \ref{lem:DIIIn} (note that the integration domain in the latter is bigger than the region that remains). The only task left is to rewrite the double integral over the boundary. We notice that for $z\in\partial E_\tau$
\begin{align*}
\frac{z - \sqrt{z^2 - 4\tau}}{2} = \sqrt \tau \cosh(\xi_\tau+i\eta) - \sqrt \tau \sinh(\xi_\tau+i\eta) 
= \sqrt \tau e^{-\xi_\tau-i\eta} = \tau e^{-i\eta}.
\end{align*}
Hence
\begin{align*}
-4\sin^2 \eta_- = \frac{\left(z - \sqrt{z^2 - 4\tau} - w + \sqrt{w^2 - 4\tau}\right)^2}{(z - \sqrt{z^2 - 4\tau})(w - \sqrt{w^2 - 4\tau})}.
\end{align*}
The Jacobian factors of the transformations from elliptic coordinates to $z$ and $z'$ are canceled by $\sqrt{z^2-4\tau} = 2\sqrt\tau\sinh(\xi_\tau+i\eta)$ and $\sqrt{w^2-4\tau} = 2\sqrt\tau\sinh(\xi_\tau+i\eta')$.

{Finally, we prove that this implies a CLT \eqref{eq:CLT_fixed_tau} for all $f\in C^1$ by a standard argument. The important point to note is that the limiting variance $\Sigma(f)^2$ is continuous in $C^1$ (with canonical norm), and that we already know that a CLT holds for $C^\infty$ functions with compact support \cite{AmHeMa2,AmHeMa}. 
We start by centering the linear statistic $X_n(f)=\mathfrak{X}_n(f)-\mathbb E \mathfrak{X}_n(f)$ and mention that for all $f\in C^1$ and $g\in C^\infty,$ we have
\begin{align*}
   \left| \mathbb E \left[e^{i t X_n(f)} \right]-e^{-\frac{t^2}{2\Sigma^2(f)}}\right|\leq  \left| \mathbb E \left[e^{i t X_n(f)}-e^{i t X_n(g)}\right]\right|+  \left| \mathbb E \left[e^{i t X_n(g)}\right]-e^{-\frac{t^2}{2\Sigma^2(g)}}\right|+\left|e^{-\frac{t^2}{2\Sigma^2(f)}}-e^{-\frac{t^2}{2\Sigma^2(g)}}\right|.
\end{align*}
Now we use the standard estimate for linear statistics to estimate 
$$
\left| \mathbb E \left[e^{i t X_n(f)}-e^{i t X_n(g)}\right]\right|\leq \mathbb E \left[\left| e^{i t X_n(f-g)}-1\right|\right]\leq |t|  \mathbb E \left[|X_n(f-g)|\right]\leq |t|\sqrt {\mathop {\mathrm{ Var}} X_n(f-g)}.$$
Using this estimate and the fact that we have a CLT for $g \in C^\infty$,
\begin{align*}
   \limsup_{n\to \infty} \left| \mathbb E \left[e^{i t X_n(f)} \right]-e^{-\frac{t^2}{2\Sigma^2(f)}}\right|\leq  |t| \Sigma(f-g)+\left|e^{-\frac{t^2}{2\Sigma^2(f)}}-e^{-\frac{t^2}{2\Sigma^2(g)}}\right|.
\end{align*}
which holds for arbitrary $g \in C^\infty$. Finally, by continuity of the map $f \mapsto \Sigma(f)^2$,  the right-hand side can be made arbitrarily small by taking $f \to g$ in $C^1$. This proves the statement. 
}
\end{proof}

\section{Smooth linear statistics in the weak non-Hermiticity regime} \label{sec:smoothVarWeak}

We consider (rescaled) linear statistics in a regime of weak non-Hermiticity. 
\begin{align*}
\mathfrak X_n^\gamma(f) = \sum_{j=1}^n f(n^{\gamma} z_j), \hspace{3cm} \tau = 1-\frac{\kappa}{n^\alpha}.
\end{align*}
Here $\kappa, \gamma> 0$ and $0< \alpha<1$ are fixed constants, and $f:\mathbb C\to \mathbb R$ is some differentiable test function which we will assume to have compact support.  Without loss of generality, we assume that $f$ has support $[-T, T]\times [-iT, iT]$ for some $T>\kappa$. 

The variance of the linear statistic $S$ is given by the formula
\begin{align*}
\text{Var } \mathfrak X_n^\gamma(f)
&= \frac{1}{2}\int_{\mathbb C} \int_{\mathbb C} (f(n^\gamma z)-f(n^\gamma w))^2 \left|\mathcal K_n(z, w)\right|^2 d^2z \, d^2w\\
&= \frac{1}{2} \int_{\mathbb C} \int_{\mathbb C} (f(u)-f(u'))^2 \left|\frac{1}{n^{2\gamma}}\mathcal K_n\left(\frac{u}{n^\gamma}, \frac{u'}{n^\gamma}\right)\right|^2 d^2u \, d^2u'.
\end{align*}


To calculate the variance, we have to consider
\begin{align}
z &= \frac{u}{n^\gamma \sqrt{2\tau}}, \hspace{2cm}
w = \frac{\overline{u'}}{n^\gamma \sqrt{2\tau}},
\end{align}
in the integral $I_n(\tau;z,w)$ as defined in \eqref{eq:defIn}. Remember that { $z=2\sqrt \tau \cosh(\xi+i\eta)$ and $w=2\sqrt \tau \cosh(\xi'+i\eta')$} in elliptic coordinates, where $\xi\geq 0$ and $\eta\in (-\pi,\pi]$ when $\xi>0$, $\eta\in [0,\pi]$, when $\xi=0$, and similar for $\xi'$ and $\eta'$. We remind the reader of the notations
\begin{align*}
    \xi_\pm = \frac{\xi\pm \xi'}2, \quad \eta_\pm = \frac{\eta\pm \eta'}2, \quad \xi_\tau=-\frac12\log \tau.
\end{align*}

\subsection{The case $\alpha<\gamma$.}

The case $\alpha<\gamma$ essentially follows from Theorem III.5 of \cite{ADM}, which we formulated as Proposition \ref{prop:kernelIneq} above. 
When $\alpha<\gamma$ it follows from the definition of $g$ in Proposition \ref{prop:mainThmADM} that as $n\to\infty$
\begin{align*} 
-(\xi - \xi_\tau)^2 g(\xi+i\eta) &=  \frac{1}{2} +\xi -\xi_\tau + \frac{1}{2} e^{-2\xi} \cos(2\eta)
- 2\tau \frac{\cosh^2 \xi \cos^2 \eta}{1+\tau} 
- 2\tau \frac{\sinh^2 \xi \sin^2 \eta}{1-\tau}\\
&= 2\xi - \xi_\tau + \mathcal O(n^{\alpha-2\gamma}) = -\frac{\kappa}{2 n^\alpha} \left(1+\mathcal O(n^{\alpha-\gamma})\right), 
\end{align*}
Combining this with Proposition \ref{prop:mainThmADM} we infer that there exists a constant $c>0$ such that uniformly for $u, u'$ in compact sets as $n\to\infty$
\begin{align*}
n^{-4\gamma}\left|\mathcal K_n(n^{-\gamma} u, n^{-\gamma} u')\right|^2
= \frac{n^{2(1+\alpha-2\gamma)}(1+\mathcal O(n^{-\alpha}))}{(2\pi \kappa)^2} e^{-n^{1+\alpha-2\gamma} \frac{|u-u'|^2}{2\kappa}}
+ \mathcal O\left(e^{- c n^{1-\alpha}}\right).
\end{align*}

\begin{proof}[Proof of Theorem \ref{thm:VarSgamma>delta}(i)]
Let us first consider the case that $\gamma<\frac{1+\alpha}{2}$. Then, denoting $\partial_1=\partial/\partial\Re u$ and $\partial_2=\partial/\partial\Im u$, we have
\begin{align*}
\text{Var }S 
&= \frac{1}{2}\int_{\mathbb C} \int_{\mathbb C} (f(u)-f(u'))^2 \frac{n^{2(1+\alpha-2\gamma)}(1+\mathcal O(n^{-\alpha}))}{(2\pi \kappa)^2} e^{-n^{1+\alpha-2\gamma} \frac{|u-u'|^2}{2\kappa}} d^2u d^2u'
+ \mathcal O\left(n^{-4\gamma} e^{- c n^{1-\alpha}}\right)\\
&= \frac{(1+\mathcal O(n^{-\alpha}))}{2 n^{4\gamma}}\int_{\mathbb C} \int_{\mathbb C} \left(\partial_1 f(u) \Re(u-u')+\partial_2 f(u) \Im(u-u')+ \mathcal O((u-u')^2)\right)^2 \frac{n^{2+2\alpha}}{(2\pi \kappa)^2} e^{-n^{1+\delta-2\gamma} \frac{|u-u'|^2}{2\kappa}} d^2u d^2u'\\
&\quad+ \mathcal O\left(n^{-4\gamma} e^{- c n^{1-\alpha}}\right)\\
&= \frac{\pi(2\kappa)^2 n^{2+2\delta} (1+\mathcal O(n^{-\alpha})) (n^{2\gamma-\alpha-1})^2}{2 n^{4\gamma} 2 (2\pi \kappa)^2} \int_{\mathbb C} \left((\partial_1 f(u))^2+(\partial_2 f(u))^2\right) d^2u
+ \mathcal O\left(\frac{(1+\mathcal O(n^{-\alpha}))}{n^{1+\alpha-2\gamma}}\right)\\
&= \frac{1}{4\pi} \int_{\mathbb C} \left|\nabla f(u)\right|^2 d^2u + \mathcal O\left(n^{-\alpha}+n^{2\gamma-\alpha-1}\right)
\end{align*}
as $n\to\infty$.
If, on the other hand, $\gamma>\frac{1+\alpha}{2}$, then we get
\begin{align*}
\text{Var }S
&= \frac{1}{2} \iint (f(u)-f(u'))^2 \frac{n^{2+2\alpha}(1+\mathcal O(n^{-\alpha}))}{(2\pi \kappa)^2} (1+\mathcal O(n^{1+\alpha-2\gamma})) d^2u d^2u'
+ \mathcal O\left(n^{-4\gamma} e^{- c n^{1-\alpha}}\right)\\
&= \mathcal O(n^{2+2\alpha - 4\gamma}),
\end{align*}
and this converges to $0$. Lastly, when $\gamma = \frac{1+\alpha}{2}$, we get as $n\to\infty$ that
\begin{align*}
\text{Var }S
&= \frac{1}{2 n^{4\gamma}} \int_{\mathbb C} \int_{\mathbb C} (f(u)-f(u'))^2 \frac{n^{2+2\alpha}(1+\mathcal O(n^{-\delta}))}{(2\pi \kappa)^2} e^{-\frac{|u-u'|^2}{2\kappa}} d^2u \, d^2u'
+ \mathcal O\left(n^{-4\gamma} e^{- c n^{1-\alpha}}\right)\\
&= \frac{1}{8 \pi^2 \kappa^2} \int_{\mathbb C}\int_{\mathbb C} (f(u)-f(u'))^2 e^{-\frac{|u-u'|^2}{2\kappa}} d^2u \, d^2u' + \mathcal O(n^{-\alpha}). 
\end{align*}
\end{proof}

\begin{remark}
The case $\gamma = \frac{1+\alpha}{2}$ is a transitional case. The intuition behind this case is that $\gamma = \frac{1+\alpha}{2}$ corresponds to the microscopic distance. Namely, the area of the bulk (combined with the compact support) is of order $n^{-\alpha}$. When one distributes $n$ points in this region, each point should take up an area of order $n^{-1-\alpha}$. Since we are working in two dimensions, this gives a microscopic distance of order $n^{-\frac{1+\alpha}{2}}$.
In the limit $\kappa\to 0$, the variance for the case $\gamma = \frac{1+\alpha}{2}$ approaches the variance for the case $\gamma<\frac{1+\alpha}{2}$, while for $\kappa\to\infty$, it approaches $0$. 
\end{remark}

\subsection{The case $\alpha=\gamma$}

Our approach for the case $\alpha=\gamma\in (0,1)$ will be similar to that in Section \ref{sec:sec1}, where the variance was calculated in the strong non-Hermiticity regime (fixed $\tau$). The main difference is a rescaling, i.e., compared to Section \ref{sec:sec1} similar formulae hold where $n$ is replaced by $n^{1-\alpha}$, as derived in Section \ref{sec:steepestWeakNonH}. Here we shall use the formula
\begin{align*}
\text{Var } \mathfrak X_n^\gamma(f)
&= \frac{1}{2}\int_{\mathbb C} \int_{\mathbb C} (f(n^\alpha z)-f(n^\alpha w))^2 \left|\mathcal K_n(z, w)\right|^2 d^2z \, d^2w.
\end{align*}

\begin{lemma} \label{lem:DInWeak}
Suppose that $f:\mathbb C\to \mathbb R$ is differentiable and has compact support. Let $0<\nu<\frac{1-\alpha}{6}$, and let $D^I_n$ be the set of $(z, w)\in \mathbb C^2$ such that $|\xi_+-\xi_\tau|^{-1}=\mathcal O(n^{-\frac{1+\alpha}{2}+\nu})$.
Then
\begin{align*}
\lim_{n\to\infty} \frac{1}{2} \iint_{D^I_n} (f(n^\alpha z)-f(n^\alpha w))^2 |\mathcal K_n(z, w)|^2 d^2z d^2w
= \frac{1}{4\pi^2} \int_{|\Im u|\leq\kappa} |\nabla f(u)|^2 d^2 u. 
\end{align*} 
\end{lemma}

\begin{proof}
    Proposition \ref{prop:mainThmADM} gives us that
\begin{align*}
    |\mathcal K_n(z, w)|^2 = (f(n^\alpha z)-f(n^\alpha w))^2\frac{n^2 \mathfrak{1}_{\xi_+<\xi_\tau}}{\pi^2(1-\tau^2)^2} e^{-n \frac{|z-w|^2}{1-\tau^2}}+e^{-2n (\xi-\xi_\tau)^2 g(\xi+i\eta)}e^{-2n (\xi'-\xi_\tau)^2 g(\xi'+i\eta)}\mathcal O(n^{2+2\alpha-2\nu})
\end{align*}
as $n\to\infty$, uniformly for $|\xi-\xi_\tau|^{-1}=\mathcal O(n^{\frac{1+\alpha}{2}-\nu})$, where $\nu\in[0,\frac{1-\alpha}{6})$. 
Using Lemma \ref{lem:behavg} we know that
\begin{align*}
    e^{-2n (\xi-\xi_\tau)^2 g(\xi+i\eta)}e^{-2n (\xi'-\xi_\tau)^2 g(\xi'+i\eta)}
    = \mathcal O(e^{-\frac{4\tau}{\kappa(1+\tau)}n^{2\nu}})
\end{align*}
as $n\to\infty$, uniformly for $|\xi-\xi_\tau|^{-1}=\mathcal O(n^{\frac{1+\alpha}{2}-\nu})$. For this one uses the fact that $f$ has compact support, and thus at least one of $\eta, \eta'$ is close to $\pm \frac{\pi}{2}$. The dominant behavior of the variance, after a rescaling $(z,w)\mapsto n^{-\alpha} (u,v)$, is thus given by
\begin{align*}
    \frac12 \iint_{n^\alpha D_n^I} (f(u)-f(u'))^2 \frac{n^{2(1-\alpha)} \mathfrak{1}_{\xi_+<\xi_\tau}}{\pi^2 (1+\tau)^2\kappa^2} e^{-n^{1-\alpha}\frac{|u-u'|^2}{(1+\tau)\kappa}} d^2u \, d^2u',
\end{align*}
where $\xi_+$ now refers to the rescaled variables $u$ and $u'$.
From here, the proof is essentially the same as that of Lemma \ref{lem:DIn}. One may rescale the Ginibre kernel (the right-hand side of \eqref{KGin}) 
    by $n^{-\alpha}$.  Replacing the condition $\xi_+<\xi_\tau$ by $|\Im u|+|\Im u'|<2\kappa$ in the indicator function will yield a negligible difference. Carrying out the integration over $u'$ will turn this condition into $|\Im u|\leq\kappa$. 
\end{proof}

\begin{lemma} \label{lem:DIInalpha=gamma}
    Suppose that $f:\mathbb C\to \mathbb R$ is differentiable and has compact support. Let $0<\nu<\frac{1-\alpha}{6}$, and let $D^{II}_n$ be the set of $(z, w)\in \mathbb C^2$ such that $|\xi-\xi_\tau|,|\xi'-\xi_\tau|\leq n^{-\frac{1+\alpha}{2}+\nu}$ and $\sin^2\eta_-\geq n^{-1-\alpha+2\nu}$. Then we have 
    \begin{align*}
\frac{1}{2} & \iint_{D^{II}_n} (f(n^\alpha z)-f(n^\alpha w))^2 |\mathcal K_n(z, w)|^2 d^2z \, d^2w\\
&= \frac{1}{16\pi^2} \int \int \left(\frac{f((2n^\alpha-\kappa)\cos\eta+i\kappa\sin\eta)-f((2n^\alpha-\kappa)\cos\eta'+i\kappa\sin\eta')}{\sin \eta_-}\right)^2 d\eta \, d\eta'+\mathcal O(n^{-\frac{1+\alpha}{2}+\nu}),
\end{align*}
where the integration over $d\eta d\eta'$ is over all $(\eta, \eta')\in(-\pi,\pi]^2$ that satisfy $\sin^2\eta_-\geq n^{-1-\alpha+2\nu}$.
\end{lemma}

\begin{proof}
    Proposition \ref{prop:mainThmADM} gives us that
\begin{multline*}
    (f(n^\alpha z)-f(n^\alpha w))^2|\mathcal K_n(z, w)|^2 \\
    =         (f(n^\alpha z)-f(n^\alpha w))^2\frac{n}{32\pi^3\tau(1-\tau^2)} 
        \frac{e^{-2n(\xi-\xi_\tau)^2 g(\xi+i\eta))}e^{-2n(\xi-\xi_\tau)^2 g(\xi'+i\eta'))}}{|\sinh(\xi_+-\xi_\tau+i\eta_-)|^2 |\sinh(\xi+i\eta) \sinh(\xi'+i\eta')|}  (1+\mathcal O(n^{-1+\alpha}))
\end{multline*}
as $n\to\infty$, uniformly for $|\xi-\xi_\tau|,|\xi'-\xi_\tau|\leq n^{-\frac{1+\alpha}{2}+\nu}$ and $\sin^2\eta_-\geq n^{-1-\alpha+2\nu}$, with $g$ defined as in \eqref{eq:defG}. This is because the bulk term, the first term on the right-hand side of \eqref{eq:mainThmADM}, is negligible, namely
\begin{align*}
(f(n^\alpha z)-f(n^\alpha w)) \frac{n \mathfrak{1}_{\xi_+<\xi_\tau}}{\pi(1-\tau^2)} e^{-\frac12n \frac{|z-w|^2}{1-\tau^2}}
= \mathcal O(n^{1+\alpha} e^{-c n^{2\nu}}),
\end{align*}
uniformly for $(z,w)\in D_n^{II}$ as $n\to\infty$, for some constant $c>0$. We have used here that 
\[
|z-w|^2 = 4\tau |\cosh(\xi+i\eta)-\cosh(\xi'+i\eta')|^2
= 4\tau |\sinh(\xi_++i\eta_+) \sinh(\xi_-+i\eta_-)|^2.
\]
When $\eta_+$ is bounded away from $0$ or $\pm\pi$ we see that $|z-w|\geq C \sin^2\eta_-\geq C n^{-1-\alpha+2\nu}$ for some constant $C>0$. On the other hand, since $f$ has compact support, when $\eta_+\to 0$ to get a nonzero contribution we necessarily have $\eta=\pm\frac\pi2+\mathcal O(n^{-\alpha})$ and $\eta'=\mp\frac\pi2+\mathcal O(n^{-\alpha})$ and then we get $|z-w|\geq C \sinh^2\xi_+ \sim \frac14 C n^{-2\alpha}$ for some constant $C>0$, which gives an even better estimate since $\alpha+2\nu<1$.

    Next, we will argue that the dominant contribution of the integral of $(f(n^\alpha z)-f(n^\alpha w))^2|\mathcal K_n(z,w)|^2$ over $D_n^{II}$ comes from a small neighborhood of $\xi=\xi_\tau$ and $\xi'=\xi_\tau$, which will allow us to use Laplace's method. Remember that $g$, as given in \eqref{eq:defG}, is defined through the relation
    \begin{align*}
        -(\xi-\xi_\tau)^2 g(\xi+i\eta) = \frac{1}{2}+\xi-\xi_\tau+\frac{1}{2}e^{-2\xi} \cos(2\eta)-\frac{2\tau}{1+\tau} \cosh^2\xi\cos^2\eta
        - \frac{2\tau}{1-\tau} \sinh^2\xi\sin^2\eta.
    \end{align*}
    When $\xi-\xi_\tau$ is small we have by a Taylor series expansion for any $\eta$ and any $0<\tau<1$ that
    \begin{align*}
        -(\xi-\xi_\tau)^2 g(\xi+i\eta) = -(\xi-\xi_\tau)^2g(\xi_\tau+i\eta)+\mathcal O\left(n\frac{(\xi-\xi_\tau)^3}{1-\tau}\right).
    \end{align*}
    Identifying $\tilde\xi=n^\alpha(\xi-\xi_\tau)$ we have in our present situation that
    \begin{align*}
        -n(\xi-\xi_\tau)^2 g(\xi+i\eta) = - n^{1-\alpha}\frac{1+\tau^2+2\cos(\eta)}{\kappa(1+\tau)}\tilde\xi^2+\mathcal O\left(n^{1-2\alpha}\tilde\xi^3\right),
    \end{align*}
    as $n\to\infty$, where $\tilde\xi=\mathcal O(n^{-\frac{1-\alpha}{2}+\nu})$.  Integrating over $d\xi d\xi'$ (taking into account the Jacobian factors), and applying Laplace's method, putting $\tilde\xi=\tilde\xi'=0$ in the integrand, we find the stated behavior. For this last step one rewrites the integrand using
    \begin{align*}
        n^\alpha z= n^\alpha 2 \sqrt\tau \cosh(\xi+i\eta)
        = n^\alpha 2\sqrt\tau (\cosh(\xi_\tau+i\eta)+\mathcal O(n^{-\alpha}\tilde \xi))
    \end{align*}
    and
    \begin{align*}
      n^\alpha 2\sqrt\tau \cosh(\xi_\tau+i\eta) = n^\alpha (1+\tau) \cos\eta +i n^\alpha (1-\tau) \sin\eta
      = (2 n^\alpha-\kappa) \cos\eta+i\kappa\sin\eta.
    \end{align*}
\end{proof}

\begin{lemma} \label{lem:DIIInWeak}
Suppose that $f:\mathbb C\to \mathbb R$ is differentiable and has compact support.
Let $0<\nu<\frac{1-\alpha}{6}$, and let $D^{III}_n$ be the set of $(z, w)\in \mathbb C^2$ such that $|\xi-\xi_\tau|, |\xi'-\xi_\tau|\leq n^{-\frac{1+\alpha}{2}+\nu}$ and $\sin^2 \eta_-\leq n^{-1-\alpha+2\nu}$. Then we have
\begin{align*}
\lim_{n\to\infty}\frac{1}{2} \iint_{D^{III}_n} (f(n^\alpha z)-f(n^\alpha w))^2 |\mathcal K_n(z, w)|^2 d^2z d^2w
= 0.
\end{align*} 
Furthermore, we have
\begin{align*}
    \lim_{n\to\infty}\int \int \left(\frac{f((2n^\alpha-\kappa)\cos\eta + i\kappa\sin\eta)-f((2n^\alpha-\kappa)\cos\eta' + i\kappa\sin\eta')}{\sin \eta_-}\right)^2 d\eta \, d\eta'=0,
\end{align*}
where the integration is over all $(\eta, \eta')\in (-\pi, \pi]^2$ such that $\sin^2 \eta_-\leq n^{-1-\alpha+2\nu}$.
\end{lemma}

\begin{proof}
    The proof is essentially the same as in Lemma \ref{lem:DIIIn}. First we use some simple estimates to express the kernel via the $\erfc$ formula, see \eqref{eq:weakNonHerfcResult}, as a sum of bulk and edge terms and then we follow the proof. Namely, from Proposition \ref{prop:behavEdgeClose} (which is the counterpart to Proposition 5.1 in \cite{Mo}) we have
    \begin{align*}
        |\mathcal K_n(z, w)|^2 = \mathcal O\left(\frac{n}{1-\tau^2} e^{-n \frac{|z-w|^2}{1-\tau^2}}\right)
        + \mathcal O\left(n^{1-\alpha} e^{-2n(\xi-\xi_\tau)^2 g(\xi+i\eta)} e^{-2n(\xi-\xi_\tau)^2 g(\xi+i\eta)}\right)
    \end{align*}
    as $n\to\infty$, uniformly for $(z,w)\in D_n^{III}$.
    Using Lemma \ref{lem:behavg} we thus infer that the situation is the same if we substitute $(z, w)\to n^{-\alpha}(u, u')$ and replace $n$ by $n^{1-\alpha}$ in the argument used. For the second part of the lemma, we notice that the integrand is integrable and the integration region has a size of order $n^{-\frac{1+\alpha}{2}+\nu} n^{-\alpha}$ while the integrand behaves at most like order $n^{2\alpha}$. The integral is thus $\mathcal O(n^{-\frac{1-\alpha}{2}+\nu})$ as $n\to\infty$.
\end{proof}

\begin{lemma} \label{lem:limitffkappa}
Suppose that $f:\mathbb C\to \mathbb R$ is differentiable and has compact support. 
We have
    \begin{align*}
    \lim_{n\to\infty} &\int_{-\pi}^\pi \int_{-\pi}^\pi\left(\frac{f((2n^\alpha-\kappa)\cos\eta+ i\kappa\sin\eta)-f((2n^\alpha-\kappa)\cos\eta'+ i\kappa\sin\eta')}{\sin \eta_-}\right)^2 d\eta \, d\eta'\\
    &= 4 \int_{-\infty}^\infty \int_{-\infty}^\infty \left(\frac{f(t+ i\kappa)-f(t'+ i\kappa)}{t-t'}\right)^2 dt \, dt'
    + 4 \int_{-\infty}^\infty \int_{-\infty}^\infty \left(\frac{f(t- i\kappa)-f(t'- i\kappa)}{t-t'}\right)^2 dt \, dt'.
\end{align*}
\end{lemma}

\begin{proof}
First, since $f$ is assumed to have compact support, we have
\begin{align*}
    f((2n^\alpha-\kappa)\cos\eta+ i\kappa\sin\eta)
    = f((2n^\alpha-\kappa)\cos\eta\pm i\kappa)+\mathcal O(n^{-\alpha})
\end{align*}
as $n\to\infty$. The dominant contributions will come from the regions where $\eta$ and $\eta'$ are both close to $\pm\frac{\pi}{2}$ (with the same sign) and consequently close to each other. To this end let us write $\eta=\pm \frac{\pi}{2}+n^{-\alpha}t$ and $\eta'=\pm \frac{\pi}{2}+n^{-\alpha}t'$, where $|t|, |t'|\leq  n^\nu$ for some constant $0<\nu<\alpha/2$. For these two regions of integration the contribution to the integral is
\begin{align} \nonumber
    \int_{-n^\nu}^{n^\nu} & \int_{-n^\nu}^{n^\nu} \left(\frac{f(\mp 2t\pm i\kappa)-f(\mp 2t'\pm i\kappa)+\mathcal O(n^{-\alpha}(t-t'))}{\sin(\frac12 n^{-\alpha}(t-t'))}\right)^2 n^{-2\alpha} dt \, dt'\\ \label{eq:dominantftik}
    &= 4\int_{-2n^\nu}^{2n^\nu} \int_{-2n^\nu}^{2n^\nu} \left(\frac{f(t\pm i\kappa)-f(t'\pm i\kappa)}{t-t'}\right)^2 dt \, dt'+ \mathcal O(n^{-\alpha+2\nu}).
\end{align}
In the case that $\eta=\pm\frac{\pi}{2}+n^{-\alpha}t$ and $\eta'=\mp\frac{\pi}{2}+n^{-\alpha}t'$, we have $\sin^2\eta_-=1+\mathcal O(n^{-\alpha+\nu})$ and the corresponding contribution is of order $\mathcal O(n^{-2\alpha})$, hence does not contribute to the dominant order.
The integration over the remaining region of $(\eta,\eta')$ is negligible. 
To explain why, it is enough to show that the integration over $[\frac{\pi}{2}-n^{-\alpha+\nu}, \frac{\pi}{2}+n^{-\alpha+\nu}]\times [\frac{\pi}{2}+n^{-\alpha+\nu}, \pi]$ tends to $0$ (substitute $\eta, \eta'\to \frac{\pi}{2}+n^{-\alpha}t, \frac{\pi}{2}+n^{-\alpha} t'$). This can be done with easy estimates, using that $f$ has compact support. Namely
    \begin{align*}
        &\int_{\frac{\pi}{2}-n^{-\alpha+\nu}}^{\frac{\pi}{2}+n^{-\alpha+\nu}} \int_{\frac{\pi}{2}+n^{-\alpha+\nu}}^\pi \left(\frac{f((2n^\alpha-\kappa)\cos\eta\pm i\kappa)-f((2n^\alpha-\kappa)\cos\eta'\pm i\kappa)}{\sin \eta_-}\right)^2 d\eta \, d\eta'\\
        &\leq 4 n^{-2\alpha}  \int_{-n^\nu}^{n^\nu} \int_{n^\nu}^{\frac{\pi}{2} n^\alpha}\frac{f((1+\tau)  t\pm i\kappa)^2}{\sin^2 \frac{t-t'}{2 n^\alpha}} dt dt' (1+\mathcal O(n^{-\alpha}))\\
        &= 8 n^{-\alpha} \int_{-n^\nu}^{n^\nu} f((1+\tau)t\pm i\kappa)^2 \left(\cot \frac{t-n^\nu}{2 n^\alpha}-\cot\Big(\frac{t}{2 n^\alpha}-\frac{\pi}{4}\Big)\right) dt (1+\mathcal O(n^{-\alpha})).
    \end{align*}
    While the integrand appears to be singular near $t=n^\nu$ this is actually not the case since the singularity is outside of the support of $f$. The integrand is of order $n^{\alpha-\nu}$ as $n\to\infty$ and nonzero only on a bounded subset of the integration domain. We conclude that the above expression is of order $n^{-\nu}$. Finally, a similar argument shows that the limits of integration of the dominant part, i.e., \eqref{eq:dominantftik}, may be replaced by $\pm \infty$. Namely, we have, e.g., that
    \begin{align*}
        \int_{-2n^\nu}^{2n^\nu} \int_{2n^\nu}^\infty \frac{(f(t\pm i\kappa)-f(t'\pm i\kappa))^2}{(t-t')^2} dt'\, dt
        \leq 4\int_{-2n^\nu}^{2n^\nu} \frac{f(t\pm i\kappa)^2}{2 n^\nu-t} dt,
    \end{align*}
    which, using the fact that $f$ has compact support, is of order $\mathcal O(n^{-\nu})$ as $n\to\infty$.
\end{proof}

\begin{proof}[Proof of Theorem \ref{thm:VarSgamma>delta}(ii)]
We remind the reader that the variance is given by the formula
\begin{align*}
\text{Var } \mathfrak X_n^\gamma(f)
&= \frac{1}{2}\int_{\mathbb C} \int_{\mathbb C} (f(n^\alpha z)-f(n^\alpha w))^2 \left|\mathcal K_n(z, w)\right|^2 d^2z \, d^2w.
\end{align*}
The dominant contributions to the integral correspond to Lemma \ref{lem:DInWeak} and Lemma \ref{lem:limitffkappa}, while the remaining contributions to the integral defining the variance are seen to be negligible by the combination of Lemma \ref{lem:DIInalpha=gamma} and Lemma \ref{lem:DIIInWeak}.
\end{proof}

\begin{remark}
We notice that the limiting variance approaches 
\begin{align*}
\frac{1}{2\pi^2} \int_{-\infty}^\infty \int_{-\infty}^\infty \left(\frac{f(x)-f(y)}{x-y}\right)^2 dx dy
\end{align*}
as $\kappa\downarrow 0$, while, when $\kappa\to\infty$, it approaches
\begin{align*}
\frac{1}{4\pi} \int_{\mathbb C} |\nabla f(z)|^2 d^2z.
\end{align*}
\end{remark}

\subsection{The case $\alpha>\gamma$}

In the case $\alpha>\gamma$ the condition $\xi_+<\xi_\tau$ is not satisfied for $n$ big enough, because the first expression is of order $n^{-\gamma}$ while the second expression is of order $n^{-\alpha}$ and thus smaller for large $n$. Therefore we do not have the bulk term in Proposition \ref{prop:behavKerWeakNotFaddeeva} and we conclude that
\begin{align*}
    \left|\frac{1}{n^{2\gamma}} \mathcal K_n\left(\frac{u}{n^\gamma}, \frac{u'}{n^\gamma}\right)\right|^2
    = \frac{n^{1+\alpha-2\gamma}}{2\kappa \pi^2} e^{-2n (\xi-\xi_\tau)^2 g(\xi+i\eta)} e^{-2n (\xi'-\xi_\tau)^2 g(\xi'+i\eta')} \frac{1+\mathcal O(n^{-\frac{1-\gamma}{2}+3\nu})}{|e^{\xi+\xi'+i(\eta-\eta')}-\tau|^2}
\end{align*}
as $n\to\infty$, uniformly for $(u- u'\pm n^{\gamma-\alpha} i\kappa)^{-1}=\mathcal O(n^{\frac{1-\gamma}{2}-\nu})$, where $0\leq \nu\leq \frac{1-\gamma}{6}$. The proof of Theorem \ref{thm:VarSgamma>delta}(iii) is almost entirely analogous to that of Theorem \ref{thm:VarSgamma>delta}(ii), the difference being that we need to replace 
\begin{align*}
    f((2n^\alpha-\kappa)\cos\eta-i\kappa)-f((2n^\alpha-\kappa)\cos\eta'-i\kappa)
\end{align*}
by
\begin{align*}
    f(2\sqrt\tau n^\gamma \cos\eta)-f(2\sqrt\tau n^\gamma \cos\eta')+\mathcal O(n^{-\alpha+\gamma}(\eta-\eta'))
\end{align*}
in Lemma \ref{lem:DIInalpha=gamma} and subsequent lemmas. We omit the details of the proof.

\section{Central limit theorem} \label{sec:CLT}

To prove Theorem \ref{thm:GFF} we will follow the method of Ward identities in \cite{AmHeMa}. The method in \cite{AmHeMa} does not apply to a weak non-Hermiticity regime nor the mesoscopic regime. While we consider the case that $\tau\to 1$ as $n\to\infty$, it will turn out that several expressions used in \cite{AmHeMa} are nevertheless valuable, i.e., which corresponds to fixed $\tau$. It will turn out to be more convenient to rescale our variables here by a factor $1-\tau$. Namely, we shall consider the correlation kernel
\begin{align} \label{eq:defRescaledKn}
    K_n(z, w) = (1-\tau)^2 \mathcal K_n\big((1-\tau) z, (1-\tau) w\big).
\end{align}
This essentially means that we are considering a potential
\begin{align} \label{eq:defRescaledQtau}
    Q_\tau(z) = V\big((1-\tau) z\big) = \frac{1-\tau}{1+\tau}\left(|z|^2-\tau \Re(z^2)\right) = \frac{(1-\tau)^2}{1+\tau} \Re(z)^2+(1-\tau) \Im(z)^2.
\end{align}
For fixed $\tau$ the corresponding droplet is given by
\begin{align*}
    E_\tau = \left\{z\in \mathbb C : \left(\frac{1-\tau}{1+\tau}\right)^2 \Re(z)^2+ \Im(z)^2\leq 1\right\}.
\end{align*}
Under this scaling we are viewing the linear statistic with test function $z\mapsto f(\kappa z)$, which is with the rescaled kernel. Without loss of generality we shall assume that $\kappa=1$ from now on. In other words, we investigate the linear statistic
\begin{align*}
    X_n(f) = \sum_{j=1}^n f(z_j),
\end{align*}
where $z_1, \ldots, z_n$ are the eigenvalues of the random normal matrix model with potential given by \eqref{eq:defRescaledKn}. We will also look at the linear statistics of the perturbed potential
\begin{align}
    Q_\tau^h(z) = Q_\tau(z) - \frac1n h(z),
\end{align}
where $h$ is any smooth function with compact support. We denote the corresponding correlation kernel by $K_n^h$. We will denote the corresponding linear statistic of a test function $f$ by $X_n^h(f)$. Important for the method in \cite{AmHeMa} is the functional
\begin{align} \label{eq:limitWard}
    v \mapsto \int_{\mathbb C} \left[v(z) \partial\bar\partial Q_\tau(z)+\bar\partial v(z) (\partial Q_\tau(z)-\partial \check{Q}_\tau(z))\right]  D_n^h(z) d^2z,
\end{align}
where $v$ is any smooth function. With a clever choice for $v$ it was shown in \cite{AmHeMa} that the functional equals $X_n(f)$. Contrary to \cite{AmHeMa} the expression between brackets depends on $n$ when we consider the weak non-Hermiticity regime $\tau=1-\kappa n^{-\alpha}$. In particular, $\check Q_\tau$ is not a limiting object when we consider the weak non-Hermiticity regime, and neither is $Q_\tau$ for that matter. It turns out that the functional \eqref{eq:limitWard} can be understood in terms of a Ward identity. We explain this rigorously in Section \ref{sec:WardId}.

We need to investigate several expressions used in \cite{AmHeMa}, some of which do not have an obvious equivalent in the weak non-Hermiticity regime. 

\subsection{An explicit expression for the obstacle function for any $\tau$}

In what follows we consider the obstacle function $\check Q_\tau$ for any $0\leq \tau<1$. As mentioned in the introduction, $\check Q_\tau$ is defined to be the maximal subharmonic function satisfying both $\check{Q}_\tau\leq Q_\tau$ and $\check Q_\tau(z)=\log|z|+\mathcal O(1)$ as $|z|\to \infty$. Furthermore, if one assumes that $Q_\tau$ is $C^2$, and $\Delta V>0$ on $S_V$, then the coincidence set $\check{Q}_\tau=Q_\tau$ equals $S_V=E_\tau$ up to a set of measure $0$. The obstacle function also has a potential theoretic relation. Define the logarithmic potential by
\begin{align*}
    U_\tau(z) = \frac2{\pi} \frac{1-\tau}{1+\tau}\int_{E_\tau} \log \frac{1}{|z-w|} d^2 w.
\end{align*}
It is a well-known fact that the obstacle function equals the logarithmic potential up to a constant. That is, there exists a Robin-type constant $c_\tau$ such that for any $0\leq \tau<1$
\begin{align*}
    \check Q_\tau+U_\tau = c_\tau.
\end{align*}
Furthermore, we know that $Q_\tau=\check Q_\tau$ on $E_\tau$. For each $0\leq \tau<1$ we will calculate the obstacle function explicitly. Since the formulae are slightly nicer for the elliptic Ginibre ensemble with the original scaling, i.e., with $V$ as defined in \eqref{eq:defVelliptic}, and since such a result is of independent interest, we first formulate the result for the obstacle function corresponding to \eqref{eq:defVelliptic}. 
The result uses the conformal map $\psi$ from the exterior of the droplet to the exterior of the unit disc, as defined in Theorem \ref{thm:CLTfixedTau}. For convenience to the reader we repeat its definition here.
\begin{align} \label{eq:defConfMapElliptic}
    \psi(z) = \frac{1}{2} (z + \sqrt{z^2-4\tau}),
\end{align}
where we remind the reader that the square root is chosen such that $\sqrt{z^2-4\tau}$ is positive for large positive numbers. It extends to an analytic map on $\mathbb C\setminus [-2\sqrt\tau, 2\sqrt \tau]$.

\begin{proposition} \label{prop:obstacle}
    The obstacle function of the elliptic Ginibre ensemble corresponding to \eqref{eq:defVelliptic} is given by 
    \begin{align*}
        \check{V}(z) = \begin{cases}
            V(z), & z\in\mathcal E_\tau,\\
            2\displaystyle\log |\psi(z)| + 2 + 2\tau \Re \frac{1}{\psi(z)^2}, & z\in\mathbb C\setminus \mathcal E_\tau,
        \end{cases}
    \end{align*}
    where $\psi$ is the conformal map \eqref{eq:defConfMapElliptic} from the exterior of $\mathcal E_\tau$ to the exterior of the unit disc. 
\end{proposition}

\begin{proof}
In this case the logarithmic potential is given by
\begin{align*}
    U(z) = \frac{1}{\pi(1-\tau^2)} \int_{\mathcal E_\tau} \log \frac1{|z-w|} d^2w.
\end{align*}
    Notice that
    \begin{align*}
        \partial U(z) = -\frac{1}{\pi(1-\tau^2)}\int_{\mathcal E_\tau} \frac{d^2w}{z-w}, \qquad 
        \bar\partial U(z) = -\frac{1}{\pi(1-\tau^2)}\overline{\int_{\mathcal E_\tau} \frac{d^2w}{z-w}}.
    \end{align*}
    Now we calculate 
    \begin{align*}
        g_\tau(z) := \frac{1}{\pi}\int_{\mathcal E_\tau} \frac{d^2w}{z-w}.
    \end{align*}
    We make a change of variables $w = 2\sqrt\tau \cosh(\tilde\xi+i\tilde\eta)$, where $\tilde\xi\in[0,\xi_\tau]$ and $\tilde\eta\in[-\pi, \pi]$. Then we have
    \begin{align*}
        g_\tau(z) &= \frac{4\tau}{\pi} \int_0^{\xi_\tau} \int_{-\pi}^\pi \frac{|\sinh(\tilde\xi+i\tilde\eta|^2}{z-2\sqrt\tau \cosh(\tilde\xi+i\tilde\eta)} d\tilde\eta d\tilde\xi\\
        &= i\frac{4\tau}{\pi} \int_0^{\xi_\tau} \oint_{|\zeta|=1} \frac{\sinh(\tilde\xi+\log\zeta) \sinh(\tilde\xi-\log\zeta)}{\sqrt\tau e^{\tilde\xi}\zeta^2-z\zeta+\sqrt\tau e^{-\tilde\xi}} d\zeta d\tilde\xi.
    \end{align*}
    Notice that the function $\zeta\mapsto \sinh(\tilde\xi+\log\zeta) \sinh(\tilde\xi-\log \zeta)$ is meromorphic with a pole of order $2$ at $\zeta=0$. The integrand in the contour integral thus has a pole of order two in $\zeta=0$, and two simple poles $\zeta_\pm = e^{- \tilde\xi\pm (\xi+i\eta)}$. The pole $\zeta=\zeta_-$ always contributes, but the pole $\zeta=\zeta_+$ contributes only when $\tilde\xi>\xi$. Let us consider the case that $\xi>\xi_\tau$. Then $\tilde\xi$ does not contribute and we get
    \begin{align*}
        i\frac{4\tau}{\pi} \oint_{|\zeta|=1} \frac{\sinh(\tilde\xi+\log\zeta) \sinh(\tilde\xi-\log\zeta)}{\sqrt\tau e^{\tilde\xi}\zeta^2-z\zeta+\sqrt\tau e^{-\tilde\xi}} d\zeta 
        &= -4\sqrt\tau\frac{\sinh(\xi+i\eta) \sinh(2\tilde\xi+\xi+i\eta)}{\sinh(\xi+i\eta)}+ 2z e^{2\tilde\xi}\\
        &=-4\sqrt\tau \sinh(2\tilde\xi+\xi+i\eta)+ 2e^{2\tilde\xi} z.
    \end{align*}
    (Alternatively, this follows from \cite{ByunPlanar}, Lemma 2.4.) Integrating this with respect to $\tilde\xi$ yields
    \begin{align*}
        g_\tau(z) &= -2\sqrt\tau (\cosh(2\xi_\tau+\xi+i\eta)-\cosh(\xi+i\eta)) + \frac{1-\tau}{\tau} z\\
        &= -\sinh(2\xi_\tau)\sqrt{z^2-4\tau} - \cosh(2\xi_\tau) z+ z +\frac{1-\tau}{\tau} z\\
        &= \frac{1-\tau^2}{2\tau} \left(z-\sqrt{z^2-4\tau}\right).
    \end{align*}
    We conclude that there exist a smooth function $h$ such that
    \begin{align*}
        U(z) &= -\log\Big(\frac{z}{2\sqrt\tau}+\sqrt{\frac{z^2}{4\tau}-1}\Big)-\frac{1}{2\tau} \left(z^2-z \sqrt{z^2-4\tau}+h(\overline z)\right)\\
        &= -\xi_\tau - \log \psi(z)- \frac{z}{\psi(z)}+h(\overline z).
    \end{align*}
    A similar formula holds with $z$ interchanged with $\overline z$. We then infer that for $\xi>\xi_\tau$ the obstacle function is given by
    \begin{align*}
        \check{V}(z) = \log |\psi(z)| + \Re \frac{z}{\psi(z)} + C,
    \end{align*}
    where $C$ is a constant that we need to determine. Since $V=\check V$ on $\mathcal E_\tau$, and $\check V$ should be continuous on $\partial \mathcal E_\tau$, we can find the constant by matching the expressions for $\check V$ on the boundary. 
    On $\partial \mathcal E_\tau$ we have
    \begin{align*}
        V(2\sqrt\tau \cosh(\xi_\tau+i\eta)) = 1+\tau\cos(2\eta).
    \end{align*}
    On the other hand, for $z=2\sqrt\tau\cosh(\xi_\tau+i\eta)$ we also have
    \begin{align*}
        \log |\psi(z)| + \Re \frac{z}{\psi(z)}
        = 1+\tau \cos(2\eta).
    \end{align*}
    Indeed with the particular choice $C=0$ we have $\check{V}-V=0$ on the boundary $\partial \mathcal E_\tau$. Using the explicit form of $\psi$ the reader may verify that $z=\psi(z)+\tau \psi(z)^{-1}$ and thus that $\check V$ can be rewritten as stated in the proposition. 
\end{proof}

Note in particular, considering the limit $z\to\infty$, that the Robin-type constant is
\begin{align*}
    \check{V}+U = -\xi_\tau.
\end{align*}
Perhaps unsurprisingly we infer that
\begin{align*}
    \check{V}(z) = F(\tau;z,z;a^{-1}), \qquad \xi\geq\xi_\tau,
\end{align*}
where $F$ was defined in \eqref{eq:defF}.
In particular, with $g$ as defined in \eqref{eq:defG}, we have
\begin{align*}
    V(z) - \check{V}(z) = g(\xi+i\eta).
\end{align*}

\begin{corollary}
    The obstacle function of the random normal matrix model with potential \eqref{eq:defRescaledQtau} is given by
    \begin{align}
        \check{Q}_\tau(z) = \begin{cases}
            Q_\tau(z), & z\in E_\tau,\\
            2\displaystyle\log |\psi_\tau(z)| + 2 + 2\tau \Re \frac{1}{\psi_\tau(z)^2}, & z\in\mathbb C\setminus E_\tau,
        \end{cases}
    \end{align}
    where
    \begin{align}
        \psi_\tau(z) = \psi((1-\tau) z)
        = \frac12 \big((1-\tau) z + \sqrt{(1-\tau)^2 z^2-4\tau}\big).
    \end{align}
\end{corollary}

Note that $\psi_\tau$ is the conformal map from the exterior of $E_\tau$ to the exterior of the unit disc. As $\tau\to 1$ the set $E_\tau$ will start to look more and more like a strip $|\Im z|\leq 1$. On the other hand, on any bounded set $\psi_\tau$, it will start to look like $\pm i$. In particular, $\check Q_\tau$ will be of order $1-\tau=\mathcal O(n^{-\alpha})$ on any bounded subset of $\mathbb C\setminus E_\tau$. 

\subsection{Decomposition of the test function}

In what follows we let $f$ be any smooth function on $\mathbb C$ with compact support.
As stated in \cite{AmHeMa}, and applied to our specific situation, we may decompose any smooth function $f$ with compact support as $f=f_{\tau+}+f_{\tau-}+f_{\tau}$, where for $z\in\mathbb C\setminus \mathcal E_\tau$ we have
\begin{align*}
    f_{\tau+}(z) &= \sum_{m=-\infty}^0 a_m(\tau) \psi_\tau(z)^m,\\
    f_{\tau-}(z) &= \sum_{m=1}^\infty a_m(\tau) \overline{\psi_\tau(z)}^{-m},
\end{align*}
which can be smoothly extended to $\mathbb C$, and with positive orientation we have
\begin{align} \label{eq:amtauCoeffs}
    a_m(\tau) = \frac{1}{2\pi i} \oint_{|\zeta|=1} \frac{f(\psi_\tau^{-1}(z))}{\zeta^{m+1}} d\zeta
    = \frac{1}{2\pi i} \oint_{\partial E_\tau} \frac{f(\zeta)}{\psi_\tau( \zeta)^{m+1}} \psi_\tau'(\zeta) d\zeta.
\end{align}
Of course, $f_\tau=f-f_{\tau+}-f_{\tau-}$. Then according to \cite{AmHeMa} all three functions are smooth and bounded, we have $\bar\partial f_{\tau+}=\partial f_{\tau-}=0$ in $\mathbb C\setminus \mathcal E_\tau$, and we have $f_\tau=0$ on $\partial \mathcal E_\tau$. Let us first focus on the case that $f=f_\tau+f_{\tau+}$, i.e., $a_m(\tau)=0$ for all $m<0$. Actually, as explained in \cite{AmHeMa}, it is enough to consider the case where $f=f_\tau+f_{\tau+}$ since the general case follows by an argument involving complex conjugation. 

Following Section 4.2 in \cite{AmHeMa}, with the particular decomposition $v=v_\tau+v_{\tau+}$, we have with
\begin{align*}
    v_\tau(z) &= \frac{\bar\partial f_\tau}{\Delta Q_\tau} \mathfrak{1}_{E_\tau}+\frac{f_\tau}{\partial Q_\tau-\partial \check{Q}_\tau} \mathfrak{1}_{\mathbb C\setminus E_\tau},\\
    v_{\tau+}(z) &= \frac{\bar\partial f_+}{\Delta Q_\tau} \mathfrak{1}_{E_\tau},
\end{align*}
that
\begin{align} \label{eq:limitWardId}
    -\pi  X_n^h(f) &= \int_{\mathbb C} \left[v(z) \partial\bar\partial Q_\tau(z)+\bar\partial v(z) (\partial Q_\tau(z)-\partial \check{Q}_\tau(z))\right]  D_n^h(z) d^2z.
\end{align}
where 
\begin{align} \label{eq:Dnh}
     D_n^h(z) =  \int_{\mathbb C} \frac1{z-\zeta} K_n^h(\zeta, \zeta) d^2\zeta - n \partial \check{Q}_\tau(z).
\end{align}
Notice that, while we do not indicate it, $v$ depends on $\tau$. $D_n^0$ equals $D_n^h$ with $h=0$.

\subsection{Using the Ward identity} \label{sec:WardId}

The content of this section is essentially a summary of part of \cite{AmHeMa}. Nevertheless we feel it is important to explain how the Ward identity is used. In \cite{AmHeMa}, a Ward identitiy is shown to hold for random normal matrices, given by
\begin{align} \label{eq:WardIdTau}
    \mathbb E_n^h \left(\operatorname{I}_n[v]-\operatorname{II}_n[v]+\operatorname{III}_n[v]\right) = 0.
\end{align}
where
\begin{align*}
    \operatorname{I}_n[v](z) &= \sum_{1\leq i<j\leq n} \frac{v(z_i)-v(z_j)}{z_i-z_j},\\
    \operatorname{II}_n[v](z) &= n\sum_{j=1}^n \partial Q_\tau^h(z_j) v(z_j)
    = \sum_{j=1}^n (n \partial Q_\tau(z_j) - \partial h(z_j)) v(z_j),\\
    \operatorname{III}_n[v](z) &= \sum_{j=1}^n \partial v(z_j).
\end{align*}
Here $\mathbb E_n^h$ denotes the expectation value with potential $Q_\tau^h$. Indeed, by $\mathbb E_n^0$ we shall mean the expectation value with respect to $Q_\tau$. We remind the reader that $h$ is a smooth function with compact support. Our goal is to explain how \eqref{eq:limitWard} is related to the Ward identity. Our first step is to write
\begin{align*}
    \mathbb E_n^h \operatorname{I}_n[v]
    &= \mathbb E_n^h \operatorname{II}_n[v]
    - \mathbb E_n^h \operatorname{III}_n[v]\\
    &= \int_{\mathbb C} \left(n \partial Q_\tau(z) v(z)-\partial h(z) v(z) - \partial v(z)\right) K_n^h(z,z) d^2z.
\end{align*}
Since we are dealing with a DPP, we can alternatively write
\begin{align*}
    \mathbb E_n^h \operatorname{I}_n[v]
    &= \int_{\mathbb C^2} \frac{v(z)-v(w)}{z-w} \det
    \begin{pmatrix} 
    K_n^h(z,z) & K_n^h(z,w)\\ K_n(w,z) & K_n(w,w)
    \end{pmatrix} d^2z d^2w\\
    &= \int_{\mathbb C^2} \frac{v(z)-v(w)}{z-w} K_n^h(z,z) K_n^h(w,w) d^2z d^2w
    - \int_{\mathbb C^2} \frac{v(z)-v(w)}{z-w} |K_n^h(z,w)|^2 d^2z d^2w.
\end{align*}
By a symmetry consideration, the first term equals
\begin{align*}
   \int_{\mathbb C^2} \frac{v(z)-v(w)}{v-w} K_n^h(z,z) K_n^h(w,w) d^2z d^2w
   &= 2\int_{\mathbb C^2} \frac{v(w)}{w-z} K_n^h(z,z) K_n^h(w,w) d^2z d^2w\\
   &= 2\int_{\mathbb C} v(z) \left(D_n^h(z)+\partial \check Q_\tau\right) d^2z.
\end{align*}
Combining the identity with the Ward identity yields
\begin{align*}
    \int_{\mathbb C} \left(\partial h(z) v(z) - n \partial Q_\tau(z) v(z) - \partial v(z)\right) K_n^h(z,z) d^2z
    = 2\int_{\mathbb C} v(z) \left(D_n^h(z)+\partial \check Q_\tau\right) d^2z - \int_{\mathbb C^2} \frac{v(z)-v(w)}{z-w} |K_n^h(z,w)|^2 d^2z d^2w.
\end{align*}
Rearranging some terms, we may rewrite this as
\begin{align} \label{eq:termPrecErrors}
    \frac1{\pi} \int_{E_\tau} v(z) D_n^h(z) \Delta Q_\tau(z) d^2z+ &
    \frac1{\pi} \int_{\mathbb C\setminus E_\tau} v(z) \partial (\check Q_\tau-Q_\tau)(z) \overline{\partial} D_n^h(z) d^2z\\
    &= - 2\int_{\mathbb C} (v(z) \partial h(z) + \frac12 \partial v(z)) \, K_n^h(z,z) \, d^2z +  \epsilon_{n,1}^h[v] + \epsilon_{n,2}^h[v], \nonumber
\end{align}
where
\begin{align*}
    \epsilon_{n,1}^h[v] &= \frac{1}{n} \int_{\mathbb C^2} \frac{v(z)-v(w)}{z-w} |K_n^h(z, w)|^2
    - \int_{\mathbb C} \partial v(z) K_n^h(z, z) d^2z,\\
    \epsilon_{n,2}^h[v] &= \frac1{2\pi n}\int_{\mathbb C} \bar\partial v(z) \big(D_n^h(z)\big)^2 d^2z.
\end{align*}
The left-hand side of \eqref{eq:termPrecErrors} can be written differently using integration by parts, namely
\begin{align*}
    \frac1{\pi} \int_{\mathbb C} v(z) D_n^h(z) \Delta Q_\tau(z) d^2z-
    \frac1{\pi} \int_{\mathbb C\setminus E_\tau} \overline{\partial} v(z) \partial (\check Q_\tau-Q_\tau)) D_n^h(z) d^2z.
\end{align*}
For this one uses that $\check Q_\tau$ is harmonic outside $E_\tau$. In fact we may replace the integration domain $\mathbb C\setminus E_\tau$ by $\mathbb C$ since $\check Q_\tau=Q_\tau$ on $E_\tau$. 
We thus end up with the functional \eqref{eq:limitWard} and clarified its relation to the Ward identity. We refer to \cite{AmHeMa} for the rigorous justification of the steps just mentioned. We emphasize that these relations are valid for any $n$ (and thus $\tau$). The next step is to actually take the limit $n\to\infty$. As in \cite{AmHeMa} the big challenge is to show that the error terms $\epsilon_{n,1}^h, \epsilon_{n,2}^h$ tend to $0$ as $n\to\infty$. Furthermore, the limit of the term preceeding the error terms on the right-hand side of \eqref{eq:termPrecErrors} has to be dealt with differently in the weak non-Hermiticity regime. 

\subsection{Estimates involving the kernel}

To estimate the error terms $\epsilon_{n,1}^h, \epsilon_{n,2}^h$, we need to get a good understanding of $K_n^h$ as $n\to\infty$. Let us consider the (non-weighted) Bergman kernel
\begin{align*}
    \boldsymbol k_n^h(z,w) = \sum_{j=0}^{n-1} p_j^h(z) \overline{p_j^h(w)},
\end{align*}
where the polynomials are orthogonal with respect to the weight $e^{- n Q_\tau^h(z)}=e^{-n Q_\tau(z)+h(z)}$.\\
Note that $K_n^h(z, w)=\boldsymbol k_n^h(z, w) e^{-\frac12 n Q_\tau^h(z)} e^{-\frac12 n Q_\tau^h(w)}$. We start with the following global estimate.

\begin{lemma} \label{lem:thm3.1}
    We have uniformly in $n=0,1,\ldots$ and $z\in\mathbb C$ and $0\leq\tau<1$ that
    \begin{align} \label{eq:lem:thm3.1}
        K_n^h(z, z) \leq (1-\tau) n C_h e^{\frac4{1+\tau}} e^{-n (Q_\tau(z)-\check{Q}_\tau(z))},
    \end{align}
    where $C_h>0$ is a constant that depends only on $h$. 
\end{lemma}
\begin{proof}
    As explained in \cite{AmHeMa} there is an obvious norm equivalence 
    \begin{align*}
    \int_{\mathbb C} |p(z)|^2 e^{-n Q_\tau^h(z)} d^2z
    \asymp \int_{\mathbb C} |p(z)|^2 e^{-n Q_\tau(z)} d^2z,
    \end{align*}
    for any polynomial $p$.
    Here the implied constants depend only on $\max h$ and $\min h$. Using the reproducing property and Cauchy-Schwarz we find that
    \begin{align*}
        \boldsymbol k_n^h(z, z) &= 
        \int_{\mathbb C} \boldsymbol k_n(z, \zeta) \boldsymbol k_n^h(\zeta, z) e^{-n Q_\tau(\zeta)} d^2\zeta\\
        &\leq \sqrt{\int_{\mathbb C} |\boldsymbol k_n(z,\zeta)|^2 e^{-n Q_\tau(z)} d^2\zeta
        \int_{\mathbb C} |\boldsymbol k_n^h(\zeta, z)|^2 e^{-n Q_\tau(\zeta)} d^2\zeta}\\
        &\leq e^{-\frac12\min h}\sqrt{\int_{\mathbb C} |\boldsymbol k_n(z,\zeta)|^2 e^{-n Q_\tau(z)} d^2\zeta
        \int_{\mathbb C} |\boldsymbol k_n^h(\zeta, z)|^2 e^{-n Q_\tau^h(\zeta)} d^2\zeta}
        = e^{-\frac12\min h} \sqrt{\boldsymbol k_n(z,z)} \,
        \sqrt{\boldsymbol k_n^h(z,z)}.
    \end{align*}
    From this inequality it follows that $\boldsymbol k_n^h(z,z)\leq e^{-\min h} \boldsymbol k_n(z,z)$. A similar argument gives an inequality in the other direction and we conclude that
    \begin{align*}
        K_n^h(z,z) \asymp K_n(z,z),
    \end{align*}
    with the implied constants depending only on $\max h$ and $\min h$. It therefore suffices to consider the case $h=0$. 
    Now we apply Proposition 3.6 in \cite{AmHeMa3} (with $n=m$) to the kernel 
    $(1-\tau) \mathcal K_n(\sqrt{1-\tau}\, z, \sqrt{1-\tau}\, z)$,
    i.e., the weighted Bergman kernel with respect to the potential $V(\sqrt{1-\tau} \, z)$. The constant $C$ in that proposition is given by $e^{\frac{4}{1+\tau}}$ (this constant expressed as an essential supremum is easy to determine in our case since $\Delta V$ is constant). We conclude that
    \begin{align*}
        (1-\tau) \mathcal K_n(\sqrt{1-\tau}\, z, \sqrt{1-\tau}\, z) \leq
        C_h e^{\frac{4}{1+\tau}} n e^{- n (V(\sqrt{1-\tau} \, z)-\check V(\sqrt{1-\tau} \, z))}
    \end{align*}
    for all $z\in\mathbb C$, $\tau\in (0,1)$ and $n=0,1,\ldots$, where $C_h$ is a constant that depends only on (the maximum or minimum of) $h$. 
    After a scaling $z\to \sqrt{1-\tau} \, z$ we obtain the result for $K_n$.  
\end{proof}

We shall need more refined approximations of $K_n^h$ eventually. To this end,  we introduce the projection operator 
\begin{align*}
    P_n^h[f](z) &= \int_{\mathbb C} \boldsymbol k_n^h(z,\zeta) f(\zeta) e^{-n Q_\tau(\zeta)+h(z)+(\zeta-z) \partial h(z)} d^2\zeta,
\end{align*}
say, on the space of complex polynomials. 
With $P_n^0$ we shall mean the projection operator with $h$ identically $0$. By the reproducing property we have
\begin{align*}
    P_n^h[f](z) = f(z)
\end{align*}
for any complex polynomial $f$ of degree less than $n$. In particular we have
\begin{align*}
    \boldsymbol k_n^h(z,w) &= P_n^0[\zeta\mapsto \boldsymbol k_n^h(\zeta, w)](z),\\
    \boldsymbol k_n^0(z,w) &= P_n^h[\zeta\mapsto \boldsymbol k_n^0(\zeta, w)](z).
\end{align*}
When $h=0$ we will often drop the superscript $h$ in the above expressions.

\begin{lemma} \label{lem:Knh-Knzz}
    We have uniformly for $z\in\mathbb C$  that
    \begin{align}
        |K_n^h(z,z)-K_n(z,z)|\leq C_h \sqrt{(1-\tau) n} \log n
    \end{align}
    as $n\to\infty$, where $C_h>0$ is a constant that depends only on $h$.
\end{lemma}

\begin{proof}
Using the reproducing property we have that 
\begin{align*}
    e^{-n Q_\tau(z)}\boldsymbol k_n(z,z) - e^{-n Q_\tau^h(z)}\boldsymbol k_n^h(z,z)
    = e^{-n Q_\tau(z)} \int_{\mathbb C} \boldsymbol k_n(z, \zeta) \boldsymbol k_n^h(\zeta, z) e^{-n Q_\tau(\zeta)} (e^{h(\zeta)}-e^{h(z)}) d^2\zeta. 
\end{align*}
Using this formula and the fact that $\boldsymbol k_n^h(z,z)\asymp \boldsymbol k_n(z, z)$ ,with implied constants depending only on $h$, we find that
\begin{align*}
    |K_n^h(z, z)-K_n(z, z)| 
    &\leq e^{-n Q_\tau(z)} C_h \int_{\mathbb C} |\boldsymbol k_n(z, \zeta)|^2 e^{-n Q_\tau(\zeta)} |e^{h(\zeta)}-e^{h(z)}| d^2\zeta\\
    &= C_h
    \int_{\mathbb C} |K_n(z, \zeta)|^2 |e^{h(\zeta)}-e^{h(z)}| d^2\zeta,
\end{align*}
for some constant $C_h>0$ that depends only on $h$. The lemma now follows by arguments similar to those in Section \ref{sec:sec1}. Write $z=\frac{2\sqrt\tau}{1-\tau} \cosh(\xi+i \eta)$ and $\zeta=\frac{2\sqrt\tau}{1-\tau} \cosh(\xi'+i \eta')$. Denote $\delta_n=n^{-\frac{1+\alpha}{2}} \log n$. When $\xi<\xi_\tau-\delta_n$ we may use Proposition \ref{prop:kernelIneq} to show that
\begin{align*}
    \int_{\mathbb C} |K_n(z, \zeta)|^2 |e^{h(\zeta)}-e^{h(z)}| d^2\zeta \leq C_h \sqrt{(1-\tau) n} |\nabla h(z)|,
\end{align*}
for some constant $C_h>0$. When $\xi_\tau-\delta_n<\xi, \xi'<\xi_\tau+\delta_n$ and $|\eta_-|=|\eta-\eta'|\geq \delta_n$ we see that as $n\to\infty$
\begin{align*}
    \int_{|\xi'-\xi_\tau|<\delta_n, |\eta_-|\geq \delta_n} |K_n(z, \zeta)|^2 |e^{h(\zeta)}-e^{h(z)}| d^2z
    &\leq C_h \sqrt{(1-\tau) n} e^{-n g(\xi+i\eta)} \int_{\delta_n<|\eta_-|\leq \pi} \frac{d\eta}{|\eta_-|} d\eta_-\\
    &\leq C_h \sqrt{(1-\tau) n} (\log\pi-\log\delta_n),
\end{align*}
for some constant $C_h>0$ depending only on $h$. The remaining regions can be estimated using Lemma \ref{lem:thm3.1} above outside the droplet, and the complementary error function for points $z, \zeta$ close to each other and close to $\partial E_\tau$ (see Proposition \ref{prop:behavEdgeClose}).  
\end{proof}

To estimate the kernel $\boldsymbol k_n^h$ outside the diagonal too, we will use an approximate projection operator. We define
\begin{align} \label{PnhApprox}
    P_n^{h,\#}[f](w) &= \int_{E_\tau} \overline{\boldsymbol k_n^{h,\#}(\zeta,w)} f(\zeta) e^{-n Q_\tau(\zeta)+h(\zeta)} d^2\zeta.
\end{align}
where  
\begin{align} \label{eq:knhApprox}
\boldsymbol k_n^{h,\#}(\zeta,w) &= \frac{n}{\pi} \frac{1-\tau}{1+\tau} e^{n Q_\tau(\zeta, \overline w)} e^{h_w(\zeta)},
\end{align}
$h_w(\zeta) = h(w)+(\zeta-w) \overline\partial h(w)$, and $Q_\tau$ is the polarization, explicitly given by
\begin{align*}
    Q_\tau(\zeta, \overline w)=\frac{1-\tau}{1+\tau} (\zeta  \overline w - \frac{\tau}{2}(\zeta^2+\overline w^2)).
\end{align*}
For $\zeta\approx w$ in the bulk, $\overline{\boldsymbol k_n^{h,\#}(\zeta,w)}$ is a good approximation of $\boldsymbol k_n^h(w, \zeta)$. 

\begin{lemma}
    For any complex polynomial $f$ and $w\in \mathring E_\tau$ we have
    \begin{align*}
        \left|P_n^{h,\#}[f](w) - f(w)  + \frac{1}{2 n i} \frac{1+\tau}{1-\tau} \int_{\partial E_\tau} \frac{\overline{\boldsymbol k_n^{h,\#}(\zeta,w)} f(\zeta) e^{-n Q_\tau^h(\zeta)}}{\zeta-w} d\zeta\right| 
        \leq  \frac{C_h e^{\frac12 n Q_\tau(w)}}{\sqrt{n(1-\tau)}} \left(\int_{E_\tau} |f(\zeta)|^2 e^{- n Q_\tau(\zeta)} d^2\zeta\right)^\frac12,
    \end{align*}
    where the constant $C_h>0$ depends only on $h$.
\end{lemma}

\begin{proof}
    By the Cauchy-Pompeiu formula \cite{pompeiu} we have for $w\in \mathring E_\tau$ that
    \begin{align*}
        \frac{n}{\pi} \frac{1-\tau}{1+\tau} f(w) =  \frac1{2\pi i} \int_{\partial E_\tau} \frac{\overline{\boldsymbol k_n^{h,\#}(\zeta,w)} f(\zeta) e^{-n Q_\tau(\zeta)+h(\zeta)}}{\zeta-w} d\zeta
        - \frac1{\pi} \int_{E_\tau} \overline\partial\left(\overline{\boldsymbol k_n^{h,\#}(\zeta,w)} f(\zeta) e^{-n Q_\tau(\zeta)+h(\zeta)}\right) \frac{d^2\zeta}{\zeta-w}.
    \end{align*}
    Using the fact that $f$ is analytic and the explicit form of $P_n^{h,\#}$ as defined in \eqref{PnhApprox}, we find that 
    \begin{align*}
        -\frac1{\pi} \int_{E_\tau} \overline\partial\left(\overline{\boldsymbol k_n^{h,\#}(\zeta,w)} f(\zeta) e^{-n Q_\tau(\zeta)+h(\zeta)}\right) \frac{d^2\zeta}{\zeta-w}
        = \frac{n}{\pi} \frac{1-\tau}{1+\tau} P_n^{h,\#}[f](w)
        - \frac1{\pi} \int_{E_\tau} \overline{\boldsymbol k_n^{h,\#}(\zeta,w)} f(\zeta) e^{-n Q_\tau(\zeta)+h(\zeta)} \frac{\overline\partial h(\zeta)-\overline{\partial} h(w)}{\zeta-w} d^2\zeta.
    \end{align*}
    We notice that
    \begin{align*}
        \int_{E_\tau} \overline{\boldsymbol k_n^{h,\#}(\zeta,w)} f(\zeta) e^{-n Q_\tau(\zeta)+h(\zeta)} \frac{\overline\partial h(\zeta)-\overline{\partial} h(w)}{\zeta-w} d^2\zeta &\leq 
        C_h  \int_{E_\tau} |\overline{\boldsymbol k_n^{h,\#}(\zeta,w)} f(\zeta)| e^{-n Q_\tau(\zeta)}
        d^2\zeta,
    \end{align*}
    where $C_h>0$ is a constant that depends only on $h$. Applying Cauchy-Schwarz gives
    \begin{align*}
        \left|\int_{E_\tau} \overline{\boldsymbol k_n^{h,\#}(\zeta,w)} |f(\zeta)| e^{-n Q_\tau(\zeta)}
        d^2\zeta\right|^2 \leq \int_{E_\tau} |\boldsymbol k_n^{h,\#}(\zeta, w)|^2 e^{-n Q_\tau(\zeta)} d^2\zeta
        \int_{E_\tau} |f(\zeta)|^2 e^{- n Q_\tau(\zeta)} d^2\zeta.
    \end{align*}
    Using the explicit form of $\boldsymbol k_n^{h,\#}$ as defined in \eqref{eq:knhApprox} we get
    \begin{align*}
        \int_{E_\tau} |\boldsymbol k_n^{h,\#}(\zeta, w)|^2 e^{-n Q_\tau(\zeta)} d^2\zeta 
        = \left(\frac{n}{\pi} \frac{1-\tau}{1+\tau}\right)^2 e^{n Q_\tau(w)}\int_{E_\tau} e^{- n \frac{1-\tau}{1+\tau} |\zeta-w|^2} d^2\zeta \leq \frac{n}{\pi} \frac{1-\tau}{1+\tau} e^{n Q_\tau(w)}.
    \end{align*}
    Combining this with the previous calculations yields the result (with a different $C_h$). 
\end{proof}

\begin{corollary}
    We have uniformly for $w\in \mathring E_\tau$ and $z\in\mathbb C$ that
    \begin{align} \label{eq:BergmanKerContourEstimate}
        P_n^{h,\#}[\zeta\mapsto \boldsymbol k_n^h(z, \zeta)](w)= \boldsymbol k_n^{h}(z, w) 
        -\frac{1}{2 n i} \frac{1+\tau}{1-\tau} \int_{\partial E_\tau} \frac{\overline{\boldsymbol k_n^{h,\#}(\zeta,w)} \boldsymbol k_n^h(\zeta,z) e^{-n Q_\tau^h(\zeta)}}{\zeta-w} d\zeta
        + \mathcal O(e^{\frac12 n  \check Q_\tau(z)+\frac12 n  Q_\tau(w)})
    \end{align}
    as $n\to\infty$, where the implied constant depends only on $h$.
\end{corollary}

\begin{proof}
    By the reproducing property we have with $f(\zeta)=\boldsymbol k_n^h(\zeta, z)$ that
    \begin{align*}
        \int_{E_\tau} |f(\zeta)|^2 e^{- n Q_\tau(\zeta)} d^2\zeta 
        \leq  e^{-\min |h|} \int_{\mathbb C} |k_n^h(z, \zeta)|^2 e^{- n Q_\tau(\zeta)+h(\zeta)} d^2\zeta
        = e^{-\min |h|} \boldsymbol k_n^h(z,z),
    \end{align*}
    which is $\mathcal O(n(1-\tau) e^{n \check Q_\tau(z)})$ as $n\to\infty$ by Lemma \ref{lem:thm3.1}. 
\end{proof}

By $d(z, \partial E_\tau)$ we shall denote the distance between $z$ and $\partial E_\tau$. Furthermore, let us denote
\begin{align*}
    \delta_n^\alpha = n^{-\frac{1-\alpha}{2}} \log n.
\end{align*}

\begin{lemma} \label{lem:operatorDifference}
    We have uniformly for $w\in \mathring E_\tau$ with $d(w, \partial E_\tau)\geq \delta_n^\alpha$ and $z\in\mathbb C$ that
    \begin{align*}
        \left|P_n^h[\zeta\mapsto \boldsymbol k_n^{h,\#}(\zeta, w)](z)
    - \overline{P_n^{h,\#}[\zeta\mapsto \boldsymbol k_n^h(z,\zeta)](w)}\right| = \mathcal O(e^{\frac12 n Q_\tau(z)+\frac12 n Q_\tau(w)}).
    \end{align*}
\end{lemma}

\begin{proof}
    First we notice that
    \begin{align*}
        \left|P_n^h[\zeta\mapsto \boldsymbol k_n^{h,\#}(\zeta, w)](z)
    - \overline{P_n^{h,\#}[\zeta\mapsto \boldsymbol k_n^h(z,\zeta)](w)}\right|
    \leq \int_{\mathbb C\setminus E_\tau} \left|\boldsymbol k_n^h(z,\zeta) \boldsymbol k_n^{h,\#}(\zeta, w) e^{-n Q_\tau(\zeta)+h(\zeta)}\right| d^2\zeta.
    \end{align*}
    By Cauchy-Schwarz and Lemma \ref{lem:thm3.1} we have for some constants $C_h>0$ that depend only on $h$ that
    \begin{align*}
        |\boldsymbol k_n^{h}(\zeta, w)| \leq C_h (1-\tau) n e^{\frac12 n \check Q(\zeta)+\frac12 n \check Q(w)}
    \end{align*}
    for $\zeta\in \mathbb C\setminus E_\tau$.
    On the other hand, it follows from its definition that
    \begin{align*}
        |\boldsymbol k_n^{h,\#}(z,\zeta)|  \leq  C_h \frac{n}{\pi} \frac{1-\tau}{1+\tau} e^{\frac12 n Q_\tau(z)+\frac12 n Q_\tau(\zeta)} e^{-\frac12 n \frac{1-\tau}{1+\tau}|\zeta-w|^2}
    \end{align*}
    for some possibly different $C_h>0$. We thus have for some possibly different constant $C_h>0$
    \begin{align*}
        \left|P_n^h[\zeta\mapsto \boldsymbol k_n^{h,\#}(\zeta, w)](z)
    - \overline{P_n^{h,\#}[\zeta\mapsto \boldsymbol k_n^h(z,\zeta)](w)}\right|
    \leq C_h \frac{n^2}{\pi} \frac{1-\tau}{1+\tau} e^{\frac12 n Q_\tau(z)+\frac12 n \check Q_\tau(w)} \int_{\mathbb C\setminus E_\tau} e^{-\frac12 n \frac{1-\tau}{1+\tau}|\zeta-w|^2} e^{-\frac12 n (Q_\tau(\zeta)-\check Q_\tau(\zeta))} d^2\zeta,
    \end{align*}
    and the integral on the right hand side is small enough due to the condition $d(w, \partial E_\tau)\geq \delta_n^\alpha$.
\end{proof}

\begin{lemma} \label{lem:thm3.2}
    Uniformly for all $w\in \mathring E_\tau$ with $d(w, \partial E_\tau)>\delta_n^\alpha$ and $z\in\mathbb C$ we have
    \begin{align*}
        K_n^h(z, w) = \frac{n}{\pi} \frac{1-\tau}{1+\tau} e^{n( Q_\tau(z, \overline w)-\frac12 Q_\tau(z)-\frac12 Q_\tau(w))} e^{-h_w(z)+\frac12 h(z)+\frac12 h(w)}
        +\mathcal O(1)
    \end{align*}
    as $n\to\infty$, where the implied constant depends only on $h$. In particular, we have
    \begin{align*}
        |K_n^h(z, w)|^2 = \frac{n^2}{\pi^2} \left(\frac{1-\tau}{1+\tau}\right)^2 e^{-n\frac{1-\tau}{1+\tau}|z-w|^2} e^{-2h_w(z)+h(z)+h(w)}+\mathcal O(n(1-\tau))
    \end{align*}
    as $n\to\infty$.
\end{lemma}

\begin{proof}
    We have to estimate the contour integral in \eqref{eq:BergmanKerContourEstimate}. When $d(w, \partial E_\tau)>\delta_n^\alpha$ we have
    \begin{align*}
        |\boldsymbol k_n^{h,\#}(\zeta,w)| = \mathcal O(n^{-c \log n})
    \end{align*}
    for some constant $c>0$. Now using that $|\zeta-w|\geq \delta_n^\alpha$ on $\partial E_\tau$, the global estimate \eqref{eq:lem:thm3.1} and that $|\partial E_\tau| = \mathcal O(n^\alpha)$, we infer that the contribution of the contour integral is $\mathcal O(n^{-c \log n})$ for a possibly different constant $c>0$. Now multiplying the Bergman kernel with $e^{-\frac12 n Q_\tau(z)-\frac12 n Q_\tau(w)}$ and using that $Q_\tau\geq 0$ we infer that
    \begin{align*}
        |\boldsymbol k_n^{h}(z,w) - P_n^{h,\#}[\zeta\mapsto \boldsymbol k_n^h(z, \zeta)](w) | \leq C_h e^{\frac12 n Q_\tau(z)+\frac12 n Q_\tau(w)}
    \end{align*}
    for some constant $C_h>0$ that depends only on $h$. Using the result of 
    Lemma \ref{lem:operatorDifference} and the reproducible property of $\boldsymbol k_n^h$ we infer that
    \begin{align*}
        |\boldsymbol k_n^{h}(z,w)-\boldsymbol k_n^{h,\#}(z,w)| &\leq |\boldsymbol k_n^{h}(z,w) - P_n^{h,\#}[\zeta\mapsto \boldsymbol k_n^h(z, \zeta)](w) |
        + \left|P_n^h[\zeta\mapsto \boldsymbol k_n^{h,\#}(\zeta, w)](z)
    - \overline{P_n^{h,\#}[\zeta\mapsto \boldsymbol k_n^h(z,\zeta)](w)}\right|\\
    &\leq \tilde C_h e^{\frac12 n Q_\tau(z)+\frac12 n Q_\tau(w)}
    \end{align*}
    as $n\to\infty$, for some constant $\tilde C_h>0$ that depends only on $h$. Reinstating the weight factors to obtain the weighted Bergman kernel, we obtain the result. 
\end{proof}


\subsection{Estimating the error terms}

Using Lemma \ref{lem:thm3.2} we infer that the Berezin kernel rooted at $w\in \mathring E_\tau$ satisfies
\begin{align} \label{eq:behavBerezin}
    B_{n, h}^{(w)}(z) := \frac{|K_n^h(z, w)|^2}{K_n^h(w,w)}
    = \frac{(1-\tau)n}{\pi} \frac{1}{(1+\tau)} e^{- (1-\tau)n \frac{|z-w|^2}{2(1+\tau)}}+ \mathcal O(1)
\end{align}
as $n\to\infty$, where the implied constant is independent of $n$ and $\tau$, and it is assumed that $\operatorname{dist}(w, \partial E_\tau)>\delta_n^\alpha$. We may now follow exactly the same arguments as in Section 3.2 in \cite{AmHeMa}, with $n$ replaced by $n^{1-\alpha}$ (and $\delta_n$ replaced by $\delta_n^\alpha$), to conclude that $\epsilon_{n,1}^h\to 0$ as $n\to\infty$. Namely, one writes
\begin{align*}
    \epsilon_{n,1}^h[v] = \int_{\mathbb C} \frac1n K_n^h(w,w) F_n^h[v](w) d^2w,
\end{align*}
where
\begin{align*}
    F_n^h[v](w) = \int_{\mathbb C} \left(\frac{v(z)-v(w)}{z-w}-\partial v(w)\right) B_{n,h}^{(w)}(z) d^2z.
\end{align*}
One can show that the region $\operatorname{dist}(w,\partial E_\tau)<\delta_n^\alpha$ is negligible, and for the remaining region one may use the behavior of the Berezin kernel in \eqref{eq:behavBerezin}. 

\begin{lemma}
    We have for any bounded and Lipschitz continuous $C^2$ function $v:\mathbb C\to \mathbb C$ that
    \begin{align*}
        \lim_{n\to\infty} \epsilon_{n,1}^h[v] = 0.
    \end{align*}
\end{lemma}

\begin{proof}
    The proof is essentially a copy and paste of the proof of Proposition 3.5 with $n$ replaced by $n^{1-\alpha}$ in \cite{AmHeMa}. However, there is one main difference, which is the contribution of the region  $\operatorname{dist}(w,\partial E_\tau)<\delta_n^\alpha$. In our case the size of $E_\tau$ is growing with $n$. It will suffice to show that 
    \begin{align*}
        \int_{\operatorname{dist}(w,\partial E_\tau)>\delta_n^\alpha} \frac1n K_n^h(w,w) d^2w = 1 + o(1)
    \end{align*}
    as $n\to\infty$. For $w\in \mathring E_\tau$ with $\operatorname{dist}(w,\partial E_\tau)>\delta_n^\alpha$, we have by Lemma \ref{lem:thm3.2} that
    \begin{align*}
        \frac1n K_n(w,w) = \frac1{\pi} \frac{1-\tau}{1+\tau} + \mathcal O(1/n).
    \end{align*}
    Hence the contribution of this region to the integral is
    \begin{align*}
        (|E_\tau|+\mathcal O(n^\frac{\alpha}{2}\log n))\left(\frac1{\pi} \frac{1-\tau}{1+\tau} + \mathcal O(1/n)\right) = 1+ \mathcal O(n^{-1+\alpha})+\mathcal O(n^{-\frac{\alpha}{2}}\log n).
    \end{align*}
    For this we used the formula of the area of an ellipse with axes that are of order $\delta_n^\alpha$ smaller.
    One may use the global estimate \eqref{eq:lem:thm3.1} to show that the contribution of the region $w\in \mathbb C\setminus E_\tau$ with $\operatorname{dist}(w,\partial E_\tau)>\delta_n^\alpha$ is negligible. We conclude that
    \begin{align*}
        \int_{\operatorname{dist}(w,\partial E_\tau)<\delta_n^\alpha} \frac1n K_n^h(w,w) d^2w =\mathcal O(n^{-1+\alpha})+\mathcal O(n^{-\frac{\alpha}{2}}\log n)
    \end{align*}
    as $n\to\infty$. The rest of the proof is the same as in \cite{AmHeMa}. 
\end{proof}

Now we turn to the second error term 
\begin{align*}
    \epsilon_{n,2}^h[v] = \frac1{2\pi n}\int_{\mathbb C} \bar\partial v(z) \big(D_n^h(z)\big)^2 d^2z.
\end{align*}

By Lemma \ref{lem:thm3.1} and Lemma \ref{lem:Knh-Knzz} we have
\begin{align} \label{eq:DnhDnDiffIneq}
    |D_n^h(z)-D_n(z)| \leq C_h \sqrt{(1-\tau) n} \log n \int_{d(w, \partial E_\tau)<\delta_n^\alpha} \frac{d^2\zeta}{|\zeta-z|}+\mathcal O(1) 
    = \mathcal O(\log^2 n)
\end{align}
as $n\to\infty$. Therefore, it suffices to limit ourselves to the case $h=0$.

\begin{lemma} \label{lem:Dn(z)}
    We have that
    \begin{align*}
        D_n(z) = \frac{\mathcal O(1)}{((1-\tau)^2 z^2-4\tau)^\frac{1}{4}\sqrt{z-\coth(2\xi_\tau) \sqrt{(1-\tau)^2 z^2-4\tau}}}
    \end{align*}
    as $n\to\infty$, uniformly for a.e. $z\in \mathbb C$. 
\end{lemma}

\begin{proof}
    For every fixed $n$ we have a.e. that
    \begin{align*}
        \frac{1}{\pi(1-\tau^2)} \mathfrak{1}_{\xi<\xi_\tau}(z) = -\frac{\omega\left(\sqrt n z\right)}{\pi\sqrt{1-\tau^2}}   \frac{1}{2\pi i} \oint_{\gamma_0} \frac{e^{n F(\tau;z,z;s)}}{s-\tau} \frac{ds}{\sqrt{1-\tau^2}}.
    \end{align*}
    Thus, denoting $z=2\sqrt\tau \cosh(\xi+i\eta)$ and $\zeta=2\sqrt\tau \cosh(\tilde\xi+i\tilde\eta)$, we have
    \begin{align*}
        D_n(z/(1-\tau)) &= -\sqrt\frac{1-\tau}{1+\tau}\frac{n}{2\pi i}\int_{\mathbb C} \frac{\omega\left(\sqrt n \zeta\right)}{\pi}  \oint_{\gamma_0} \frac{e^{n F(\tau;\zeta,\zeta;s)}}{s-\tau} \left(\frac1{\sqrt{1-s^2}}-\frac1{\sqrt{1-\tau^2}}\right) ds \, \frac{d^2z}{z-\zeta}\\
        &= \frac{1}{2} \frac{\sqrt\pi \sqrt n}{(1+\tau)^\frac{3}{2} (1-\tau)} 
        \int_{\mathbb C} \frac{1}{|\sinh(\tilde\xi+i\tilde\tau)|} e^{- n (\tilde\xi-\xi_\tau)^2 g(\tilde\xi+i\tilde\eta)} (1+\mathcal O(1/n)) \frac{d^2z}{z - \zeta},
    \end{align*}
    Changing variables to elliptic coordinates $\zeta=2\sqrt\tau\cosh(\tilde\xi+i\tilde\eta)$ we get
    \begin{align*}
        D_n(z) &= 2\tau \frac{(1-\tau)\sqrt\pi\sqrt n}{(1+\tau)^\frac{3}{2}} 
        \int_{-\pi}^\pi \int_0^\infty  \sinh(\tilde\xi+i\tilde\eta) e^{- n (\tilde\xi-\xi_\tau)^2 g(\tilde\xi+i\tilde\eta)} (1+\mathcal O(1/n)) \frac{d\tilde\xi d\tilde\eta}{z - \zeta}\\
        &= \frac{\sqrt{\pi}}{2} \frac{\sqrt{1-\tau}}{1+\tau} \int_{-\pi}^\pi \sqrt\frac{\sinh(\xi_\tau+i\tilde\eta)}{\sinh(\xi_\tau-i\tilde\eta)}(1+\mathcal O(1/n)) \frac{d\tilde\eta}{\cosh(\xi+i\eta) - \cosh(\xi_\tau+i\tilde\eta)}.
    \end{align*}
    When $\xi<\xi_\tau$ we have by complex contour integration over $\tilde\zeta=e^{i\tilde\eta}$ that a.e.
    \begin{align*}
        \int_{-\pi}^\pi \sqrt\frac{\sinh(\xi_\tau+i\tilde\eta)}{\sinh(\xi_\tau-i\tilde\eta)}\frac{d\tilde\eta}{\cosh(\xi+i\eta) - \cosh(\xi_\tau+i\tilde\eta)}
        &= \sqrt\frac{\sinh(\xi+i\eta)}{\sinh(2\xi_\tau-\xi-i\eta)} \frac{2\pi}{\sqrt{z^2-4\tau}}\\
        &= 2\pi\frac{1}{\sqrt{\sinh(2\xi_\tau) z - \cosh(2\xi_\tau) \sqrt{z^2-4\tau}}} \frac1{(z^2-4\tau)^\frac14}.
    \end{align*}
    Actually, when $\xi=\xi_\tau$ the same equality but with a factor $\frac{1}{2}$ holds, since the residue contributes half then. The case $\xi>\xi_\tau$ is similar, here one has to use the other pole. This yields
    \begin{align*}
        \int_{-\pi}^\pi \sqrt\frac{\sinh(\xi_\tau+i\tilde\eta)}{\sinh(\xi_\tau-i\tilde\eta)}\frac{d\tilde\eta}{\cosh(\xi+i\eta) - \cosh(\xi_\tau+i\tilde\eta)}
        = \sqrt\frac{\sinh(2\xi_\tau-\xi-i\eta)}{\sinh(\xi+i\eta)} \frac{2\pi}{\sqrt{\overline z^2-4\tau}}.
    \end{align*}
\end{proof}

\begin{corollary}
    We have for any bounded and Lipschitz continuous $C^2$ function $v:\mathbb C\to \mathbb C$ that
    \begin{align*}
        \lim_{n\to\infty} \epsilon_{n,2}^h[v] = 0.
    \end{align*}
\end{corollary}

\begin{proof}
    By \eqref{eq:DnhDnDiffIneq} and Lemma \ref{lem:Dn(z)} we have as $n\to\infty$ that
    \begin{align*}
        \epsilon_{n,2}^h[v] = \mathcal O\big(\frac{\log^4 n}{n^{1+\alpha}}+\frac1{n^{1-\alpha}}\big).
    \end{align*}
\end{proof}

\subsection{Ward identity limit}

Using \eqref{eq:limitWardId} we conclude that
\begin{align} \label{eq:limitWardId2}
    -\lim_{n\to\infty} \pi X_n^h(f) = \lim_{n\to\infty}   \int_{\mathbb C} (\frac{1}{2} \partial v(z)+2 v(z) \partial h(z)) \frac{1}{n} K_n^h(z) d^2z,
\end{align}
provided that this limit exists. Notice that
\begin{align*}
    v = \frac{1}{2}\frac{1+\tau}{1-\tau} \bar\partial f(z) \mathfrak{1}_{E_\tau}+ \frac{f_\tau}{\partial(Q_\tau-\check{Q}_\tau)} \mathfrak{1}_{\mathbb C\setminus E_\tau}.
\end{align*}
For the first term we obviously obtain a well-defined limit in the integral in \eqref{eq:limitWardId2} as $n\to\infty$ (use Lemma \ref{lem:thm3.2} and the factor $\frac{1+\tau}{1-\tau}$). The second term is more complicated because $f_\tau$ depends on $\tau$ and thus on $n$. However, from the explicit formula \eqref{eq:amtauCoeffs} for the coefficients $a_m(\tau)$ it is clear that
\begin{align*}
    f_\tau(z) = f(z) - f_{\tau_+}(z) - f_{\tau-}(z)
\end{align*}
is of order $1$. Dividing this by $\partial (Q_\tau-\check{Q}_\tau)$ for $z\in \mathbb C\setminus E_\tau$ gives something of order $n^\alpha/|z|$ (use the proof of Lemma 4.2 in \cite{AmHeMa} for this near $\partial E_\tau$). 
We conclude that $X_n^h(f)$ has a well-defined limit as $n\to\infty$, using Lemma \ref{lem:thm3.1} and dominated convergence ($\frac{1}{n} K_n^h$ works as a delta function for the boundary, although in the end we only need to know that the limit exists), which is furthermore linear in $h$. As explained in the proof of Lemma 1.4 in \cite{AmHeMa} (and in Johansson \cite{Johansson}) this implies that we have a CLT. 
For the sake of completeness, we repeat this argument here. 

\begin{proof}[Proof of Theorem \ref{thm:GFF}]
Let us write $\mathcal X_n(f) = X_n(f) - \mathbb E X_n(f)$ and $a_n^h(t)=\mathbb E^{t h} \mathcal X_n(f)$, where $t\geq 0$ and $\mathbb E^{t h}$ denotes the expectation value associated to the potential $Q_\tau^{t h}$. We have proved that
\begin{align*}
    \lim_{n\to\infty} X_n^h(f) = L(f,h)
\end{align*}
for some limiting function $L$ that is linear in both its variables, apart from a constant term. This implies for any fixed $t$ that
\begin{align*}
    \lim_{n\to\infty} a_n^h(t) = L(f, t h) - \lim_{n\to\infty} \mathbb E X_n(f)
    = L(f,0)+ t L(f, h) - \lim_{n\to\infty} \mathbb E X_n(f)
    = t L(f, h).
\end{align*}
Let us take $h=f$ in what follows. Differentiation of the cumulant generating function $C_n^h(t) = \log \mathbb E(e^{t \mathcal X_n(h)})$ yields
\begin{align*}
    (C_n^h)'(t) = \frac{\mathbb E (\mathcal X_n(h) e^{t X_n(h)})}{\mathbb E(e^{t X_n(h)})}
    = \mathbb E^{t h}(\mathcal X_n(h)),
\end{align*}
where the last step follows by viewing the expectation value as an integration over $\mathbb C^n$ and absorbing the factors $e^{t h}$ into the weight factors associated to $Q_\tau^h$. We conclude that $(C_n^h)'(t)=a_n^h(t)$. Next we use the standard observation that a cumulant generating function is convex (or, equivalently, that a moment-generating function is log-convex), giving
\begin{align*}
    0 = a_n^h(0) = (C_n^h)'(0) \leq (C_n^h)'(s) \leq (C_n^h)'(t) = a_n^h(t)
\end{align*}
for any $0\leq s\leq t$. We thus have
\begin{align*}
    C_n^h(t) = \int_0^t (C_n^h)'(s) ds \leq t a_n^h(t).
\end{align*}
We may now apply the dominated convergence theorem to conclude that
\begin{align*}
    \lim_{n\to\infty} C_n^h(t) = \int_0^t s L(h,h) ds = \frac12 L(h,h) t^2.
\end{align*}
This implies that only the second order cumulant survives in the limit $n\to\infty$. We thus have a central limit theorem. 
\end{proof}

\vspace{1cm}

\noindent \textbf{\small{CONFLICT OF INTEREST STATEMENT.}}\\
The authors declare that they have no conflict of interest.\\




\begin{thebibliography}{1}


\bibitem{OUP} G. Akemann, J. Baik, and P. Di Francesco (Editors), 
The Oxford Handbook of Random Matrix Theory, 
Oxford University Press, Oxford,  2011. 

\bibitem{ABender} 
G. Akemann and  M. Bender,  \emph{"Interpolation between Airy and Poisson statistics for unitary chiral non-Hermitian random matrix ensembles."} J. Math. Phys. 51(10), 103524 (2010). 



\bibitem{AkByEb} G. Akemann, S.-S. Byun, and M. Ebke, \emph{``Universality of the number variance in rotational invariant two-dimensional Coulomb gases''.}
J. Stat. Phys. Volume 190, 9, (2023). 

\bibitem{ACV} G. Akemann, M. Cikovic, and M. Venker, \emph{"Universality at weak and strong non-Hermiticity beyond the elliptic Ginibre ensemble".} Comm. Math. Phys. 362(3), 1111–1141 (2018). 

\bibitem{ADM} G. Akemann, M. Duits, and L. D. Molag, \emph{``The Elliptic Ginibre Ensemble: A Unifying Approach to Local and Global Statistics for Higher Dimensions''.}
J. Math. Phys. 64, 023503 (2023). 

\bibitem{AmBy}
Y. Ameur and S.-S. Byun, \emph{"Almost-Hermitian random matrices and bandlimited point processes."} Analysis Math. Phys. 13(3), 52 (2023).

\bibitem{AmChCr} Y. Ameur, C. Charlier and J. Cronvall,
\emph{"The two-dimensional Coulomb gas: fluctuations through a spectral gap".}
arXiv:2210.13959 (2022).

\bibitem{AmChCrLe}
Y. Ameur, C. Charlier, J. Cronvall and J. Lenells, \emph{"Exponential moments for disk counting
statistics at the hard edge of random normal matrices".}, J. Spectr. Theory 13 (2023), no. 3,
841-902.

\bibitem{AmChCrLe2}
Y. Ameur, C. Charlier, J. Cronvall and J. Lenells, 
\emph{"Disk counting statistics near hard edges
of random normal matrices: the multi-component regime".} Adv. Math. 441 (2024), Paper
No. 109549.

\bibitem{AmCr2} Y. Ameur and J. Cronvall,
\emph{"On fluctuations of Coulomb systems and universality of the Heine distribution".}  	arXiv:2411.10288 (2024). 

\bibitem{AmCr} Y. Ameur and J. Cronvall, \emph{“Szeg\H{o} type asymptotics for the reproducing kernel in spaces of full-plane weighted polynomials”.} Commun. Math. Phys. 398, 1291- 1348 (2023)
(2022). 

\bibitem{AC} 
Y. Ameur and J. Cronvall, 
\emph{“On fluctuations of Coulomb systems and universality of the Heine distribution”.} arXiv:2411.10288 (2024).

\bibitem{AmTr} Y. Ameur and E. Troedsson,
\emph{``Remarks on the one-point density of Hele-Shaw $\beta$-ensembles".}
arXiv:2402.13882 (2024).

\bibitem{AmHeMa2}  Y. Ameur, H. Hedenmalm, and N. Makarov, \emph{``Fluctuations of eigenvalues of random normal matrices''.}
Duke Math. J. 159, 31-81 (2011). 

\bibitem{AmHeMa}
Y. Ameur, H. Hedenmalm, and N. Makarov, \emph{“Random normal matrices and Ward identities”.} Ann. Probab. 43, 1157–1201 (2015). 

\bibitem{AmHeMa3}
Y. Ameur, H. Hedenmalm, and N. Makarov, \emph{“Berezin transform in polynomial Bergman spaces”.} Commun. Pure Appl. Math. Vol. LXIII, 1533–1584 (2010). 

\bibitem{AmKaMa}
Y. Ameur, N.-G. Kang, and N. Makarov, \emph{``Rescaling Ward identities in the random normal
matrix model”}. Constr. Approx. 50 (2019), 63-127.


\bibitem{BBNY}
R. Bauerschmidt, P. Bourgade, M. Nikula, and H.-Tzer Yau, \emph{"The two-dimensional Coulomb plasma: quasi-free approximation
and central limit theorem".} Adv. Theor. Math. Phys. 23(4), 841–1002 (2019). 


\bibitem{BCH}
P. Bourgade, G. Cipolloni, and J. Huang, \emph{" Fluctuations for non-Hermitian dynamics".} arXiv:2409.02902 (2024).




\bibitem{Bender} M. Bender,
\emph{``Edge scaling limits for a family of non-Hermitian random matrix ensembles".}
Probab. Theory Relat. Fields 147, 241-271 (2010). 

\bibitem{BerggrenDuits}
T. Berggren and M. Duits, \emph{" Mesoscopic fluctuations for the thinned Circular Unitary Ensemble"}, Math. Phys. Anal. Geom, vol. 20, no. 3, 2017.

\bibitem{Berman} R. Berman,
\emph{``Bergman kernels and weighted equilibrium measures of $\mathbb C^n$".}
Indiana Univ. J. Math. 58(4), 1921-1946 (2009). 

\bibitem{ByunPlanar}
S.-S. Byun, \emph{``Planar equilibrium measure problem in the quadratic fields with a point charge".} Comput. Methods Funct. Theory Vol. 24, 303–332, (2024).

\bibitem{BS}
S.-S. Byun and S. M. Seo,  \emph{"Random normal matrices in the almost-circular regime".} Bernoulli 29(2), 1615-1637 (2023). 

\bibitem{CDM}
P. Calabrese, P. Le Doussal, and S. N. Majumdar, \emph{"Random matrices and entanglement entropy of trapped Fermi gases."}
Phys. Rev. A 91(1), 012303 (2015). 


\bibitem{CMV} P. Calabrese, M. Mintchev and E. Vicari,
\emph{``Exact relations between particle fluctuations and entanglement in Fermi gases''.}
 Eur. Phys. Lett. 98, 20003 (2012). 

\bibitem{Charlier1}
C. Charlier, \emph{"Asymptotics of determinants with a rotation-invariant weight and discontinuities
along circles".} Adv. Math. 408 (2022), Paper No. 108600.

\bibitem{ChLe}
C. Charlier and J. Lenells, \emph{"Exponential moments for disk counting statistics of random
normal matrices in the critical regime".} Nonlinearity 36 (2023), no. 3, 1593–1616.

\bibitem{CES} 
G. Cipolloni, L. Erdős, and D. Schröder, \emph{"Mesoscopic central limit theorem for non-Hermitian random matrices".} Prob. Theory Rel. Fields, 188(3), 1131-1182 (2024). 

\bibitem{CES2}
G. Cipolloni, L. Erdős, and D. Schröder, \emph{" Central Limit Theorem for Linear Eigenvalue Statistics of Non‐Hermitian Random Matrices}. Commun. Pure Appl. Math. 76(5), 946-1034 (2023). 

\bibitem{ChEs}
L. Charles and B. Estienne, \emph{``Entanglement entropy and Berezin-Toeplitz operators".}
Commun. Math. Phys. 376 (1), 521-554.



\bibitem{CoLe} O. Costin and J. Lebowitz, \emph{"Gaussian fluctuations in random matrices".}
Phys. Rev. Lett. 75(1), 69–72 (1995). 

\bibitem{CrFyWu} M. J. Crumpton, Y. V. Fyodorov and T. R. Würfel,
\emph{``Mean eigenvector self-overlap in the real and complex elliptic Ginibre ensembles at strong and weak non-Hermiticity".}
arXiv:2402.09296v2 (2024).

\bibitem{DDMSrev}
D. S. Dean, P. Le Doussal, S. N. Majumdar, and G. Schehr, 
\emph{"Noninteracting fermions in a trap and random
matrix theory".} 
J. Phys. A 52(14), 144006 (2019). 



\bibitem{DeLa} A. Deleporte, and G. Lambert, 
\emph{``Universality for free fermions and the local Weyl law for semiclassical Schrödinger operators".} arXiv:2109.02121 (2021).

\bibitem{DeleporteLambert} A. Deleporte and  G. Lambert, \emph{"Widom’s conjecture: variance asymptotics and entropy bounds for counting statistics of free fermions".} arXiv:2405.07796 (2024).

\bibitem{Demni} Nizar Demni, Zouha\"ir Mouayn, \emph{"Polyanalytic Hermite polynomials associated with the elliptic Ginibre model".} arXiv:2501.17291 (2025).


\bibitem{DuitsJohansson} M. Duits and K. Johansson, \emph{''On mesoscopic equilibrium for linear statistics in Dyson's Brownian Motion''}, Mem. Amer. Math. Soc. 255 (2018), no. 1222, v+118 pp.

\bibitem{Eisler}
V. Eisler, \emph{"Universality in the full counting statistics of trapped fermions".} Phys. Rev. Lett. 111(8), 080402 (2013). 

\bibitem{FeLa}
M. Fenzl and G. Lambert, \emph{"Precise deviations for disk counting statistics of invariant determinantal
processes".} Int. Math. Res. Not. IMRN 2022 (2022), 7420-7494.

\bibitem{Peter}
P. J. Forrester, \emph{"Fluctuation formula for complex random matrices".} J. Phys. A 32(13), L159–L163 (1999).  

\bibitem{PeterLog}
P. J. Forrester, \emph{"Log-gases and random matrices".} Princeton University Press, Princeton, NJ, 2010.

\bibitem{PeterVar}
P. J. Forrester, \emph{"A review of exact results for fluctuation formulas in random matrix theory."} Prob. Surveys 20, 170-225 (2023).

\bibitem{FH99} P. J. Forrester and G. Honner, \emph{“Exact statistical properties of the zeros of complex random polynomials.”} J. Phys. A: Math. Gen. 32, 2961-2981 (1999).

\bibitem{FoJa}
P. J. Forrester and B. Jancovici, \emph{“Two-dimensional one-component plasma in a quadrupolar field.”} Int. J. Mod. Phys. A 11, 941-949 (1996). 

\bibitem{FKS}
Y. V. Fyodorov, B. A. Khoruzhenko, and H.-J. Sommers, \emph{"Almost Hermitian random matrices: crossover from
Wigner-Dyson to Ginibre eigenvalue statistics."} Phys. Rev. Lett. 79(4), 557 (1997). 


\bibitem{FKS2}\
Y. V. Fyodorov, B. A. Khoruzhenko, and H.-J. Sommers, \emph{"Almost-Hermitian random matrices: eigenvalue
density in the complex plane".} Phys. Lett. A 226(1), 46–52 (1997). 


\bibitem{FyKhSo} Y. V. Fyodorov, B. A. Khoruzhenko, and H.-J. Sommers, \emph{``Universality in the random matrix spectra in the regime of weak non-Hermiticity''.}
Ann. Inst. H. Poincar\'e Phys. Th\'eor. 68, 449 (1998). 

\bibitem{Ginibre} J. Ginibre, \emph{“Statistical ensembles of complex, quaternion, and real matrices.”} J. Math. Phys. 6, 440-449 (1965).

\bibitem{Girko}
L. Girko, \emph{“Elliptic law.”} Theory Probab. Appl. 30, 677-690 (1986).

\bibitem{HeMa} H. Hedenmalm and N. Makarov, \emph{``Coulomb gas ensembles and Laplacian growth".}
Proceedings of the London Mathematical Society 106, 859-907 (2013). 

\bibitem{HeWe}
H. Hedenmalm, and A. Wennman, \emph{"Planar orthogonal polynomials and boundary universality in the random normal matrix model."} Acta Math. 227, 309-406
(2021). 


\bibitem{Johansson} 
K. Johansson, 
\emph{``On fluctuations of eigenvalues of random Hermitian matrices".}
Duke Math. J. 91, 151-204 (1998).

\bibitem{JohanssonLambert} K. Johansson and G. Lambert, \emph{"Gaussian and non-Gaussian fluctuations for mesoscopic linear statistics in determinantal processes"}, Ann. Probab. 46(3): 1201-1278 (May 2018).

\bibitem{KrYo}
M. Krishnapur and D. Yogeshwaran, \textit{``Stationary random measures : Covariance asymptotics,
variance bounds and central limit theorems."} arXiv preprint. arXiv:2411.08848

\bibitem{KMS}
M. Kulkarni, S. N. Majumdar, and G. Schehr. \emph{"Multilayered density profile for noninteracting fermions in a rotating
two-dimensional trap."} Phys. Rev. A 103(3), 033321 (2021). 

\bibitem{KDMS}
M. Kulkarni, P. Le Doussal, S. N. Majumdar, and G. Schehr. \emph{"Density profile of noninteracting fermions in a rotating
two-dimensional trap at finite temperature".} Phys. Rev. A 107(2), 023302 (2023). 



\bibitem{LMS} B. Lacroix-A-Chez-Toine, S. N. Majumdar, and G. Schehr, \emph{``Rotating trapped fermions in two dimensions and the complex Ginibre ensemble:
Exact results for the entanglement entropy and number variance''.}
Phys. Rev. A 99, 021602(R) (2019).

\bibitem{Lambert} G. Lambert, \emph{"Incomplete determinantal processes: from random matrix to Poisson statistics"},  J Stat Phys 176, 1343–1374 (2019).







\bibitem{LS}
T.  Lebl\'e and S. Serfaty, \emph{"Fluctuations of two-dimensional Coulomb gases."} Geom. Funct. Anal. 28(2), 443–508 (2018).


\bibitem{LeRi} S.-Y. Lee and R. Riser, \emph{“Fine asymptotic behavior for eigenvalues of random normal matrices: Ellipse case.”} J. Math. Phys. 57, 023302 (2016). 

\bibitem{LeMaOC}
M. Levi, J. Marzo and J. Ortega-Cerdà, \emph{"Linear statistics of determinantal point processes and norm representations."} Int. Mat. Res. Not. 2024(19), 12869-12903 (2024). 

\bibitem{Lin}
Z. Lin, \emph{"Nonlocal energy functionals and determinantal point processes on non-smooth
domains."} Zeitschr. Math. 307, 56 (2024). 

\bibitem{Macchi}
O. Macchi, \emph{"The coincidence approach to stochastic point processes".}
Adv. Appl. Prob. 7, 83-122 (1975).

\bibitem{MaMoOC}
J. Marzo, L. D. Molag, and J. Ortega-Cerdà. \emph{``Universality for fluctuations of counting statistics of random normal matrices".} J. Lond. Math. Soc. (2), 113(2), 2026.

\bibitem{MMSV}
R. Marino, S.N. Majumdar, G. Schehr, and P. Vivo, \emph{"Phase transitions and edge scaling of number variance in Gaussian random matrices." } Phys. Rev. Lett. 112(25), 254101 (2014). 

\bibitem{Marceca}
F. Marceca and J.-L. Romero, \emph{``Improved discrepancy for the planar Coulomb gas at low temperatures."}
arxiv: 2212.14821 (to appear in J. Anal. Math.)

\bibitem{Mo} L. D. Molag, \emph{``Edge behavior of higher complex-dimensional determinantal point processes''.}
Ann. Henri Poincar\'e 24,  4405–4437 (2023). 

\bibitem{NIST} F. W. J. Olver, D. W. Lozier, R. F. Boisvert, and C. W. Clark (eds.). NIST Handbook of MathematicalFunctions. Cambridge University Press, Cambridge (2010).

\bibitem{pompeiu}
D. Pompeiu, 
\emph{``Sur la continuit\'e des fonctions de variables complexes".} Annales de la Faculté des sciences de l'Université de Toulouse pour les sciences mathématiques et les sciences physiques, Serie 2, Volume 7(3), 265-315 (1905).


\bibitem{RiVi} B. Rider and B. Virág, \emph{``The noise in the circular law and the Gaussian free field''.}
Int. Math. Res. Not. 2007, rnm006 (2007). 

\bibitem{Shirai}
T. Shirai, \emph{"Ginibre-type point processes and their asymptotic behavior".} J. Math. Soc. Japan 67(2), 763-787 (2015).


\bibitem{SDMS}
N. R. Smith, P. Le Doussal, S. N. Majumdar, and G. Schehr, \emph{"Counting statistics for noninteracting fermions in a
rotating trap".} Phys. Rev. A, 105(4), 043315 (2022). 

\bibitem{SoWeYa}
M. Sodin, A. Wennman and O. Yakir. \textit{``The Random Weierstrass Zeta Function II. Fluctuations of the Electric Flux Through Rectifiable Curves."} J. Stat. Phys. 190, 164 (2023).

\bibitem{SCSS} 
H. J. Sommers, A. Crisanti, H. Sompolinsky, and Y. Stein, \emph{"Spectrum of Large Random Asymmetric Matrices".}
Phys. Rev. Lett. 60, 1895-1899 (1988).

\bibitem{Soshnikov}
A. Soshnikov, \emph{"Gaussian limit for determinantal random point fields".} Ann. Probab. 30,
171-187 (2002)

\bibitem{ToSt}
S. Torquato and F. H. Stillinger, \emph{``Local density fluctuations, hyperuniformity, and
order metric".}
Phys. Rev. E (3) 68 (2003), pp. 041113, 25.

\bibitem{Vicari}
E. Vicari, \emph{"Entanglement and particle correlations of
Fermi gases in harmonic traps".} Phys. Rev. A 85, 062104
(2012). 

\bibitem{WiZa} A. Zabrodin and P. Wiegmann, 
\emph{"Large $N$ expansion for the 2D Dyson gas".}
J. Phys. A 39, 8933-8964 (2006). 

\bibitem{Za} A. Zabrodin, 
\emph{"Matrix models and growth processes: from viscous flows to the quantum Hall effect."} 
In: Brézin, É., Kazakov, V., Serban, D., Wiegmann, P., Zabrodin, A. (eds) Applications of Random Matrices in Physics. NATO Science Series II: Mathematics, Physics and Chemistry, vol 221. Springer, Dordrecht (2006).


\end{thebibliography}
\end{document}